\documentclass[a4paper,UKenglish,11pt]{article}
\usepackage{times}

\usepackage[utf8]{inputenc}
\usepackage[T1]{fontenc}

\usepackage{relsize,xspace}
 \usepackage{xcolor}
 \usepackage{mathtools}
\usepackage{microtype}
\usepackage{amsmath}
 \usepackage{amssymb}
 \usepackage{amsfonts}
\usepackage{hyperref}
\usepackage{tikz}
\usepackage{bm}
\usepackage{enumerate,comment,environ,etoolbox}

 \usepackage{tikz,wrapfig}
\usetikzlibrary{shapes,decorations}

\usetikzlibrary{decorations.markings}
\newdimen\arrowsize
\newlength{\arrowlength}
\newlength{\arrowangle}
\newlength{\arrowthickness}

\providebool{showproofs}
\pgfarrowsdeclare{slim}{slim}
{
\arrowsize=0.7pt
\advance\arrowsize by .5\pgflinewidth
\pgfarrowsleftextend{-4\arrowsize-.5\pgflinewidth}
\pgfarrowsrightextend{.5\pgflinewidth}
}
{
\advance\arrowsize by .5\pgflinewidth
\pgfsetdash{}{0pt}
\pgfsetroundjoin
\pgfsetroundcap
\pgfpathmoveto{\pgfpointpolar{\arrowangle}{\arrowlength}}
\pgfpathlineto{\pgfpointorigin}
\pgfusepathqstroke
\pgfpathmoveto{\pgfpointpolar{\arrowangle}{\arrowlength}}
\pgfpathcurveto{\pgfpointpolar{\arrowangle*1.04}{\arrowthickness+0.5*(\arrowlength-\arrowthickness)}}%
{\pgfpointpolar{\arrowangle*1.1}{\arrowthickness+0.3*(\arrowlength-\arrowthickness)}}
{\pgfpointpolar{180}{\arrowthickness}}
\pgfpathlineto{\pgfpointorigin}
\pgfusepathqfill
\pgfpathmoveto{\pgfpointpolar{-\arrowangle}{\arrowlength}}
\pgfpathlineto{\pgfpointorigin}
\pgfusepathqstroke
\pgfpathmoveto{\pgfpointpolar{-\arrowangle}{\arrowlength}}
\pgfpathcurveto{\pgfpointpolar{-\arrowangle*1.04}{\arrowthickness+0.5*(\arrowlength-\arrowthickness)}}%
{\pgfpointpolar{-\arrowangle*1.1}{\arrowthickness+0.3*(\arrowlength-\arrowthickness)}}
{\pgfpointpolar{180}{\arrowthickness}}
\pgfpathlineto{\pgfpointorigin}
\pgfusepathqfill
}

\tikzstyle{vertex}=[circle,inner sep=1.5, outer sep=2, minimum size =5pt,semithick,fill=black!10, draw=black]
\tikzstyle{point}=[circle,inner sep=1,fill=black, draw=black]
 \tikzstyle{path2}=[-stealth,thin,decorate,%
 decoration={snake,amplitude=6pt,segment length=#1,%
 pre length=#1/10,post length=#1/10}]
\tikzstyle{brace}=[thin,decorate,decoration=brace]
\tikzstyle{ie}=[thin,dashed,gray]
  \tikzset{sepin/.style={postaction={decorate},decoration={markings,
    mark=between positions 0.1 and 1 step 0.5cm
    with
    {
      \draw[->,>=Stealth,very thin, black] (0,-3pt) to (0,3pt);
    }}}}
  \tikzset{sepout/.style={postaction={decorate},decoration={markings,
    mark=between positions 0.1 and 1 step 0.5cm
    with
    {
      \draw[->,>=Stealth,very thin, black] (0,3pt) to (0,-3pt);
    }}}}

\usepackage{tikz}
\usetikzlibrary{arrows}
\usetikzlibrary{shapes.geometric}
\usetikzlibrary{calc}
\usetikzlibrary{automata,positioning,arrows.meta}
\usetikzlibrary{decorations.markings}
\usetikzlibrary{decorations.pathmorphing}
\usetikzlibrary{matrix}
\usetikzlibrary{shadows}
\usetikzlibrary{backgrounds}
\usetikzlibrary{calendar,graphs}
\usetikzlibrary{calendar}

\definecolor{gruen}{rgb}{0,0.8,0.2}
\definecolor{rot}{rgb}{0.7,0,0}
\newcommand{\red}{\color{rot}}
\newcommand{\blue}{\color{blue}}

\newcommand{\black}{\color{black}}

\usepackage{empheq}
\everymath{\color{blue}}
\everydisplay{\color{blue}}

\tikzstyle{vertex}=[circle,inner sep=1,minimum size =2mm,semithick,fill=orange!10, draw=black]
\tikzstyle{point}=[circle,inner sep=1,fill=black, draw=black]
\tikzstyle{brace}=[thin,decorate,decoration=brace]
\tikzstyle{ie}=[thin,dashed,gray]
\tikzstyle{edge}=[draw=black]
\tikzstyle{arc}=[draw=black, >=Stealth, ->]

\tikzset{remarknode/.style={anchor=north west, rounded corners, rectangle, fill=blue!10, draw=black,
    thin}}
\tikzset{blocknode/.style={anchor=north west, rounded corners, rectangle, fill=orange!10, draw=black,
    thin}}
\tikzset{nodearc/.style={red,{[sep=2pt]}-{>[sep=2pt]},>=Stealth}}
\tikzset{explain/.style={draw=black, inner sep=2pt, fill=red!20, line width=0.1pt,font={\itshape\small}}}
\tikzset{skboxbase/.style={rectangle, rounded corners,
      draw=black, inner sep=2pt, align=center, line width=.1pt, anchor=north west}}
\tikzset{box/.style={skboxbase, fill=blue!10}}
\tikzset{obox/.style={skboxbase, fill=orange!10}}
\tikzset{nbox/.style={skboxbase, fill=blue!10}}
\tikzset{boxoptions/.style={}}

\newcommand{\tangcirc}[3][]{%
  \tikzstyle{tng}=[circle,fill=white,draw=black,inner sep = 1pt]
  \foreach \x in { (#2,#3-1), (#2-0.85,#3-0.5), (#2+0.85,#3-0.5),     (#2-0.85,#3+0.5), (#2+0.85,#3+0.5), (#2,#3+1) }
  \foreach \y in { (#2,#3-1), (#2-0.85,#3-0.5), (#2+0.85,#3-0.5),
    (#2-0.85,#3+0.5), (#2+0.85,#3+0.5), (#2,#3+1) }
    \draw[draw=gray!50] \x -> \y ;

  \node[tng] (Tt#1)  at  (#2,#3-1)         { 1 } ;
  \node[tng] (Ttl#1) at  (#2-0.85,#3-0.5)  { 2 } ;
  \node[tng] (Ttr#1) at  (#2+0.85,#3-0.5)  { 3 } ;
  \node[tng] (Tbl#1) at  (#2-0.85,#3+0.5)  { 4 } ;
  \node[tng] (Tbr#1) at  (#2+0.85,#3+0.5)  { 5 } ;
  \node[tng] (Tb#1)  at  (#2,#3+1)         { 6 } ;

  \node[fill=white, inner sep=0pt] (Tc#1) at (#2,#3) { #1 } ;
}

\colorlet{fillA}{gray!50}
\colorlet{fillB}{gray!15}

\usepackage{refcount}
\usepackage{wrapfig}
\usepackage{marginnote}
\usepackage[inline]{enumitem}
\usepackage{latexsym}
\usepackage{pdfpages}
\usepackage{todonotes}
\presetkeys{todonotes}{inline}{}

\definecolor{dark-blue}{rgb}{0.05,0.25,0.85}
\hypersetup{colorlinks=true,linkcolor=dark-blue,citecolor=dark-blue}

\makeatletter
\renewcommand\footnotesize{%
   \@setfontsize\footnotesize\@ixpt{11}%
   \abovedisplayskip 8\p@ \@plus2\p@ \@minus4\p@
   \abovedisplayshortskip \z@ \@plus\p@
   \belowdisplayshortskip 4\p@ \@plus2\p@ \@minus2\p@
   \def\@listi{\leftmargin\leftmargini
               \topsep 4\p@ \@plus2\p@ \@minus2\p@
               \parsep 2\p@ \@plus\p@ \@minus\p@
               \itemsep \parsep}%
   \belowdisplayskip \abovedisplayskip
}
\makeatother

\usepackage[amsmath,thmmarks,hyperref]{ntheorem}
\usepackage{cleveref}

\crefformat{page}{#2page~#1#3}%
\Crefformat{page}{#2Page~#1#3}%
\crefformat{equation}{#2(#1)#3}%
\Crefformat{equation}{#2(#1)#3}%
\crefformat{figure}{#2Figure~#1#3}%
\Crefformat{figure}{#2Figure~#1#3}%
\crefformat{section}{#2Section~#1#3}
\Crefformat{section}{#2Section~#1#3}
\crefformat{chapter}{#2Chapter~#1#3}
\Crefformat{chapter}{#2Chapter~#1#3}
\crefformat{chapter*}{#2Chapter~#1#3}
\Crefformat{chapter*}{#2Chapter~#1#3}
\crefformat{part}{#2Part~#1#3}
\Crefformat{part}{#2Part~#1#3}
\crefformat{enumi}{#2(#1)#3}
\Crefformat{enumi}{#2(#1)#3}

\theoremnumbering{arabic}
\theoremstyle{plain}
\theoremsymbol{}
\theorembodyfont{\itshape}
\theoremheaderfont{\normalfont\bfseries}
\theoremseparator{}

\newtheorem{theorem}{Theorem}[section]
\crefformat{theorem}{#2Theorem~#1#3}
\Crefformat{theorem}{#2Theorem~#1#3}

\newcommand{\newtheoremwithcrefformat}[2]{%
  \newtheorem{#1}[theorem]{#2}%
  \crefformat{#1}{##2\MakeUppercase#1~##1##3}%
  \Crefformat{#1}{##2\MakeUppercase#1~##1##3}%
}

\newtheoremwithcrefformat{proposition}{Proposition}
\newtheoremwithcrefformat{observation}{Observation}
\newtheoremwithcrefformat{lemma}{Lemma}
\newtheoremwithcrefformat{conjecture}{Conjecture}
\newtheoremwithcrefformat{corollary}{Corollary}
\theorembodyfont{\upshape}
\newtheoremwithcrefformat{example}{Example}
\newtheoremwithcrefformat{remark}{Remark}
\newtheoremwithcrefformat{definition}{Definition}
\newtheoremwithcrefformat{notation}{Notation}
\newtheoremwithcrefformat{claim}{Claim}

\theoremsymbol{\ensuremath{\square}}

\theoremstyle{nonumberplain}
\theoremheaderfont{\scshape}
\theorembodyfont{\normalfont}
\theoremsymbol{\ensuremath{\square}}
\newtheorem{proof}{Proof.}

\newcommand{\nprob}[4]{%
\begin{center}\normalfont\fbox{%
\begin{tabular}[t]{rp{#1}}%
\multicolumn{2}{l}{\textsc{#2}}\\%
\textit{Input:} & #3\\%
\textit{Problem:} & #4%
\end{tabular}}%
\end{center}}

\newcommand{\dtw}{\mathrm{dtw}}

\def\cqedsymbol{\ifmmode$\lrcorner$\else{\unskip\nobreak\hfil
\penalty50\hskip1em\null\nobreak\hfil$\lrcorner$
\parfillskip=0pt\finalhyphendemerits=0\endgraf}\fi}

\newcommand{\sth}{\mathrel{ : } }
\newcommand{\nin}{\not\in}

\usepackage[inline]{enumitem}

\newcommand{\BBB}{\mathcal{B}}
\newcommand{\CCC}{\mathcal{C}}
\newcommand{\DDD}{\mathcal{D}}

\newcommand{\LLL}{\mathcal{L}}

\newcommand{\OOO}{\mathcal{O}}

\newcommand{\SSS}{\mathcal{S}}
\newcommand{\TTT}{\mathcal{T}}

\newcommand{\WWW}{\mathcal{W}}

\newcommand{\BB}{\mathfrak{B}}

\newcommand{\LL}{\mathfrak{L}}

\newcommand{\bn}{\textup{bn}}

\newenvironment{cenv}{\begin{list}{}{%
      \setlength{\labelwidth}{1.5em}%
      \setlength{\leftmargin}{\labelwidth}%
      \addtolength{\leftmargin}{\labelsep}%
      \setlength{\listparindent}{0em}%
      \setlength{\topsep}{10pt}%
      \setlength{\itemsep}{5pt}%
      \setlength{\parsep}{0pt}%
    }
  }{
  \end{list}
}

\newcounter{claimcounter}
\newcounter{conditioncounter}

\newenvironment{ClaimProof}[1][]{\noindent{%
\ifthenelse{\equal{#1}{}}{{\itshape Proof.\ }}{{\itshape #1.\ }}%
}}{\hspace*{1em}\nobreak\hfill$\dashv$\endtrivlist\addvspace{2ex plus
0.5ex minus0.1ex}}

\RenewEnviron{proof}[1][]
{%
  \ifbool{showproofs}{%
    \setcounter{claimcounter}{0}
    \ifthenelse{\equal{#1}{}}{\noindent\textit{Proof. }}{\noindent\textit{#1. }}%
    \BODY  \hspace*{1pt}\hfill$\Box$\par\bigskip
  }%
}%

\renewenvironment{proof}[1][]
{\setcounter{claimcounter}{0}\ifthenelse{\equal{#1}{}}{\noindent\textit{Proof.
    }}{\noindent\textit{#1. }}}%
{\hspace*{1pt}\hfill$\Box$\par\bigskip}

\newcommand{\dsl}{{\ensuremath\leftarrow}}
\newcommand{\dsr}{{\ensuremath\rightarrow}}
\newcommand{\sep}{\operatorname{sep}}
\newcommand{\septop}{\operatorname{out}}
\newcommand{\sepbot}{\operatorname{in}}

\newcommand{\interior}[1]{\mathring{#1}}
\newcommand{\cone}{\textit{cone}}
\newcommand{\scone}{\textit{s-cone}}

\excludecomment{comment}
\excludecomment{alternative}
\includecomment{full}
\includecomment{short}

\booltrue{showproofs}  

\usepackage{todonotes}
\presetkeys{todonotes}{fancyline, color=red!30}{}

\newcommand\drop[1]{}
\usepackage{authblk,subfiles}
\usetikzlibrary{patterns,calc}
\newcommand{\sktgrid}[4]%
{
  \begin{scope}
    \draw[step=1mm,black!10] (#1,#2) grid (#3,#4);
    \draw[step=5mm,black!40] (#1,#2) grid (#3,#4);
    \draw[step=10mm,black!70] (#1,#2) grid (#3,#4);
    \draw[red] (#1,0) to (#3,0) ;
    \draw[red] (0,#2) to (0,#4) ;
  \end{scope}
}

\usepackage[margin=1in]{geometry}

\title{The canonical directed tree decomposition and its applications
  to the directed disjoint paths problem\thanks{The research of
    Archontia C. Giannopoulou, Stephan Kreutzer, and O-joung Kwon have
    been supported by the European Research Council (ERC) under the
    European Union's Horizon 2020 research and innovation programme
    (ERC consolidator grant DISTRUCT, agreement No.\ 648527. Ken-ichi Kawarabayashi was supported by JST ERATO Kawarabayashi Large Graph Project JPMJER1201 and by JSPS Kakenhi JP18H05291. O-joung Kwon has been supported by Institute for Basic Science (IBS-R029-C1) and the National Research Foundation of Korea (NRF) grant funded by the Ministry of Education (No. NRF-2018R1D1A1B07050294).}}

\author[1]{Archontia C. Giannopoulou \thanks{\texttt{archontia.giannopoulou@gmail.com} (A. Giannopoulou)}}
\author[2]{Ken-ichi Kawarabayashi \thanks{\texttt{keniti@nii.ac.jp} (K. Kawarabayashi)} }
\author[3]{Stephan Kreutzer \thanks{\texttt{stephan.kreutzer@tu-berlin.de} (S. Kreutzer)} }
\author[4,5]{O-joung Kwon \thanks{\texttt{ojoungkwon@gmail.com} (O. Kwon)} }

\affil[1]{Department of Informatics and Telecommunications, National and Kapodistrian University of Athens, Athens, Greece}
\affil[2]{National Institute of Informatics, 2-1-2 Hitotsubashi, Chiyoda-ku, Tokyo 101-8430, Japan}
\affil[3]{Logic and Semantics, TU Berlin, Berlin, Germany}
\affil[4]{Department of Mathematics, Incheon National University, Incheon, South Korea.}
\affil[5]{Discrete Mathematics Group, Institute~for~Basic~Science~(IBS), Daejeon,~South~Korea.}

\let\oldtext\text
\renewcommand{\text}[1]{\oldtext{\black{#1}}}

\begin{document}
\maketitle
\begin{abstract}
 The canonical tree-decomposition theorem, given by Robertson and
 Seymour in their seminal graph minors series, turns out to be one of
 the most important tool in structural and algorithmic graph theory. 
 In this paper, we provide the canonical tree decomposition theorem for digraphs.
 More precisely, we construct directed tree-decompositions of digraphs
 that distinguish all their tangles of order $k$, for any fixed
 integer $k$, in polynomial time. 
	
 As an application of this canonical tree-decomposition theorem, we provide
 the following result for the directed disjoint paths problem:
 For every fixed $k$ there is a
 polynomial-time algorithm which, on input $G$, and source and terminal vertices $(s_1,
 t_1), \dots, (s_k, t_k)$, either
 \begin{enumerate}
 \item
   determines that there is no set of pairwise vertex-disjoint paths
   connecting each source $s_i$  to its terminal $t_i$, or
 \item
   finds a half-integral solution, i.e., outputs paths $P_1, \dots, P_k$
   such that $P_i$ links $s_i$ to $t_i$, so that every vertex of the graph
   is contained in at most two paths.
 \end{enumerate}

 Given known hardness results for the directed
 disjoint paths problem, our result cannot be improved for
 general digraphs, neither to fixed-parameter tractability nor to
 fully vertex-disjoint directed paths.
 As far as we are aware, this is the first time to obtain a tractable
 result for the $k$-disjoint paths problem for general digraphs.
 We expect more applications of our canonical tree-decomposition for directed results.
\end{abstract}

\section{Introduction}

\subsection{Tangle and canonical tree-decomposition}

\emph{Tangles} in a graph $G$, introduced by Robertson and Seymour in the tenth paper \cite{GMX} of their
graph minors series \cite{GM-series}, are orientations of the low order separations that
consistently point towards some ``highly'' connected component of $G$.
As they were a
fundamental tool for their graph minors project, they play a crucially important role in
structural and algorithmic graph theory (for example in \cite{grohe,KW,kkr,GMXIII,GMXVI}).

Tangles are also connected to tree-decompositions and treewidth. Indeed the treewidth of a graph is large if and only if it admits a tangle of large order.
Robertson and Seymour \cite{GMX}
proved that every  graph has a tree-decomposition that distinguishes
every two maximal tangles. This is the so-called {\sl canonical tree-decomposition} (see \cite{Reed97}) and they used it to prove their famous
graph decomposition theorem.

Let us briefly illustrate how to use this canonical tree-decomposition in the graph minor case.
The graph minor decomposition theorem, roughly, says that every $H$-minor-free graph has a tree-decomposition such that
each {\sl torso} can be {\sl almost} embeddable in a surface (we do not give a precise definition of ``almost'' embeddability as it is too technical and out
of the scope of this paper).

In order to prove this decomposition theorem,
we would like to focus on the internal structure of one bag of the tree-decomposition. In order
to do so we can assume that the
tree is as refined as possible in the sense that no bag can be split into two smaller
bags. Then for every low order cutset, any component will, essentially, lie on one side or
the other of the cutset (except for very small components which we do not carry a lot of significance).
So if we fix one component that contains a highly connected component,
every small cutset has a ‘‘big’’ side
(containing most of the highly connected component) and a ‘‘small’’ side. Thus a component defines a {\sl tangle}, which is such an
assignment of big and small sides to the low order cutsets. Even the converse is true; indeed the canonical tree-decomposition theorem says that
any tangle of sufficiently high ‘‘order’’ will be associated with some
part of the tree-decomposition. So it remains to analyze the local structure with respect to some high order
tangle. A high order tangle also corresponds to a high connected component, which is usually described by a grid minor.
Indeed, one of the most fundamental theorems in this context is  \emph{the grid theorem}, proved by Robertson and Seymour in \cite{GMV} (see also \cite{ChekuriC14,rein,KawarabayashiK12a,ls,RST94}. It
states that there is a function $f(k)$ such
that every graph of treewidth (or, equivalently, tangle of order) at least $f(k)$ contains a $k\times k$-grid as a minor. Using this terminology,
the main part of the graph minor structure theorem says that if there is a grid minor of large order,
it must {\sl capture} an almost embeddable surface in some bag of the tree-decomposition (again, a precise definition of ``capture'' is out of the scope of this paper).

\subsection{From undirected graphs to directed graphs}

In this paper, we would like to extend this machinery to directed graphs. To this end, we have to define tangles, tree-decompositions, and highly connected components.

As a first step towards a structure theory specifically for directed
graphs, Reed \cite{Reed99} and Johnson, Robertson, Seymour and Thomas
\cite{JRST} proposed a concept of
\emph{directed treewidth} and \emph{brambles}, and conjectured a directed analog of the
grid theorem. The
conjecture had been open for nearly 20 years and, two of the authors (Kawarabayashi and Kreutzer) solved it quite recently \cite{stephan}.

The grid theorem for digraphs is a first but
important step
towards a more general structure theory for directed graphs based on
directed treewidth, similar to the grid theorem for
undirected graphs being the basis of the graph minor structure theorem.
Indeed, the grid seems to be the most natural ``highly connected directed component''.
Hence we already have definitions of  brambles, tree-decompositions, and highly connected components (i.e., the directed grid theorem). 
Moreover, we can define tangles in a natural way (to the best of our knowledge, this is the first time that tangles are defined for directed graphs). 

So the next target, which is indeed our main result, is to find a canonical directed tree-decomposition, similar to the undirected graph case
by Robertson and Seymour \cite{GMX}. We do not yet present the precise statement. Informally, our main result is as follows:

\begin{theorem}\label{thm:mainc}
 We can construct directed tree-decompositions of digraphs that distinguish all their tangles of order $k$ for any integer $k$. Moreover, such tree-decompositions can be constructed in polynomial time.
\end{theorem}

Formally, see Theorem \ref{thm:tangle-decomp}. 
Technically speaking we first give the ``tangle tree labelling'' theorem in Section 5, which distinguishes all tangles of large order. Let us just mention that for the undirected case, this is enough, as this yields a tree-decomposition (see \cite{GMX}). 
But in our case, we do not really get a tree-decomposition. Thus 
in Section 6, we convert the tree-labelling to the canonical tree-decomposition theorem for digraphs, which proves our main theorem. 
Note that the bound (for the separation) is not really tight (see Theorem \ref{thm:tangle-decomp}), but for our applications, this is enough. 

We also provide an algorithm which just mimics the proof in Section \ref{sec:alaspect}. 
As described above, a canonical tree-decomposition for undirected graphs enables us to gain knowledge of the global structure of a graph from a knowledge of its structure relative to each high order tangle. We expect that this is the case for directed graphs, and we are planning to work on it.

We also hope that Theorem \ref{thm:mainc} has many other algorithmic applications. Below, we give one such an example for the disjoint paths problem.

\subsection{Application to the disjoint paths problem}

The disjoint paths problem is defined as follows. The input consists
of a graph $G$ and $k$
pairs $(s_1, t_1), \dots,$ $(s_k, t_k)$ of vertices, called the
\emph{sources} and \emph{terminals}.
For any $c\geq 1$, a \emph{solution with congestion $c$}
is a set $P_1, \dots, P_k$ of paths such that $P_i$ links $s_i$ and
$t_i$ and every vertex $v\in V(G)$ is contained in at most $c$ of the
paths. For directed graphs, the paths $P_i$ are
required to be directed from $s_i$ to $t_i$.
In case $c=1$ we call $P_1, \dots,P_k$ an
\emph{integral solution} and in case $c=2$ we call it a
\emph{half-integral} solution.
Based on this notation, the \emph{$k$-disjoint paths} problem is the
problem to decide whether the input $\big(G, (s_1, t_1), \dots, (s_k,
t_k)\big)$ has an integral solution. The problem is one of the
classical NP-complete problems. However, on undirected graphs,
as a consequence of Robertson and Seymour's monumental graph minors
series, the problem can be solved in polynomial time for any fixed
number $k$ of source/terminal pairs \cite[\small Graph Minors XIII]{GM-series}, \cite{kkr}.

The situation is completely different for directed graphs. Fortune et
al.~\cite{FortuneHW80} proved that the
problem is already NP-complete for $k=2$ source/terminal pairs.
This implies that the problem is also not fixed-parameter tractable
parameterized by $k$. Therefore, much work has gone into
designing efficient algorithms for the problem on specific classes of
directed graphs. It is known that it can be solved in
polynomial time for any fixed $k$ on acyclic digraphs
\cite{FortuneHW80}. But Slivkins proved that it
is not fixed-parameter tractable on acyclic digraphs \cite{Slivkins03}.
A generalization of this to a much larger class of digraphs is given
in~\cite{JRST} where it is shown that the problem can be
solved in polynomial time, for every fixed $k$, on any class of
digraphs of bounded \emph{directed treewidth}. As the class of
acyclic digraphs has bounded directed treewidth, Slivkins'
W[1]-hardness result applies to classes of bounded directed treewidth
also.
On the other hand, Cygan et al. \cite{CyganMPP13} proved that the
problem is indeed fixed-parameter tractable with parameter $k$ when
restricted to planar digraphs.

In this paper we are not interested in solving
disjoint paths problems on special classes of digraphs but in
obtaining algorithms working on all directed graphs.
The best we can hope for is to
provide a polynomial-time algorithm for a slightly relaxed version of
the problem.
Following the approach in the literature on approximation
algorithms we will allow small congestion routing, but only a
congestion of $2$. More precisely, the main result of this paper is
the following theorem.

\begin{theorem}\label{thm:main-alg1}
	For every fixed $k\geq 1$ there is a polynomial-time algorithm for
	deciding the following problem.
	\nprob{12cm}{$k$-Half-Or-No-Integral Disjoint Paths}%
	{A digraph $G$ and terminals $s_1, t_1,
		s_2, t_2, \dots , s_k, t_k$.}%
	{  \vspace*{-1em}\begin{itemize}
			\item
			Find $k$ paths $P_1, \dots , P_k$ such that $P_i$ is from
			$s_i$ to $t_i$ for $i=1, \dots , k$ and every vertex in $G$ is in
			at most two of the paths or
			\item conclude that $G$ does not contain disjoint paths $P_1, \dots
			, P_k$ such that $P_i$ is from $s_i$ to $t_i$ for $i=1, \dots ,
			k$.
		\end{itemize}
	}
\end{theorem}

As mentioned above, we cannot replace the first outcome in
Theorem \ref{thm:main-alg1} by fully integral paths as the $k$-disjoint path problem is
NP-complete even for $k=2$. But perhaps the following stronger result may hold.

\begin{conjecture}
	For every fixed $k\geq 1$ there is a polynomial time algorithm for
	deciding the following problem.
	
	{\bfseries Input:} A digraph $G$ and terminals $s_1, t_1,
	s_2, t_2, \dots , s_k, t_k$.
	
	{\bfseries Output:} Are there $k$ paths $P_1, \dots , P_k$ such that $P_i$ is from
	$s_i$ to $t_i$ for $i=1, \dots , k$, and each vertex of $G$ is in at most two of the paths?
\end{conjecture}

A weaker result with some connectivity condition added is given in \cite{EdwardsMW16}.

Theorem~\ref{thm:main-alg1} improves the main result
of~\cite{KawarabayashiKK14}, where it was shown that for every fixed
$k$ there is a polynomial-time algorithm which computes a
solution to the $k$-disjoint paths problem of congestion $4$ or
determines that the input instance has no integral solution. Our
result not only improves the congestion to half-integral, the best
possible, but moreover, our proof is much cleaner, due to the canonical tree-decomposition. 

Slivkins~\cite{Slivkins03} proved that the directed disjoint paths
problem is W[1]-hard already on acyclic digraphs. Similar to it, we suspect that our algorithm for
Theorem \ref{thm:main-alg1} is optimal in the sense that the problem
is not fixed-parameter tractable (under the usual complexity theoretic
assumptions), i.e.~the running time of our algorithm
cannot be improved to $\mathcal{O}(f(k)n^c)$, for any fixed constant
$c$ and arbitrary function $f$. We leave it as an open problem.

Below, we shall highlight how we take advantage of the canonical tree-decomposition.

Let us give some overview.
For undirected graphs,
many algorithms for disjoint paths problems rely on the treewidth/excluded grid duality. That is, these algorithms make a case
distinction between graphs of bounded undirected treewidth and graphs
of high treewidth which therefore contain a large grid minor.

Now we show that a directed wall $W$ can be used as a {\sl crossbar}. More precisely,
we can find the desired $k$ paths of our solution such that each
vertex is in at most two of the
paths, provided there is no small separation from $\{s_1,\dots, s_k\}$ to
$W$ and from $W$ to $\{t_1,\dots, t_k\}$. This extends the undirected
case by Kleinberg \cite{jon} (see also \cite{KK12, kr}).

So it remains to consider the case that we cannot use $W$ as a crossbar.
So far, our proof is almost same as that of \cite{KawarabayashiKK14}.
Below, we have a big difference.
\begin{enumerate}
	\item
	In this paper we take advantage of a canonical tree-decomposition.
	It allows us to treat each distinguishable tangle (or grid minor) separately. This leads us to solve Theorem \ref{thm:main-alg1} by dynamic programming.
	\item
	On the other hand, in \cite{KawarabayashiKK14},
	we use {\sl important separators}, one of the
	key concepts developed in the parameterized complexity community (see \cite{PARBOOK}).
	Using important separators allows us to take separations so that each
	separation crosses at most $4^k$
	others. This allows us to treat each distinguishable tangle (or grid minor) {\sl almost} separately, but this ``almost'' causes a lot of technical difficulties. Moreover, due to these difficulties, we cannot quite show that the congestion is at most 2, but we can only show that it is at most 4. The improvement from ``4'' in \cite{KawarabayashiKK14} to ``3'' in \cite{stephan} can be trivially done because in the first, we use an ``almost'' grid, but in the second, we can take advantage of the full grid minor. This allows us to improve the congestion by one. 
\end{enumerate}

Thus our proof of Theorem \ref{thm:main-alg1} is much cleaner (and the result is stronger), due to the canonical tree-decomposition theorem, Theorem \ref{thm:mainc}. We expect many more important applications to come.

The paper is organized as follows.
Section~\ref{sec:prelim} provides some preliminary concepts including directed treewidth and directed walls.
In Section 3, we introduce directed separations and brambles, and we
introduce tangles in Section 4. In Section 6, we give the ``tangle
tree labelling'' theorem, which distinguishes all tangles of large
order. Let us just mention that for the undirected case, this is
enough, as this yields a tree-decomposition.  
But in our case, we do not really get a tree-decomposition. Thus 
in Section 7, we convert the tree-labelling to the canonical tree-decomposition theorem for digraphs, which proves our main theorem. We also provide an algorithm which just mimics the proof in Section \ref{sec:alaspect}. 
Using this result, we provide a polynomial-time algorithm for \textsc{$k$-Half-Or-No-Integral Disjoint Paths}
in Section~\ref{route}--\ref{sec:DP}, when $k$ is fixed.
In Section~\ref{route},
we show that if a digraph contains a large
cylindrical wall which cannot be separated from the
terminals by low order separations, then we can construct a half-integral solution in polynomial time.
Using this, in Section~\ref{leafbag},
we prove that if a given graph has no two distinguishable tangles of certain high order (with respect to $k$),
then \textsc{$k$-Half-Or-No-Integral Disjoint Paths} can be solved in polynomial time.
This will correspond to the algorithm on leaf bags of the canonical tree-decomposition,
and based on this, we design a dynamic programming algorithm for the problem in Section~\ref{sec:DP}.

        

\section{Preliminaries}
\label{sec:prelim}

In this section we fix some general notation used through the paper
and recall some results on directed tree width needed in the sequel.

\subsection{Digraphs and DAGS} 
\label{subsec:digraphs-dags}

Unless stated otherwise, all graphs in this paper are directed. We
denote the vertex set of a digraph $G$ by $V(G)$ and its edge set by
$E(G)$. If $e = (u, v)$ is an edge then $u$ and $v$ are its
\emph{ends}. We call $u$ its \emph{tail}
and $v$ its \emph{head}.

Let $G$ be a digraph and let $v \in V(G)$. Let $u_1, \dots, u_l$ be the
in-neighbours of $v$ and $w_1, \dots, w_s$ be the out-neighbours of
$v$. Let $G'$ be the digraph obtained from $G-v$ by adding two fresh
vertices $v_{in}, v_{out}$, the edge $(v_{in}, v_{out})$ and
the edges $(u_i, v_{in})$ and $(v_{out}, w_j)$, for all $1 \leq i \leq l$ and $1
\leq j \leq s$. We say that $G'$ was obtained from $G$ by \emph{splitting} $v$.

Let $G$ be a digraph and let $X \subseteq V(G)$. By a \emph{component},
or \emph{strong component}, of $G$ we mean a strongly connected component
of $G$. 
  Let $C_1, \dots, C_l$ be the strong components of $G-X$. The
  \emph{component dag $\DDD = \DDD(G, X)$} is defined as usual: the vertices are the
  components $C_1, \dots, C_l$ and 
  there is an edge from $C_i$ to $C_j$ whenever there is an edge from
  a vertex in $C_i$ to a vertex in $C_j$. The \emph{height} of a component
  $C_i$ in $\DDD$ is the maximal length of a directed path in $\DDD$ 
  starting at $C_i$. Let $C_1 < \dots < C_l$ be the topological ordering of $V(\DDD)$ in increasing
  order of height, i.e.~$C_1$ is a sink in $\DDD$ (a vertex of out-degree $0$)
   and $C_l$ is a source (in-degree $0$).
  Let $\sqsubset$ be the lexicographic order on $\{ C_1, \dots, C_l \}$,
  i.e.~if $A, B \subseteq \{ C_1, \dots, C_l \}$ then $A \sqsubset B$  if the smallest
  vertex $v$ in $B$
  (with respect to $\sqsubset$) is bigger than the smallest vertex $u$ in
  $A$ or $v = u$ and  $A - u \sqsubset B - v$.

  A set $A \subseteq \{ C_1, \dots, C_l\}$ is \emph{downwards closed} if for every $v \in
  A$ also all out-neighbours of $v$ in $\DDD$ are
  in $A$. Similarly, $A$ is \emph{upwards closed} if for every $v \in
  A$ also all in-neighbours of $v$ in $\DDD$ are
  in $A$.

  We will also use the well-known concept of subdivisions.

\begin{definition}[subdivision]
  Let $G$ be a digraph. A digraph $H$ is a \emph{subdivision} of $G$
  if $H$ can be obtained from $G$ by replacing a set $\{e_1, \dots,
  e_k\} \subseteq E(G)$ of edges by pairwise internally vertex-disjoint directed paths $P_1, \dots, P_k$
  such that if $e_i = (u, v)$ then $P_i$ links $u$ to $v$..
\end{definition}

\subsection{Directed Tree-Width and Grids}
\label{subsec:directed-tree-width}
We briefly recall the definition of directed tree width from \cite{JRST}.
See also \cite{KreutzerK18,Kreutzer13} for a thorough introduction to
directed tree width and its obstructions including proofs of many of
the results stated here. Our notation follows~\cite{KreutzerK18}.

Let $G$ be a digraph and let $X, Y \subseteq V(G)$.
We say that $Y$ \emph{guards} $X$, or is a
\emph{guard} of $X$, if every directed walk starting
and ending in $X$ which contains a vertex of $V(G) \setminus X$
also contains a vertex of $Y$. In particular, ${X \setminus Y}$ is
the union of the vertex sets of some set of strong components of $G-Y$

An \emph{arborescence}, or \emph{out-branching}, is a directed tree
$T$ obtained from an undirected rooted tree by orienting every edge
away from the root. For $t \in V(T)$ we define $T_t$ as the subtree of
$T$ rooted at $t$, i.e.~the subtree containing all vertices $s$ such
that the unique directed path from the root $r$ of $T$ to $s$ contains $t$. 
For $r, r' \in V(T)$, we write $r <_T r'$, or $r' >_T r$,  if $r' \neq r$ and there exists
a directed path in $R$ from $r$ to $r'$. If $e \in E(R)$ is an edge
with head $r$, we write $e <_T r'$ if either $r' = r$ or $r' > r$.
Finally, we write $r \leq_T r'$ if $r <_T r$ or $r = r'$ and we define $e \leq_T r'$
analogously. 
We
omit the index $T$ if $T$ is understood from the context.

Two nodes $s, t \in V(T)$ are called \emph{independent} if
none of the paths $P_s$ and $P_t$ from the root to $s$ and $t$, resp.,
contains both nodes $s$ and $t$. Otherwise $s$ and $t$ are called
\emph{dependent}. The \emph{least common ancestor} of $s$ and $t$ is
the node $a \in V(T)$ such that $T$ contains two internally vertex
disjoint paths $P$ and $Q$ such that $P$ links $a$ to $s$ and $Q$
links $a$ to $t$.

\begin{definition}\label{def:directed-tree-decompositions}
  A \emph{directed tree-decomposition} of a digraph $G$ is a triple
  $(T,\beta,\gamma)$, where $T$ is an arborescence,  $\beta \sth V(T)
  \rightarrow 2^{V(G)}$ and $\gamma \sth E(T) \rightarrow 2^{V(G)}$ are
  functions such that
  \begin{enumerate}
  \item $\{ \beta(t) \sth t\in V(T) \}$ is a partition of $V(G)$ (we allow empty
    sets $\beta(t)$) 
  \item  if $e \in  E(T)$, then  $\bigcup \{ \beta(t) \sth t \in V (T), e < t\}$ is the
    vertex set of a strong component of $G - \gamma(e)$.  
  \end{enumerate}
  For any $t \in V(T)$ we define $\Gamma(t) := \beta(t) \cup \bigcup \{ \gamma(e) \sth e \sim t\}$,
  where $e \sim t$ if $e$ is incident with $t$. If $t \in V(T)$ is not the
  root and $e = (s, t) \in E(T)$ then we define $\omega(t) := \beta(t) \cup \gamma(e)$
  and we define $\omega(r) := \beta(r)$ for the root. If $S
\subseteq V(T)$ or $S \subseteq T$, we write $\beta(S)$ for $\bigcup_{s \in S} \beta(s)$ and $\bigcup_{s \in
  V(S)}\beta(s)$, resp.

  The sets $\beta(t)$ are called the \emph{bags} and the sets
  $\gamma(e)$ are called the \emph{guards} of the directed tree
  decomposition.

  The \emph{width} of $(T, \beta, \gamma)$ is the least integer $w$ such that
  $|\Gamma(t)| \leq w + 1$  for all $t \in V(T)$.  The \emph{directed tree
    width} of $G$ is the least integer $w$ such that $G$ has a
  directed tree-decomposition of width $w$. 
\end{definition}

It is easy to see that the directed tree width of a
subdigraph of $G$ is at most the directed tree width of $G$. 

\begin{remark}
  In the original definition of directed tree-decompositions in
  \cite{JRST}, where they are called \emph{arboreal decompositions},
  instead of Condition 2. above a slightly weaker condition is
  used. This weaker condition causes problems in several proofs using
  directed tree-decompositions, including some in \cite{JRST}. In
  \cite{JohnsonRobSeyTho01b}, Johnson et al.~therefore replace it by
  Condition 2 above, which is nowadays the standard definition used in
  the literature. 

  Furthermore, in \cite{JohnsonRobSeyTho01b}, Johnson et al. also show that
  one can relax the requirement that the bags form a partition of the
  vertex set into non-empty sets and allow empty bags as well. Again
  this leads to a broadly (up to constant factors) equivalent concept
  of width. Some of the proofs below can be simplified when empty bags
  are allowed and we will therefore allow empty bags. See also \cite{KreutzerK18}. 
\end{remark}

It is known that every digraph $G$ of directed tree width $\leq k$ has a
directed tree-decomposition of the following particularly nice form.

\begin{definition}\label{def:dtw}
  Let $G$ be a digraph. A \emph{directed tree-decomposition}
  $\TTT := (T, \beta, \gamma)$ of $G$ is \emph{nice}, or \emph{monotone}, if
  \begin{enumerate}
  \item for all $e = (s, t) \in E(T)$, the set
    $\beta(T_t) := \bigcup_{t \leq_T t'} \beta(t)$ is a strong
    component of $G-\gamma(e)$ and \label{def:dtw:1}
  \item $\bigcup_{t <_T t'}\beta(t') \cap \bigcup_{e \sim t} \gamma(e)=\emptyset$ for every $t \in V(T)$.\label{def:dtw:2}
  \end{enumerate}
\end{definition}

Observe that item $2$ of the previous definition is equivalent to $\beta(T_t)
\cap \Gamma(t) =  \beta(t)$.
The proof of the next lemma is a well-known construction on directed
tree-decompositions. We repeat the proof here as we need the
construction in the proof later in the paper. By a relaxed directed
tree-decomposition we mean a directed tree-decomposition $(T, \beta, \gamma)$
where we allow for edges $e = (s, t) \in E(T)$ that $G[\beta(T_t)]$ is a
union of strong components of $G \setminus \gamma(e)$ guarded by $\gamma(e)$.  
\begin{lemma}\label{lem:dtd->scc-dtw}
  Let $G$ be a digraph and let $(T, \beta, \gamma)$ be a relaxed directed
  tree-decomposition of $G$ of width $k$. Then there is a directed
  tree-decomposition $(T', \beta', \gamma')$ of width $k$ such that
  $\beta(T_t)$ is a strong component of $G - \gamma(e)$ for all $e = (s, t) \in E(T')$.
\end{lemma}
\begin{proof}
  Let $r$ be the root of $T$. For every child $t$ of $r$, let $C_1,
  \dots, C_s$ be the strong components of $G - \gamma(r, t)$ that contain
  a vertex of $\beta(T_t)$. If $s > 1$ let $T_1, \dots, T_s$ be disjoint
  isomorphic copies of $T_t$, i.e.~$V(T_i) := \{ t^i \sth t \in V(T_t) \}$ and
  $E(T_i) := \{ (s^i, t^i) \sth (s, t) \in E(T_t) \}$. Now, for all $1 \leq i \leq
  s$,  $t^i \in V(T_i)$, and $(s^i, t^i)
  \in E(T_i)$, we define   $\beta^i(t^i) := \beta(t) \cap V(C_i)$ and  $\gamma(s^i, t^i) := \gamma(s, t)$. 
  Let $T'$ be the tree obtained from $T$ by removing the subtree $T_t$
  and instead adding, for all $1 \leq i \leq s$, the tree $T_i$ and the edge
  $e = (r, r^i)$ with $\gamma'(r, r^i) := \gamma(r, t)$, where $r^i$ is the root
  of $T_i$. We define
  $\beta'(u) := \beta(u)$ for all $u \not\in V(T_t)$, $\beta'(u^i) := \beta^i(u^i)$ for
  all $u^i \in V(T_i)$ and $\gamma'(e) = \gamma(e)$ for all $e \not\in E(T_t)$ and
  $\gamma'(e^i) := \gamma^i(e^i)$ for all $e^i \in E(T_i)$.  

  It follows from the construction that for every edge $(t_1, t) \in
  E(T')$ the set $\beta(T'_t)$ is the vertex set of a strong component of
  $G - \gamma'(e)$. We still need to show that $(T', \beta', \gamma')$ is a directed
  tree-decomposition, i.e. that for no $t' \in V(T')$ with incoming edge
  $e' = (t_1', t')$ there is a directed path $P$ in $G - \gamma'(e)$ starting and ending in
  $\beta'(T'_{t'})$ which contains a vertex in $V(G) - \beta'(T'_{t'})$. 

  Towards a contradiction, suppose there were such a node $t'$ with incoming edge
  $e = (t_1', t')$ and path
  $P$ as above. Let $t_1, t$ be the original copies of $t_1', t'$ in $T$.
  Then the start and
  endpoints $a, b$ of $P'$ are in $\beta(T_{t})$. As the $\gamma'(t_1', t') = \gamma(t_1, t)$, the path $P$ cannot
  contain a vertex outside of $\beta(T_t)$. Thus $P$ has to be contained
  in $\bigcup \{ \beta'(T_i) \sth 1 \leq i \leq s \}$. This implies that $P$ has to visit
  two different subtrees $T_i, T_j$, i.e.~$P$ is a path in $G - \gamma(t_1,
  t)$ starting and ending in the same strong component of $G - \gamma(t_1,
  t)$ but containing vertices of at least two different
  components. This is a contradiction. Thus, $T'$ is a directed
  tree-decomposition.  

  If $T'$ contains another node $t$ with incoming edge $(s, t)$
  violating the property of the lemma, then we choose one closest to
  the root and repeat the procedure above.
\end{proof}

\begin{remark}\label{rem:split-dtd}
  For later reference, note that the construction in the previous
  lemma has the property that the guards on edges are not changed.
  I.e.~if we start with a tree-decomposition $(T, \beta, \gamma)$ then we
  obtain a new decomposition with the property that every edge
  $e' = (s, t)$ points to a strong component $G[\beta(T_t)]$ of
  $G - \gamma(s, t)$ and that $e'$ originates from an edge
  $e \in E(T)$ and $\gamma(e') = \gamma(e)$. We will make use of this property
  later on.
\end{remark}

\subsection{Obstructions to Directed Tree-Width}

We now recall the concept of cylindrical
walls as defined in \cite{JRST,Reed97}. See Figure~\ref{fig:grid}.

\begin{definition}[cylindrical grid and wall]\label{def:cyl-grid}\label{def:wall}
  The \emph{cylindrical grid} of order $k$, for some $k\geq 1$, is the
  digraph $G_k$ with vertex set $\{ v^i_j \sth 1 \leq i \leq k, 0 \leq j
  \leq 2k-1 \}$ and edge set $\{ (v^i_j, v^i_{j+1 \mod 2k}) \sth 1 \leq i \leq k, 0 \leq j <
  2k-1\} \cup \{ (v^i_{2j}, v^{i+1}_{2j}), (v^{i+1}_{2j+1}, v^{i}_{2j+1}) \sth 1 \leq i <
  k, 0 \leq j < k\}$.

  The \emph{elementary cylindrical wall $W_k$} of order $k$ is defined
  as the digraph obtained from $G_k$ by splitting every vertex $v^i_j$
  with $1 < i < k$, i.e.~by splitting every vertex of $G_k$ of in-
  and out-degree $2$.

  A \emph{cylindrical wall}  of order $k$ is a subdivision of $W_k$.
\end{definition}
  
Thus, a cylindrical grid of order $k$ consists of $k$ directed
cycles $C_1, \dots, C_k$, which we refer to as the \emph{nested
  cycles} of $G_k$,  oriented in the  same direction together
with a set of $2k$ pairwise vertex disjoint paths $P_0, \dots, P_{2k-1}$
intersecting every cycle such that for odd $i$ the paths $P_i$
intersect the cycles from the inside out, i.e.~in the order $C_k, C_{k-1}, \dots, C_1$, and for even
$i$ they intersect the cycles in the reverse order. We refer to the
paths $P_i$ as the \emph{horizontal paths}.

Sometimes it will be convenient to number the horizontal paths as  $P^i_j$, for $i
= 1, 2$ and $0 \leq j \leq k-1$, such that $P^i_j = P_{2j+i-1} =
G[\{v^c_{2j+i-1} \sth 1 \leq c \leq k\}$. 

We often think of the grid as being drawn in the plane with  $C_1$ as
the outer cycle and all $C_i$ drawn clockwise and $P^0_0$ drawn on top with the other paths drawn
below $P^0_0$. Then the paths $P^0_j$ go from left to right and the
paths $P^1_j$ go from right to left.

\begin{figure}
  \begin{center}
    \beginpgfgraphicnamed{fig-grid}
    \begin{tikzpicture}[scale=0.6]
      \tikzstyle{vertex}=[shape=circle, fill=black, draw, inner
      sep=.4mm]
      \tikzstyle{emptyvertex}=[shape=circle, fill=white,
      draw, inner sep=.7mm]
      \foreach \x in {1,2,3,4}
      {
        \draw (0,0) circle (\x);
        \ifnum\x<4 \draw[<-, very thick] (180-22.5:\x) --
        (180-23.5:\x);
        \draw[<-, very thick] (-22.5:\x) -- (-23.5:\x);
        \fi
        \foreach \z in {0,45,...,350}
        {
          \node[vertex] (n\z\x) at (\z:\x){};
        }
      }
      \draw[<-, very thick] (168-12.25:4) -- (167-12.25:4);
      \draw[<-, very thick] (168+180-12.25:4) -- (167+180-12.25:4); \foreach \z in {0,90,...,350} { \draw[
        decoration={markings,mark=at position 1 with {\arrow[very
            thick]{>}}}, postaction={decorate} ] (\z:1) -- (n\z4); }
      \foreach \z in {45,135,...,350} { \draw[
        decoration={markings,mark=at position 1 with {\arrow[very
            thick]{>}}}, postaction={decorate} ] (\z:4) -- (n\z1); }
    \end{tikzpicture}\endpgfgraphicnamed
    \caption{Cylindrical grid of order $4$.}\label{fig:grid}
  \end{center}
\end{figure}
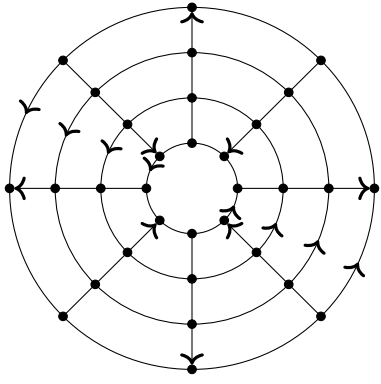

\begin{remark}\label{rem:grid-in-wall}
  Clearly, every cylindrical wall of order $k$ contains a cylindrical grid of
  order $k$ as a butterfly minor. Conversely, a cylindrical grid of
  order $k$ contains a cylindrical wall of order $\frac k2$ as
  subgraph.
\end{remark}



\section{Directed Separations}

In this section we prove some fundamental properties of separations in directed graphs.   

Let \(G\) be a digraph and let \(A, B \subseteq V(G)\).
An edge $e \in E(G)$ \emph{crosses from $A$ to $B$} if its tail is in $A
\setminus B$ and its head is in $B \setminus A$. A \emph{cross edge for $A, B$} is an
edge which crosses from $A$ to $B$ or from $B$ to $A$.

\begin{definition}[Directed Separation]\label{def:dir-separation}
  A \emph{directed separation} of \(G\) is a pair \((A, B)\) of
  subgraphs of \(G\) such that \(V(A) \cup V(B) = V(G)\) and either there are
  no cross edges from $V(A)$ to $V(B)$ or there are no cross edges from
  $V(B)$ to $V(A)$.
  The \emph{order} of the separation is $|V(A) \cap V(B)|$ and it is
  denoted by $|(A, B)|$.
  The \emph{separator} of \((A, B)\), denoted  $\sep(A, B)$, is the set \(V(A) \cap V(B)\). 

  We write \(X = (A  \dsr  B)\) to indicate that \(X\) is a
  separation with no cross edge from $V(B)$ to $V(A)$. Analogously, we write \((A \dsl
  B)\) if there are no cross edges from  $V(A)$ to $V(B)$.\footnote{I.e.,
    the arrow in $(A \dsr B)$ indicates the direction in which cross
    edges are allowed. This is repeated in our figures below.} 

  By $\septop(X)$ and $\sepbot(X)$, or simply $X^+, X^-$,  we denote the two parts of $X$
  such that $X = (\septop(X)  \dsr \sepbot(X))$. If $X$ has no cross edges
  at all, then we assign $\septop(X)$ and $\sepbot(X)$ to the two sides of
  $X$ arbitrarily.\footnote{We often illustrate separations as in
    Figure~\ref{fig:cross-sep} with one side on top and the other
    on bottom such that cross edges are allowed from top to bottom but not
    from bottom to top. The names $\septop(X), \sepbot(X)$ originate
    from this and make it easy to speak about the "bottom" and "top"
    quadrant of a pair of crossing separations later.} 

  Let \(S\) and \(T\) be two sets of vertices in a digraph \(G\). We say that
  a separation \(X\) \emph{separates $S$ from $T$}, or is a separation
  \emph{from \(S\) to \(T\)}, if \(T \subseteq X^+\) and
  \(S \subseteq X^-\). 
\end{definition}

Alternatively, we could define directed separations as pairs of
vertex sets $(V(A), V(B))$ in the obvious way.
The main difference is that if $A, B \subseteq G$ such that  $(A, B)$ is a directed
separation in $G$ and $e \in A$ is an edge with
both ends in $V(A) \cap V(B)$, then $(A - e, B + e)$ is a separation
different from $(A, B)$. Obviously they would be the same when
specified by $(V(A), V(B))$.

Our definition of directed separations is consistent with the usual definition of separations in
undirected graphs where a separation is a pair $(A, B)$ of subgraphs
with $A \cup B = G$. In the directed setting, we usually do not have that
$A \cup B = G$, as we allow cross edges from one side to the other.
Therefore, for directed separations, this difference between the two
ways of specifying separations is less 
relevant. In the sequel, abusing notation, we will use both forms of
specifying a separation $X = (A, B)$ and we will not distinguish notationally between the subgraphs
$A, B$ and their vertex
sets $V(A), V(B)$. E.g., we will write $v \in \septop(X)$ etc.

The following notation will be very convenient in the sequel.

\begin{definition}
  Let $A \subseteq V(G)$.
  The \emph{(out-)boundary of $A$} is the set $\partial^+(A) := \{ a \in A \sth $
  there is $b \not\in A$ with $(b, a) \in E(G) \}$ of vertices $v \in A$
  which have an in-neighbour outside $A$. The \emph{in-boundary} $\partial^-(A)$ is defined
  analogously.

  The set $A$ induces two separations $X^{+}(A) := (A \cup \partial^+(A), \partial^+(A) \cup
  (V(G) \setminus A))$ and $X^{-}(A) := (A \cup \partial^-(A), \partial^-(A) \cup (V(G) \setminus A))$. 
\end{definition}

The names \emph{out-} and \emph{in-boundary} are chosen so that if $X$
is a directed separation then $\sep(X) = \partial^+(X^+) = \partial^-(X^-)$.

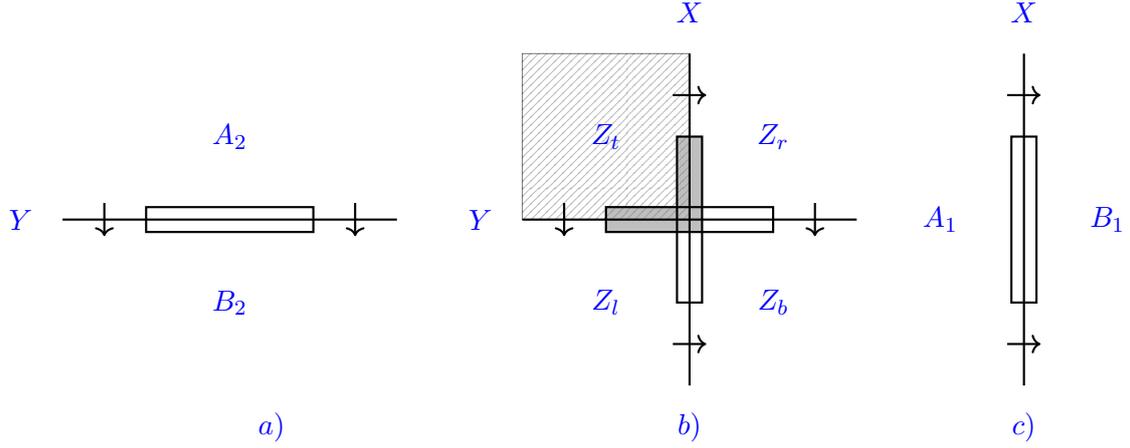
\begin{figure}
  \centering
\beginpgfgraphicnamed{fig-cross-sep}
\begin{tikzpicture}[scale=0.55]

    \draw[thick] (-15,0) -- (-7,0);
    \draw[thick] (-9,0.3) rectangle (-13,-0.3);
    \draw[thick,<-] (-8,-0.4) -- (-8,0.4);
    \draw[thick,<-] (-14,-0.4) -- (-14,0.4);
    \node[] at (-11,2) {$A_{2}$};
    \node[] at (-11,-2) {$B_{2}$};
    \node[] at (-16,0) {$Y$};

    \draw[thick] (8,-4) -- (8,4);
    \draw[thick,->] (7.6,3) -- (8.4,3);
    \draw[thick,->] (7.6,-3) -- (8.4,-3);
    \draw[thick] (8.3,-2) rectangle (7.7,2);
    \node[] at (6,0) {$A_{1}$};
    \node[] at (10,0) {$B_{1}$};
    \node[] at (8,5) {$X$};

    \draw[pattern=north east lines, opacity=0.5] (-4, 0) rectangle (0, 4) ;
    \draw[fill=gray,opacity=0.5] (-0.3,-0.3) -- (0.3,-0.3) -- (0.3,2) -- (-0.3,2) -- (-0.3,0.3) -- (-2,0.3) -- (-2,-0.3) -- (-0.3,-0.3);
    \draw[thick] (-4,0) -- (4,0);
    \draw[thick] (0,-4) -- (0,4);
    \draw[thick] (0.3,-2) rectangle (-0.3,2);
    \draw[thick] (-2,0.3) rectangle (2,-0.3);
    \node[] at (-2,2) {$Z_{t}$};
    \node[] at (-2,-2) {$Z_{l}$};
    \node[] at (2,-2) {$Z_{b}$};
    \node[] at (2,2) {$Z_{r}$};
    \node[] at (0,5) {$X$};
    \node[] at (-5,0) {$Y$};
    \draw[thick,->] (-0.4,3) -- (0.4,3);
    \draw[thick,->] (-0.4,-3) -- (0.4,-3);
    \draw[thick,<-] (3,-0.4) -- (3,0.4);
    \draw[thick,<-] (-3,-0.4) -- (-3,0.4);
    \node at (-10, -5) { $a)$ } ;
    \node at (0, -5) { $b)$ } ;
    \node at (8, -5) { $c)$ } ;
  \end{tikzpicture}\endpgfgraphicnamed
  \caption{$a)$ and $c)$: separations $Y = (A_2, B_2)$ and  $X = (A_1, B_1)$.
      $b)$ illustrates the four parts of $(X, Y)$. The
      shaded area is the \emph{top corner} and the marked region is
      the top quadrant $Z_t$.}
    \label{fig:cross-sep}
\end{figure}

\begin{definition}[Corners, Parts and Quadrants]\label{def:sep-parts}
  Let \(X\) and \(Y\) be 
  separations. The (ordered) pair \((X, Y)\) defines a partition of \(G\) into
  four \emph{parts}, or \emph{quadrants}: the \emph{top} \(Z_t := X^+
  \cap Y^+\), the \emph{left} \(Z_l := X^+ \cap Y^-\), the \emph{right}
  \(Z_r := X^- \cap Y^+\), and the \emph{bottom} \(Z_b := X^- \cap Y^-\). Note that top and
  bottom do not depend on the order of the pair \((X, Y)\) whereas
  left and right do. If the order is not important, we refer to
  $Z_l$ and $Z_r$ as the \emph{middle} parts. We define $\septop(X,
  Y)$ and $\sepbot(X, Y)$ as the top and bottom quadrant of the pair
  $(X, Y)$.

  For $Q \in \{ Z_t, Z_l, Z_r, Z_b \}$ we refer to the set $S(Q) := (\sep(X) \cup
  \sep(Y)) \cap V(Q)$ as the  \emph{corner} of the quadrant
  $Q$ and refer to $S(Z_t)$ as the  \emph{upper}, or 
  \emph{top}, corner and to  $S(Z_b)$ as the \emph{lower}, or \emph{bottom}, corner. See
  Figure~\ref{fig:cross-sep} for an illustration. 
\end{definition}



\begin{definition}[crossing separations]\label{def:crossing-sep}
  Let $X = (A, B)$ and $Y = (A', B')$ be a pair of separations.
  \begin{enumerate}
  \item We say that $X$ and $Y$ are \emph{uncrossed} if $A' \subseteq A$ and
    $B \subseteq B'$ or $A' \subseteq B$ and $A \subseteq B'$. Otherwise $X$ and $Y$ \emph{cross}.
  \item Let $Q$ be any of the four quadrants of the pair $(X, Y)$ and let $C$ be the
    corner of $Q$.  If $(Q, (G - Q) \cup C)$ is a separation (with separator $C$) then we say that
    $X$ and $Y$ \emph{uncross towards} $Q$.  
  \item If $Z = (A'', B'')$ is another separation of $G$, we say that $Z$ is
    \emph{uncrossed} with $(X, Y)$ if it is uncrossed with at least one
    of $X$ and $Y$. Otherwise $Z$ crosses the pair $X, Y$.  
\end{enumerate}
\end{definition}

In the case of undirected graphs, it is easily seen that any pair of separations
uncrosses towards each of the four quadrants. Furthermore, the \emph{submodularity property} of
undirected separations implies that for any pair of opposite
quadrants, at least one of the two is a separation of order at most the order of $X$ or
$Y$. For directed separations, only the top and bottom quadrants uncross but in general a
pair of directed separations does not uncross towards its middle parts. This makes
directed separations much harder to work with and is the main reason
why we cannot simply apply the proof techniques for undirected
separations and tangles to the directed case. 

The next lemma establishes submodularity of directed separations with
respect to their top and bottom quadrant.

\begin{lemma}[Submodularity]\label{lem:submodularity}
  Let \(X\) and \(Y\) be a pair of separations in a
  digraph \(G\). Let \(T = X^+ \cap Y^+\) be the top of the pair  $(X,
  Y)$ and $B := X^- \cap Y^-$ be the bottom. Let $\bar{T} := X^- \cup Y^-$
  and let $\bar{B} := X^+ \cup Y^+$. Let $S_t$
  be the upper corner of $(X, Y)$ and $S_l$ be the lower corner. 

  Then $\big(T, \bar{T} \big)$ and $\big(B, \bar{B}\big)$ are separations with separators
  $S_t$ and  $S_l$, resp., which are uncrossed with $X$ and $Y$. Furthermore, $|S_l| +
  |S_t| = |\sep(X)| + |\sep(Y)|$ and thus $\min\{|S_l|,|S_t|\} \leq \max \{ |\sep(X)|,
  |\sep(Y)|\}$. In particular, if $k = |\sep(X)| = |\sep(Y)|$, then at least one of $S_l,
  S_t$ is of order at most $k$. 
\end{lemma}
\begin{proof}
  Clearly, $S_t = V(T \cap \bar{T})$ and $S_l = V(B \cap \bar{B})$. 
  
  We first verify that $(T, \bar{T})$ forms a directed separation
  $\big(T \dsr \bar{T}\big)$. Towards a contradiction, suppose there
  was an edge $(v,w)$ with $v \in \bar{T} - T$ and $w \in T - \bar{T}$. 
  Note that $\bar{T} - T = (X^- \cup Y^-) - (X^+ \cap Y^+) = (X^- - X^+) \cup
  (Y^- - Y^+)$. If $v \in X^-\setminus X^+$, then $w$ must be in $X^+ \cap X^-$, as
  $(X^+\to X^-)$ is a separation. Analogously, if $v \in Y^-\setminus Y^+$, then
  $w$ must be in $Y^+ \cap Y^-$, as $(Y^+ \dsr Y^-)$ is a
  separation. Thus, $w$ is in $X^- \cup Y^-$ and therefore cannot be in
  $T - \bar{T}$.  This shows that $\big(T, \bar{T}\big)$ is a directed
  separation. 

  A symmetric argument establishes that $(\bar{B} \dsr B)$ is a directed separation.

  By construction, $\sep(T, \bar{T}) = S_t$ and $\sep(B, \bar{B}) = S_l$ and thus 
  \[
    \begin{array}{rcl}
  |S_t|+|S_l| 
  &=& |(X^+ \cap Y^+) \cap (\sep(X) \cup \sep(Y))| + |(X^- \cap Y^-) \cap (\sep(X) \cup \sep(Y))|\\ 
  &=& |(Y^+ \cap \sep(X))| + |(X^+ \cap  \sep(Y))| - |\sep(X) \cap \sep(Y)|\\  
  && +  |(Y^- \cap \sep(X))| + |(X^- \cap \sep(Y))| - |\sep(X) \cap \sep(Y)| \\
 & =& |\sep(X)| + |\sep(Y)|
    \end{array}
    \]
  from which the remaining statements of the lemma follow.
\end{proof}

The next corollary restates the submodularity property in a form
frequently applied below.

\begin{corollary}\label{cor:submodularity}
  Let $A_1, A_2 \subseteq V(G)$ be sets and let $v \in \{ +, - \}$ be an
  orientation. Then $|\partial^v(A_1)| +|\partial^v(A_2)| \geq |\partial^v(A_1 \cap A_2)| +
  |\partial^v(A_1\cup A_2)|$. 
\end{corollary}


\section{Tangles}
In this section we define the concept of directed tangles and prove
the main structural result of this paper, the directed canonical
tangle decomposition theorem. The results 
obtained in this section can be seen as a directed version of the
results by Robertson and Seymour in Graph Minors X~\cite{GMX}.

\begin{definition}  \label{def:tangle}
  Let $G$ be a digraph.  A set $\TTT$ of separations of order $< k$ is called a
  \emph{tangle of order} $k$, if
  \begin{enumerate}
  \item for all directed separations $(A, B)$ of $G$ of order $< k$, either
    $(A, B) \in \TTT$ or $(B, A) \in \TTT$, and  
  \item if $(A_1, B_1), (A_2, B_2), (A_3, B_3) \in \TTT$ then  $V(A_1 \cup A_2 \cup A_3) \neq V(G)$. 
  \end{enumerate}
  For all $(A, B) \in  \TTT$ we call $B$ the \emph{big} side of $(A, B)$
  and $A$ the \emph{small} side. For every tangle $T$ we define an
  associated function $\beta_T$ mapping every separation $(A, B) \in T$ 
  to its big side $B$.
\end{definition}

We first prove a simple property of tangles used frequently throughout
the rest of the section.

\begin{lemma}\label{lem:tangle-tech}
  Let $T$ be a tangle of order $k$ in a digraph $G$ and let $(A_1, B_1), (A_2,
  B_2)$ be directed separations in $G$. If $(A_1, B_1) \in T$ and $(A_2,
  B_2) \in T$ and $|(A_1 \cup A_2, B_1 \cap B_2)| < k$ then $(A_1 \cup A_2, B_1 \cap
  B_2) \in T$.
\end{lemma}
\begin{proof}
  Towards a contradiction, suppose $(B_1 \cap B_2, A_1 \cup A_2) \in T$. Then
  $A_1, A_2$, and $B_1 \cap B_2$ are three small sides whose union covers
  $V(G)$, a contradiction to the tangle axioms.
\end{proof}

\begin{definition}
  Let $T, T'$ be tangles in a digraph $G$. $T$ and $T'$ are
  \emph{indistinguishable} if $T \subseteq T'$ or $T' \subseteq T$. Otherwise they are
  \emph{distinguishable}.

  A separation $X$ of $G$ \emph{distinguishes} $T$ and $T'$, or is a
   \emph{$T{-}T'$-distinguisher}, if
 $\beta_T(X) \neq \beta_{T'}(X)$.   The \emph{order} of the distinguisher is the order
  of the separation.

  $T$ and $T'$ are \emph{$l$-distinguishable}, for some $l \geq 0$, if
  there is a $T{-}T'$-distinguisher of order $l$.
  
  Finally, if $\SSS, \TTT$ are sets of tangles in $G$, then by an
  $\SSS{-}\TTT$-distinguisher we mean a separation $X$ which
  distinguishes every tangle in $\SSS$ from every tangle in $\TTT$.
\end{definition}

It is easily seen that a pair of tangles $T$ and $T'$ in a digraph $G$
is indistinguishable 
if, and only if, there is no separation $X$ of $G$ that distinguishes
$T$ and  $T'$. See 
e.g.~\cite[(10.1)]{GMX}. By definition, if $k$ is the order of the
separation $(A, B)$, then $(A, B)$ distinguishes exactly the tangles of order
$> k$ in which $B$ is the big side from the tangles of order $> k$
where $A$ is the big side. In the proofs that follow we are often
given a specific separation $(A, B)$ and want to know which pairs of
tangles it distinguishes. In these cases the following notation turns
out to be very convenient, even though at a first glance it may look
counter-intuitive. 

\begin{notation}
   Let $X = (A, B)$ be a separation of order $k$ of a digraph $G$. A
   tangle $T$ in $G$ of order $>k$ 
   is \emph{contained in $B$}, or just \emph{in $B$}, denoted $T \in B$, if $\beta_T(X) = B$, 
   otherwise it is contained in $A$.
   Note that the order of
   $T$ must be $> k$, as otherwise $T$ does not contain any of $(A,
   B)$ or
   $(B, A)$. 

   Let $X$ and $Y$ be separations of order $k$.
   Then $Z = (X^+ \cap Y^+, X^- \cup Y^-)$ is also a
   separation of $G$. Let $k'$ be the order of $Z$. Suppose  $T$ is a tangle
   of order $> k, k'$ contained in $X^+$ and $Y^+$. Then, by
   \cref{lem:tangle-tech},  $T \in X^+ \cap Y^+$.

   This notation only
   makes sense if the order of $T$ is bigger than the order of $(X^+
   \cap Y^+, X^- \cup Y^-)$ and we will only apply it in cases where this is
   guaranteed.

   In the same way it is easily seen that if $T \in X^-$ and $T \in Y^-$
   then $T \in X^- \cap Y^-$.
 \end{notation}
Using this notation we can simply say that $(A, B)$
distinguishes every tangle  in $A$ from every tangle
 in $B$.

Our next goal is to prove a directed analogue to the tangle
decomposition theorem for undirected graphs proved in \cite{GMX}. See
also \cite{Reed97}. Essentially, this result states that there is a
tree-decomposition whose bags correspond to the maximal tangles in $G$
and whose adhesion sets form minimal tangle-distinguishers. Recall
that in an undirected tree-decomposition $(T, \beta)$, if $e = \{ u, v
\} \in E(T)$ and $\beta(u), \beta(v)$ are the bags at $u$ and 
$v$, then the adhesion of $e$ is $S = \beta(u) \cap \beta(v)$. Furthermore, if
$A$ is the union of all bags in the component of $T-e$ containing $u$
and $B$ is the union of all bags in the component of $T-e$ containing
$v$, then $(A, B)$ is a separation of $G$ with separator $S$. Thus
every edge of an undirected tree-decomposition yields a separation of
$G$ and it is not very hard to show that the set of these separations
for any fixed tree-decomposition is laminar. Conversely, every laminar
set of separations yields a tree-decomposition.

However, in general this is not true for directed tree-decompositions  
where the labels of edges do not induce directed separations. In fact,
one can construct examples showing that 
a strict generalisation of the tangle decomposition theorem
for undirected graphs does not hold in the directed setting. 


We therefore prove a slightly different tangle decomposition theorem.


\section{Brambles}
\label{sec:brambles}
Directed tree-width has a natural duality, or obstruction, in terms of
directed brambles (see \cite{Reed97,Reed99,KreutzerK18}). As we will
see, there also is an intimate connection between brambles and
tangles, which we establish in this section.

\begin{definition}\label{def:bramble}
  Let $G$ be a digraph. A \emph{bramble} in $G$ is a set $\BBB$ of
  strongly connected subgraphs $B \subseteq G$ such that $B \cap B'\not=\emptyset$ for all $B,
  B'\in\BBB$.

  A \emph{cover of $\BBB$} is a set $X \subseteq V(G)$ of vertices such
  that $V(B) \cap X \not= \emptyset$ for all $B \in \BBB$. Finally,
  the \emph{order} of a bramble is the minimum size of a cover of
  $\BBB$. The \emph{bramble number} $\bn(G)$ of $G$ is the maximum
  order of a bramble in $G$.
\end{definition}

\begin{remark}
  In the original definition of brambles in \cite{Reed97,Reed99} it
  was only required that any two bramble elements $B, B'$ either share a
  vertex or that there are edges $e, e'$ such that $B \cup B' + e + e'$ is
  strongly connected. That is, $e, e'$ link $B$ and $B'$ both ways.
  The two definitions are not equivalent but the respective bramble
  numbers are within a constant factor (factor $2$) of each other.
  We find our definition more convenient to work with as, e.g., in the
  original definition taking a butterfly minor can increase the
  bramble number, which is not possible in our definition.
  See \cite{KreutzerK18, Kreutzer13} for a detailed exposition. 
\end{remark}


The next lemma is mentioned in \cite{Reed99}, see e.g.~\cite{KreutzerK18,Kreutzer13} for a proof.
\begin{lemma}\label{lem:bramble=dtw}
  There are constants $c, c'$ such that
  for all digraphs $G$, $\bn(G) \leq c\cdot \dtw(G) \leq c'\cdot \bn(G)$.
\end{lemma}

We now recall some simple properties of brambles needed later on.

\begin{lemma}\label{lem:bramble-on-one-side-of-sep}
  Let $\BBB$ be a bramble of order $k$ and let $X = (A, B)$ be a directed
  separation of $G$ of order $l < k$. Then there is exactly one $D \in \{
  A, B \}$ such that there is a $C \in \BBB$ with $C \subseteq D \setminus \sep(X)$. 
\end{lemma}
\begin{proof}
  Let $S := \sep(X)$.  As $k > l = |S|$ there must be bramble elements
  $C_1, \dots, C_t \in \BBB$ such that $V(C_i) \cap S = \emptyset$. By definition,
  bramble elements are strongly connected. This implies that $C_i\subseteq A\setminus
  B$ or $C_i \subseteq B \setminus A$ for all $1 \leq i \leq t$. But as  $C_i$ and
  $C_j$ intersect for all $1 \leq i, j \leq t$, there cannot be $i, j$ such that $C_i \subseteq A - S$
  and $V(C_j) \subseteq B - S$.   
\end{proof}

Note that the previous result implies the following simple
observation. If $\BBB$ is a bramble of order $k$ and $S$ is the separator
of a separation $(A, B)$ of order $< k$, then there is exactly one
strong component of $G - S$ containing an element of $\BBB$.

\begin{notation}\label{not:bramble-subset-A}
  Let $(A, B)$ be a separation in $G$ of order $l < k$ and let $\BBB$ be a
  bramble of order $k$. If $B$ is the uniquely determined side of  $(A, B)$
  containing an element of $\BBB$, we say that $\BBB$ is \emph{contained in $B$}, denoted $\BBB \subseteq B$.
\end{notation}

\begin{definition}[Distinguishing brambles]\label{def:dist-brambles}
  A separation $(A, B)$ \emph{distinguishes} brambles $\BBB, \BBB'$ in $G$
  if $\BBB$ and $\BBB'$ are contained in different sides of $(A, B)$,
  i.e.~$\BBB \subseteq A$ and $\BBB' \subseteq B$ or vice versa.
\end{definition}
The next corollary states a consequence of the previous lemma that we will use below.

\begin{corollary}\label{cor:min-sep-uncross-top-and-bottom}
  Let $\BB$ be a set of brambles in a digraph $G$ and let $k$ be the
  minimum order of a separation distinguishing two brambles in $\BB$. 
  Let $X$ and $Y$ be separations in $G$ of order $k$. Let $T :=
  \septop(X, Y), B := \sepbot(X, Y)$ and $S_t, S_l$ the top and bottom
  corner of $(X, Y)$.

  If there are $C, C' \in \BB$ such that $C$ is contained in $T$ and
  $C'$ is contained in $B$, then $S_t$ and $S_l$ both have order $k$
  and distinguish $C$ and $C'$. 
\end{corollary}

We now establish a close relation between tangles and brambles.

\begin{lemma}   \label{lem:bramble->tangle}
  If $\BBB$ is a bramble of order $k$ in a digraph $G$, then $G$ contains a
  tangle of order $\lfloor\frac k3\rfloor$.
\end{lemma}
\begin{proof}
  By \cref{lem:bramble-on-one-side-of-sep}, for every set $X$ of fewer
  than $\frac k3$ vertices there is a unique strong component of
  $G-X$ that contains a bramble element $B \in \BBB$. For every such $X$ we
  denote this component by $C(X)$.  
  We claim that
  \[
      \TTT  = \TTT(\BBB) :=  \{ (A, B) \sth (A, B) \text{ is a separation of  $G$
        of order } <\frac k3 \text{ such that } C(\sep(A, B)) \subseteq B \}
    \]
    is a tangle  of order $\frac k3$. Condition $(1)$ of
    \cref{def:tangle} follows immediately from
    \cref{lem:bramble-on-one-side-of-sep}. To show Condition $(2)$,
    let $(A_1, B_1), (A_2, B_2), (A_3, B_3) \in \TTT$ and let $S_i :=
    \sep((A_i, B_i))$, for all $1 \leq i \leq 3$. Then $B_i$ contains the
    unique component $C_i = C(S_i)$ of $G - S_i$ which contains a
    bramble element. But $S := S_1 \cup  S_2 \cup  S_3$ contains fewer than
    $k$ elements and hence it is not a cover of $\BBB$. It follows that
    there is a bramble element $B \in \BBB$ with $V(B) \cap S = \emptyset$. Thus $B \subseteq
    C_i$ for all $i$ and therefore $V(A_1\cup A_2\cup A_3) \neq V(G)$. 
\end{proof}

We next address the opposite direction, i.e.~that a tangle induces a
bramble. First we prove a technical lemma which will be used
frequently in the sequel. 

\begin{lemma}\label{lem:fixed-sep-tangles-unique-scc}
  Let $G$ be a digraph and let $\TTT$ be a tangle in $G$ of order
  $k$. For every set $X \subseteq V(G)$ with $|X| < k$ there is exactly one 
  strong component $C(X)$ of $G-X$ such that $C(X) \subseteq B$ for all $(A,
  B) \in \TTT$ with $\sep(A, B) = X$.  
\end{lemma}
\begin{proof}
  If $G - X$ is strongly connected then we set $C(X) := G -
  X$. Clearly $C(X)$ must be contained in the big side of every
  separation with separator $X$.

  Thus we may assume that $G - X$ is not strongly connected. 
  Let $\DDD = D(G, X)$ be the component dag of $G - X$ (see
  Section~\ref{subsec:digraphs-dags} and recall the notation introduced there).


  For every downwards closed set $S \subseteq V(\DDD)$ let $D(S) := G[X \cup \bigcup S]$ and
  $U(S) := G[X \cup \bigcup V(\DDD) \setminus S])$. Then $(D(S), U(S))$ is a separation in
  $G$ with separator $X$. If $S$ is downwards closed such
  that $D(S)$ is the small side of $(D(S), U(S))$, then for every
  downwards closed $S' \subseteq S$, $D(S')$ must also be the small side of
  $(D(S'), U(S'))$. Otherwise, $U(S')$ and $D(S)$ are two small sides
  whose union spans $V(G)$.

  Let $S$ be an inclusion-wise maximal downwards closed
  set such that $D(S)$ is the small side of $(D(S), U(S))$. This
  exists as for $S = \emptyset$, $D(\emptyset)$ must be the small side of $(D(\emptyset),
  U(\emptyset))$. We claim that there is exactly one component $C \in V(\DDD)$ with
  $N^+_\DDD(C) \subseteq S$. Towards this aim, we first observe that $S \neq V(\DDD)$, as
  otherwise $U(S) = G[X]$ and thus $D(S)$ as a small side would span
  $V(G)$. Thus $V(\DDD) \setminus S \neq \emptyset$. Let $M \subseteq V(\DDD) \setminus S$ be the elements of  
  $\DDD$ of out-degree $0$ in   $\DDD \setminus S$. By construction, for every $C \in
  M$ the set $S' := S \cup \{ C \}$ is downwards closed but $D(S')$
  is the big side of $(D(S'), U(S'))$.
  This implies that $|M| < 2$. To see this, suppose there were $C_1 \neq
  C_2 \in M$. Let $S_1 := S \cup \{ C_1 \}$ and $S_2 := S \cup \{ C_2 \}$. Then
  $V(U(S_1) \cup U(S_2) \cup D(S)) = V(G)$ and as $D(S)$ is a small side,
  $U(S_1)$ and $U(S_2)$ cannot both be small sides. Thus $|M| < 2$. On
  the other hand, every finite DAG has a sink and therefore $|M| = 1$.
  Let $C$ be the unique element of $M$ and let $S' := S \cup \{ C \}$.  As
  $D(S')$ is the big side of $(D(S'), U(S'))$, $U(S')$ is small. But
  $V(D(S) \cup U(S')) = V(G) \setminus V(C)$. Thus if there were a separation $(A,
  B) \in \TTT$ with $\sep(A, B) = X$ and $C \subseteq A$ for the small side $A$ of
  $(A, B)$, then $V(U(S') \cup D(S) \cup A) = V(G)$ would be the union of three
  small sides. This implies that $C$ has to be contained in the big
  side of every separation with separator $X$. 
\end{proof}

The next lemma complements Lemma~\ref{lem:bramble->tangle}. 

\begin{lemma}\label{lem:tangle->bramble}
  Let $G$ be a digraph and let $\TTT$ be a tangle of order $k$. Then $G$
  has a bramble of order $k$.   
\end{lemma}
\begin{proof}
  For every $X \subseteq V(G)$ with $|X| < k$ let   $C(X)$ be  the strong
  component specified in
  Lemma~\ref{lem:fixed-sep-tangles-unique-scc}. 
  We claim that $\{ C(X) \sth |X| < k \}$ forms a bramble of order
  $\geq k$.

  For every set $X$ with $|X|<k$ let
  $A_\downarrow(X)$ be the minimal downwards closed and
  $A_\uparrow(X)$ be the minimal upwards closed set in $\DDD(G, X)$
  containing $C(X)$. Furthermore, we define $A_\downarrow^-(X) := A_\downarrow(X) -
  C(X)$ and  $A_\uparrow^-(X) := A_\uparrow(X) - C(X)$. 

  To simplify notation we define $V(A) := \bigcup\{ V(D) \sth D \in A\}$ for any
  set $A \subseteq  V(\DDD)$. Furthermore, if $U \subseteq V(G)$ then $\overline{U} =
  V(G) \setminus U$. Finally, we will 
  simply write $(X \cup V(A), X \cup \bar{V(A)})$ for the separation $(G[X \cup V(A)], G[X \cup
  \bar{V(A)}])$  (see the discussion following Definition~\ref{def:dir-separation}).  

  First, observe that if $C(X)$ is the strong component of $G - X$
  obtained through Lemma~\ref{lem:fixed-sep-tangles-unique-scc}, then
  $|V(C(X))| \geq k$. For otherwise, let  $X_1 = (A_1, B_1)$ be any separation of order $<k$ such that $X \subseteq \sep(X_1)$. But then
  $A_2 := X \cup V(G) \setminus A_\downarrow(X)$ and $A_3 := X \cup V(G) \setminus A_\uparrow(X)$ are the
  small sides of separations in $\TTT$ with separator $X$ and $A_1 \cup A_2
  \cup A_3 = V(G)$, contradicting the tangle axioms. Furthermore, if $G$
  contains a tangle of order $k$ then $|V(G)| > 3(k-1)$ 
  as otherwise we could choose three sets $S_1, S_2, S_3$ of order
  $<k$ which together span $V(G)$ and consider arbitrary separations
  with separators $S_1, S_2, S_3$, resp. Then the union of the three
  small sides of these separations covers $V(G)$.  

  A consequence of these observations is that if $X \subseteq V(G)$ is of
  order $< k$ such that  $G - X$ is strongly connected, then $C(X)$
  has a non-empty intersection with  $C(X')$, for all $X' \subseteq V(G)$ with
  $|X'| < k$. 

  Thus it remains to prove that $C(X) \cap C(X') \neq \emptyset$ for sets $X, X'$
  such that $G-X$ and $G-X'$ are not strongly connected. Towards a 
  contradiction, suppose there were sets $X, X' \subseteq V(G)$ with
  $|X|, |X'| < k$ and $C(X) \cap C(X') = \emptyset$.

  Let $\DDD := \DDD(G, X)$ and $\DDD' := \DDD(G, X')$ be the
  component dags of $X, X'$, resp.


  Suppose first that one of $C(X), C(X')$, say $C(X)$, is the unique source of
  $\DDD, \DDD'$, resp. Then $(S, B) = (X \cup \overline{V(C(X))}, X \cup V(C(X))) \in \TTT$.
  If $C(X')$ is also the unique source of $\DDD'$ then $(X' \cup \overline{V(C(X'))}, X' \cup
  V(C(X'))) \in \TTT$. But then $C(X) \cap C(X') \neq \emptyset$ 
  as otherwise the two small sides cover $V(G)$. If $C(X')$ is not the
  unique source of $\DDD'$, then $(S_1, B_1) = (X' \cup V(A^-_\downarrow(X')), X' \cup
  \overline(V(A^-_\downarrow(X')))) \in \TTT$ and $(S_2, B_2) = (X' \cup
  \overline{V(A_\downarrow(X'))}, X' \cup V(A_\downarrow(X'))\}$. But then $S_1 \cup S_2 = V(G) \setminus V(C(X'))$.
  Thus, if $V(C(X)) \cap V(C(X')) = \emptyset$, then $S_1 \cup S_2 \cup S = V(G)$, a
  contradiction.

  Now suppose that none of $C(X), C(X')$ is the unique source of $\DDD, \DDD'$,
  resp.
  Again we have
  \[
    \begin{array}{rcl@{\hspace*{1pt}}ll}
      (S_1, B_1) & = &(X \cup \overline{V(A_\downarrow(X))}, & X \cup V(A_\downarrow(X))) & \in \TTT\\
      (S_2, B_2) & = & (X \cup V(A_\downarrow(X) \setminus \{ C(X)\}), & X \cup \overline{V(A_\downarrow(X) \setminus \{ C(X)\})}) & \in \TTT\\ 
      (S_3, B_3) & = & (X' \cup \overline{V(A_\downarrow(X'))}, & X' \cup V(A_\downarrow(X'))) & \in \TTT\\
      (S_4, B_4) & = & (X' \cup V(A_\downarrow(X') \setminus \{ C(X')\}), & X \cup \overline{V(A_\downarrow(X') \setminus \{ C(X')\})}) & \in \TTT
    \end{array}
  \]

  Consider the pair $Y = (S_1, B_1)$ and
  $Y' = (S_3, B_3)$.
  Let $Q_T, Q_B$ be the top and bottom quadrant of  $(Y,
  Y')$ and let $X_T, X_B$ be the corner of $Q_T, Q_B$ resp.
  By Lemma~\ref{lem:submodularity}, $|X_B| < k$ or
  $|X_T|<k$. Suppose first that  $|X_B| < k$ and consider the
  separation  $(U, U') = (X_B \cup \overline{Q_B}, X_B \cup Q_B)$.
  If $(U', U) \in \TTT$ then $U' \cup S_1 \cup S_3 = V(G)$, a contradiction. Thus $(U, U') \in
  \TTT$. But $U' = Q_B = V(A_\downarrow(X)) \cap V(A_\downarrow(X')) \subseteq V(A_\downarrow(X) \setminus \{ C(X) \}) \cup V(A_\downarrow(X')
  \setminus \{ C(X')\})$ as $V(C(X)) \cap V(C(X')) = \emptyset$.  Thus, $U \cup S_2 \cup S_4 = V(G)$, a
  contradiction. 

  This implies that $|X_B| \geq k$ and therefore $|X_T| < k$. 
  Clearly, $(X_T \cup \overline{Q_T}, X_T \cup Q_T) \nin \TTT$ as otherwise
  $(X_T \cup \overline{Q_T}) \cup S_1 = V(G)$. Thus we have that $X_T \cup
  \overline{Q_T}$ is the big side of this separation. But then, for the
  two small sides $S_2, S_4$ we have
  \[
    \begin{array}{rcl}
    S_2 \cup S_4 & = &  X \cup X' \cup V(A_\downarrow(X))\setminus V(C(X)) \cup
                    V(A_\downarrow(X')) \setminus V(C(X'))\\
              &= & X \cup X' \cup (V(A_\downarrow(X)) \cup V(A_\downarrow(X'))) \setminus (V(C(X)) \cap V(C(X')))\\
      &=& X \cup X' \cup V(A_\downarrow(X)) \cup V(A_\downarrow(X')),
    \end{array}
  \]
  as $V(C(X)) \cap V(C(X')) = \emptyset$. Thus, $S_2 \cup S_4 \cup (X_T \cup Q_T) = V(G)$ is
  the union of three small sides, a contradiction.
  This proves the claim.
\end{proof}

The previous lemmas imply that every
tangle $T$ contains a bramble $\BBB$ in its big side. Note that this bramble will in general
not be unique. We say that $\BBB$ is  \emph{associated}  with the tangle and
conversely, every bramble $\BBB$ such that every $(A, B) \in  T$ contains at least one
element of $\BBB$ in its big side $B$ is said to be  \emph{controlled} by the tangle.

As a consequence of the previous lemmas we may think of a tangle in
terms of its associated bramble. This will be particularly useful in
the algorithmic sections below, as brambles are easier to work with
algorithmically than tangles. 

In the proof of Lemma~\ref{lem:tangle->bramble} we constructed for
every tangle a special type of bramble: for every separation $X$ of
order $< k$ we added a strong
component of $G - X$ contained in the big side of $X$ to the bramble. Brambles of this
form are called \emph{canonical}.

\begin{definition}[Canonical bramble]
  \label{def:bramble:canonical}
  A bramble $\BBB$ of order $k$ in $G$ is \emph{canonical} if for every
  set $X \subseteq V(G)$ of $< k$ elements, $\BBB$ contains exactly one strong
  component of $G - X$ and, conversely, every element of $\BBB$ is a strong
  component of $G - X$ for some set $X \subseteq V(G)$ with $< k$ vertices.
\end{definition}

\begin{remark}
  In the undirected setting one can show that there is a one-to-one
  correspondence between tangles and brambles with the additional
  property that the intersetion of any three bramble elements is
  non-empty. This is false in the directed setting as the 
  example in \cref{fig:no-canonical-bramble} illustrates. The digraph
  contains a tangle which orients the separations $( \{7,8,5, 6,4\},
  \{6,4,1,2,3,9\})$, $(\{7,8,1,6,2\},\{  6,2,5,4,3,9\})$ and
  $(\{7,8,6,3\},\{6,3,5,4,1,2,9\})$ towards the side containing $9$. Then
  the strong components constituting the bramble are $\{1,2,3\}$,
  $\{3,4,5\}$ and $\{ 4,5,1,2\}$, which pairwise share a vertex but have
  no vertex common to all.
\end{remark}

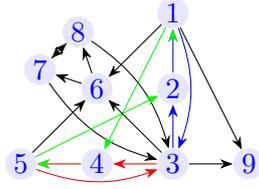
\begin{figure}
  \centering
\begin{tikzpicture}[every node/.style={circle, fill=blue!10, inner
    sep=1pt}, every edge/.style={black, ->, >=Stealth}]
  \foreach \x/\n in {{0,0}/5,{1,0}/4,{2,0}/3,{2,1}/2,{2,2}/1,{1,1}/6,{0.25,1.25}/7,{0.75,1.75}/8,{3,0}/9}
  \node (n\n) at (\x)  { $\n$ } ;
  \draw[->,black,>=Stealth,red]  (n5) to [bend right=20] (n3);
  \draw[->,black,>=Stealth,blue]  (n1) to [bend left=20] (n3);
  \draw[->,black,>=Stealth]  (n7) to [bend right=20] (n3);
  \draw[->,black,>=Stealth]  (n8) to [bend left=20] (n3);
  \foreach \a/\b in
  {3/4,4/5}
  \draw[->,black,>=Stealth,red]  (n\a) to (n\b);
  \foreach \a/\b in
  {3/2,2/1}
  \draw[->,black,>=Stealth,blue]  (n\a) to (n\b);
  \foreach \a/\b in
  {1/4,5/2}
  \draw[->,black,>=Stealth,green]  (n\a) to (n\b);
  \foreach \a/\b in
  {4/5,2/1}
  \draw[->,black,>=Stealth,green,dotted]  (n\a) to (n\b);
  \foreach \a/\b in
  {1/6,3/6,6/7,6/8,7/8,8/7,5/6,3/9,1/9}
  \draw[->,black,>=Stealth]  (n\a) to (n\b);
  \end{tikzpicture}
  
  \caption{Example of tangle with three components having an empty
    common intersection.}
  \label{fig:no-canonical-bramble}
\end{figure}

\section{Tangle Tree-Labellings}

In the remainder of this section we will use the following notation.
Let $G$ be a digraph and let $\TTT$ be a set of tangles in $G$. Let $\SSS(G)$
denote the set of all directed separations of $G$. For each $k \geq
0$, let $\SSS_k(G)$ be the set of directed separations of $G$ of order
$k$. Occasionally we will also use notation such as $\SSS_{<k}(G), \SSS_{\leq
  k}, \SSS_{> k}(G), \SSS_{\geq k}(G)$ defined in the obvious way, e.g.~$\SSS_{<k}(G) := \bigcup_{j<k} \SSS_j(G)$.

If $L$ is an
oriented tree and $e = (a, b) \in E(L)$ then we denote the side of $L
- e$ containing $a$ by $L_{e, a}$ and the side containing $b$ by
 $L_{e, b}$. 

 Let $L$ be a tree,  $\gamma \sth E(L) \rightarrow \SSS(G)$ be a function, and $k \geq 0$.
By a \emph{$k$-component} of $L$ we mean a component of $L[\{
e \sth \gamma(e) \in \SSS_{>k}\}]$.

\begin{definition}[Tangle Tree Labelling]\label{def:tangle-tree-labellings}
  A \emph{tree-labelling} of $\TTT$ is a triple $\LL \coloneqq (L, \beta, \gamma)$, where $L$ is
  a tree, $\beta \sth V(L) \rightarrow \TTT$ and $\gamma \sth E(L) \rightarrow \SSS(G)$ are functions
  such that 
  \begin{enumerate}
  \item $\beta$ is a bijection between $V(L)$ and $\TTT$, 
  \item if $t \neq t' \in V(L)$ and $P$ is the unique path in $L$ between
    $t$ and $t'$ then for every edge $e \in E(P)$ with $|\gamma(e)| = \min \{
    |\gamma(e')| \sth e' \in E(P) \}$ the separation $\gamma(e)$ is a minimum order
    distinguisher between $\beta(t)$ and $\beta(t')$,
  \item for every $e = \{ s, s'\} \in E(L)$ there are $t \in
    V(L_{e, s})$ and $t' \in V(L_{e, s'})$ such that $\gamma(e)$ is a minimum
    order distinguisher of $\beta(t)$ and $\beta(t')$.
  \end{enumerate}
  If  $e = (a, b) \in E(L)$ is an edge then we call the separation
  $\gamma(e) = (A, B)$ the \emph{separation at $e$} and the set $V(A) \cap
  V(B)$ the \emph{separator at $e$}.

  The \emph{order} of $e$ is the
  order $|\gamma(e)|$ of the separation at $e$.
  The \emph{order} of $\LL$ is the maximal order of any separation
  $\gamma(e)$ for any edge $e = (a, b) \in E(L)$.
\end{definition}

Our next goal is to show the following theorem. 

\begin{theorem}\label{thm:tangle-tree-labellings}
  Every set  $\TTT$ of tangles in a digraph $G$ has a $\TTT$-tree-labelling.
\end{theorem}

\begin{figure}
  \centering
  \beginpgfgraphicnamed{fig-no-uncross-a}\begin{tikzpicture}
   \path [fill=blue!10, circle, minimum size=1cm]
   node [fill=blue!10, circle] (A) at (5,6) { $T_1$  } 
   node [fill=blue!10, circle] (B)  at (3,4) { $T_2$  } 
   node [fill=blue!10, circle]  (C) at (7,4) { $T_3$  } 
   node [fill=blue!10, circle]  (D) at (5,2) { $T_4$  } 
;
    
    \foreach \x in {0, 1, 2, 3}
    {
     \node[fill=red!50,circle,inner sep=0.5pt] (l-\x) at
     (3.35+\x*0.25,2.5+\x*0.35) {\tiny r\x};
     \node[fill=blue!50,circle,inner sep=0.5pt] (r-\x) at
     (6.6-\x*0.25,2.5+\x*0.355) {\tiny b\x};
     \draw[->,>=Stealth,very thick] (D) to (l-\x) ;
     \draw[->,>=Stealth,very thick] (l-\x) to (B) ; 
     \draw[->,>=Stealth,very thick] (D) to (r-\x); 
     \draw[->,>=Stealth,very thick] (r-\x) to (C) ;
   } ;
   \node [fill=green,circle,inner sep=0.5pt] (u) at (5,4.75) {\footnotesize g};
   \begin{scope}[on background layer]
     \foreach \x in {0, 1, -1}
     {
       \draw [->,>=Stealth,very thick] ($(B)+(-0.125*\x,0.125*\x)$) to
       (u);
       \draw [->,>=Stealth,very thick] (u) to (A) ;
       \draw [->,>=Stealth,very thick] ($(C)+(0.125*\x,0.125*\x)$) to
       (u);
     };
     \draw [->,>=Stealth,line width=2pt, dotted, red] (A) to [bend right=20] (B) ;
     \draw [->,>=Stealth,line width=2pt, dotted, red] (A) to [bend left=20] (C) ;
     \draw [->,>=Stealth,line width=2pt, dotted, red] (B) to [bend right=60] (D) ;
     \draw [->,>=Stealth,line width=2pt, dotted, red] (C) to [bend left=60] (D) ;
     \draw [sepout,red!50,thick] (3,2) to (6.25,6.5) ;
     \draw [sepin,blue!50,thick] (7,2) to (3.75,6.5) ;
   \end{scope}
   \path (3,2) node [anchor=east] { $\sigma(T_2)$ } ;
   \path (7,2) node [anchor=west] { $\sigma(T_3)$ } ;
 \end{tikzpicture}\endpgfgraphicnamed
 \beginpgfgraphicnamed{fig-no-uncross-b}\begin{tikzpicture}[x=2cm,y=2cm]
   \tikzstyle{tangle}=[circle,fill=blue!10,inner sep=5pt]
   \node[tangle] (t4) at (1,0) {$T_4$} ;
   \node[tangle] (t3) at (2,1) {$T_3$} ;
   \node[tangle] (t2) at (0,1) {$T_2$} ;
   \node[tangle] (t1) at (1,2) {$T_1$} ;
   \draw (t3) -> (t4) (t4) -> (t2) (t2) -> (t1);
 \end{tikzpicture}\endpgfgraphicnamed{}
  \caption{No exact uncrossing of tangles of different order.}
  \label{fig:no-uncross}
\end{figure}
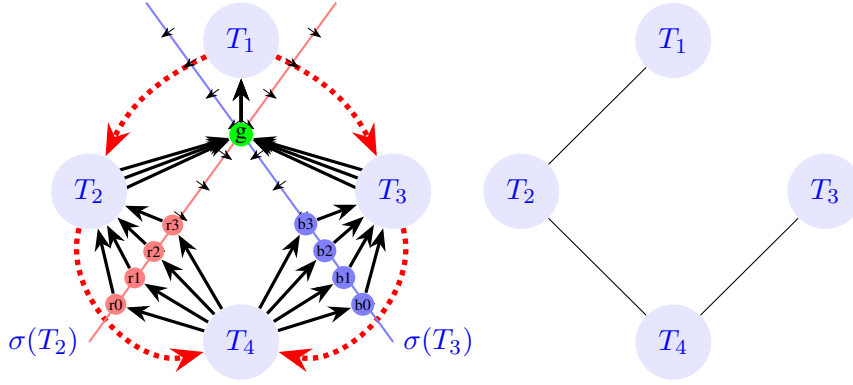

Before proving the theorem we first prove a special case of the result.

\subsection{Minimum-Order Tangle-Distinguishers}

\begin{theorem}\label{thm:order-k-tree-labellings}
  Let $\TTT$ be a set of tangles of order $> l$ in a digraph $G$ which
  are pairwise $l$-distinguishable but $(l-1)$-indistinguishable. Then
  there is a $\TTT$-tree-labelling $(L, \beta, \gamma)$ of $G$ such that $|\gamma(e)| = l$
  for all $e \in E(L)$.
\end{theorem}

For the rest of this section we fix  a set $\TTT$  of tangles of order $>
l$ which are pairwise $l$-distinguishable but  $(l{-}1)$-indistinguishable. 
  To prove the theorem we first establish some intermediate results.

\begin{lemma}\label{lem:pushing-to-one-side}
  Let $X$ be a separation of order $l$ distinguishing at least two
  tangles from $\TTT$ and let $A$ be one of its two sides.
  If there are at least two tangles in $\TTT$ contained in $A$ then
  there is a separation $X'$ of order $l$ and a side $S$ of $X'$ such
  that the set of tangles from $\TTT$ contained in $S$ is a subset of the
  set of tangles from $\TTT$ contained in $A$. 
\end{lemma}
\begin{proof}
  We show the claim for the case that $A = X^-$.  The other case is
  symmetric.  Let $B = X^+$. Let $\TTT_A$ be the set of tangles contained
  in $A$. Let $Y$ be a separation of order $l$ distinguishing  two tangles
  from $\TTT_A$. Such a separation exists as $A$ contains more than one
  tangle from $\TTT$ and any such pair of 
  tangles is $l$-distinguishable. Choose   $T, T' \in \TTT_A$ such that
  $Y^- = \beta_T(Y)$ and $Y^+ = \beta_{T'}(Y)$.

  Now consider the pair $(X, Y)$. By Lemma~\ref{lem:submodularity} $U
  = (X^+ \cap Y^+, X^- \cup Y^-)$ and $L = (X^+ \cup Y^+, X^- \cap Y^-)$ are
  separations and, by Lemma \ref{lem:tangle-tech}, if $|L| \leq l$ then $\beta_T(L)
  = L^- = X^- \cap Y^-$.  

  If $Y^+$ contains a tangle from $X^+$, then this implies that
  $X^+ \cap Y^+$ and $X^- \cap Y^-$ both contain a tangle from
  $\TTT$. As no pair of tangles from $\TTT$ is $l - 1$-distinguishable, by
  Lemma~\ref{lem:submodularity}, $U$ and $L$ both have order exactly $l$.
  Thus, in this case setting  $X' \coloneqq L$ satisfies the
  requirements of the lemma, as it contains 
  a strict subset of the tangles of $\TTT$ contained in $A$.

  We may therefore assume that $Y^+$ does not contain a tangle
  from $X^+$. But then the set of tangles contained in $Y^+$ is a
  strict subset of the set of tangles contained in $A$ and therefore
  setting $X' \coloneqq Y$ satisfies the requirements of the lemma.
\end{proof}

The previous lemma has the following simple consequence.

\begin{corollary} \label{cor:sep-splitting-off-one-tangle}
  There is a tangle $T \in \TTT$, a separation $X$
  of order $l$ and a side $B$ of $X$ such that $\beta_T(X) = B$ but
  $\beta_T(T') \neq B$ for all $T' \in \TTT \setminus \{ T \}$.
\end{corollary}
\begin{proof}
  Choose among all separations of order $l$ a separation $X$ such that
  the number $t$ of tangles from $\TTT$ contained in one side $B$ of $X$ is
  minimised.   Towards a contradiction,
  suppose $t > 1$. Then, by Lemma~\ref{lem:pushing-to-one-side}, there is a separation
  $X'$ and a side $B'$ of $X'$ such that the set of tangles from
  $\TTT$ contained in $B'$ is a strict subset of the tangles contained in $B$,
  contradicting the choice of $B$.
\end{proof}

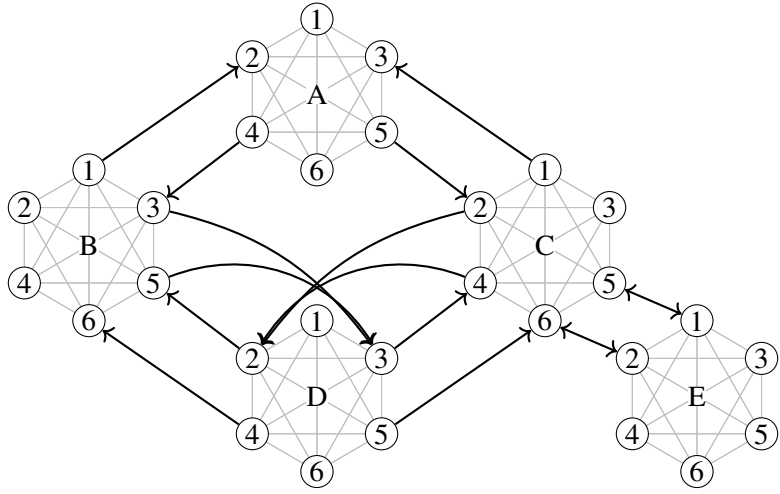
\begin{figure}
  \centering
  \beginpgfgraphicnamed{fig-ranks}\begin{tikzpicture}[x=1cm,y=-1cm]
    \tangcirc[A]{4}{1} 
    \tangcirc[B]{1}{3} 
    \tangcirc[C]{7}{3} 
    \tangcirc[D]{4}{5}
    \tangcirc[E]{9}{5}

    \draw[draw=black,thick,->] (TblA) -- (TtrB) ;
    \draw[draw=black,thick,->] (TbrA) -- (TtlC) ;
    \draw[draw=black,thick,->] (TbrB) to [bend left = 40] (TtrD) ;
    \draw[draw=black,thick,->] (TtrB) to [bend left = 20] (TtrD) ;
    \draw[draw=black,thick,->] (TtlD) to  (TbrB) ;
    \draw[draw=black,thick,->] (TblD) to  (TbB) ;
    \draw[draw=black,thick,->] (TblC) to [bend right = 40] (TtlD) ;
    \draw[draw=black,thick,->] (TtlC) to [bend right = 20] (TtlD) ;
    \draw[draw=black,thick,->] (TtrD) to (TblC) ;
    \draw[draw=black,thick,->] (TbrD) to (TbC) ;

    \draw[draw=black,thick,->] (TtB) to (TtlA) ;
    \draw[draw=black,thick,->] (TtC) to (TtrA) ;
    
    \draw[draw=black,thick,<->] (TbC) to (TtlE) ;
    \draw[draw=black,thick,<->] (TbrC) to (TtE) ;

    
  \end{tikzpicture}\endpgfgraphicnamed{}
  \caption{a) A digraph $G$ and  $5$ tangles of order $3$. b) a $\TTT$-tree-labelling. }
  \label{fig:ranks}
\end{figure}

\paragraph*{Ranks.}
 
We inductively construct disjoints sets $\CCC_i \subseteq \TTT$, $i\geq 0$, and a
function $\sigma$ assigning to every tangle $T$ in $\CCC_i$ a separation
$\sigma(T)$ of order exactly $l$  as follows. 
\begin{itemize}
\item Let $\CCC_0 \coloneqq \emptyset$.
\item Assume that $\CCC_i$ has already been defined and
  let $\CCC'= \TTT \setminus \bigcup_{j \leq i} \CCC_j$. If
  $|\CCC'| \leq 1$, then we set $\CCC_{i+1} \coloneqq \CCC'$ and the construction stops
  here. Otherwise, let $\CCC_{i+1} \subseteq \CCC'$ be the set of tangles
  $T$ such that there exists a separation $(A, B)$ of order 
  $l$ such that $T$ is the unique tangle from $\CCC'$ containing $(A, B)$.
  By Corollary  \ref{cor:sep-splitting-off-one-tangle}, such a tangle
  $T$ exists. For each $T \in \CCC_{i+1}$ we choose among all separations of
  order $l$ distinguishing $T$ from all other tangles in $\CCC'$ a
  separation 
  $\sigma(T) \coloneqq (A, B) \in T$ such that
  $V(B)$ is inclusion-wise minimal.
 This completes the construction of the set $\CCC_{i+1}$.
\end{itemize}
By construction, the sets $\CCC_i, i \geq 0$ are pairwise disjoint and therefore the
construction stops at some level $m$. We say that the tangles
contained in $\CCC_i$ are the tangles of \emph{rank} $i$ or \emph{at level} $i$.

To simplify the presentation, we agree on the following notation.
Let $T \in \TTT$ be a tangle. By  $B(T)$ and $A(T)$ we
denote the two sides of $\sigma(T)$ such that $(A(T), B(T)) \in T$. 
If $B(T) = \sigma(T)^+$ we say
that \emph{$T$ orients $\sigma(T)$ outwards}, otherwise  $T$
orients $\sigma(T)$ \emph{inwards}. 
In cases where the tangle $T$ is understood, especially when we speak
about the orientation of some $\sigma(T)$, we simplify the notation further
and simply say that  $\sigma(T)$ is
\emph{outgoing} or \emph{incoming}.

If $T' \neq T$ is a tangle in $\TTT$ then we write $T' \in B(T)$ if $(A(T), B(T)) \in T'$.
If $S \subseteq \TTT$ is a set of tangles, we will adopt similar notation and write, e.g.,
$S \subseteq B(T)$ to say that each tangle in $S$ contains $(A(T), B(T))$.

\begin{example}
  Consider the digraph depicted in Figure~\ref{fig:ranks}. The graph
  consists of $5$ bidirected cliques $A, B, C, D, E$ on six vertices each connected as
  indicated in the figure. We will refer to the vertices in the clique
  $A$ by $a_1, \ldots, a_6$, to those in $B$ by $b_1, \ldots, b_6$ etc., where the
  number corresponds to the label of the vertex in the figure.

  Each of the five cliques induces a tangle of order three. For instance,
  the clique $B$ defines the tangle $T_A$ as follows: if $(X, Y)$ is a
  separation of order $\leq 2$ in $G$ then $T_A$ contains $(X, Y)$ if $Y$
  contains at least $4$ vertices of $A$ and otherwise $T_A$ contains
  $(Y, X)$.

  We define separations $\sigma(T_A), \ldots, \sigma(T_E)$ as follows:
  \begin{itemize}
  \item
     $\sigma(T_E) = (G[A \cup B \cup C \cup D \cup \{ e_1, e_2 \}], E)$,
\item     $\sigma(T_D) = ( G[A \cup B \cup C \cup D \cup \{ d_2, d_3 \}], D)$,
\item     $\sigma(T_C) = (G[A \cup B \cup \{ d_3, c_2 \}], C \cup D \cup E)$,
\item     $\sigma(T_B) = (G[A \cup C \cup E \cup \{ b_3, d_2 \}], B \cup D)$, and
\item     $\sigma(T_A) = (G[ B \cup C \cup D \cup E \cup \{ a_4, a_5 \}], A)$.
  \end{itemize}
  
  Then $\sigma(T_A), \sigma(T_D)$, and $\sigma(T_E)$ distinguish $T_A$, $T_D$, and $T_E$,
  resp., from all other tangles and therefore $T_A, T_D$, and $T_E$
  are tangles of rank $1$.
  The tangles $T_B$ and $T_C$ are of rank $2$. For instance, there are
  no separations of order two distinguishing $T_C$ from all other tangles.

  Finally, $\sigma(T_A)$ is incoming whereas the other four separations
  are outgoing.
\end{example}

We now prove some simple properties of the set $\{ \sigma(T) \sth T \in \TTT \}$ of separations.

\begin{definition}\label{def:tangle:dependent}
  Let $T \neq T' \in \TTT$ and let $\sigma(T) = (A, B) \in T$ and $\sigma(T') = (A', B') \in
  T'$. 
  $T'$ is a \emph{descendent of $T$} if $(A, B) \in T'$ but $(B', A') \in T$.  $T$
  and $T'$ are \emph{dependent} if $T$ is a descendant of $T'$ or $T'$
  is a descendant of $T$. Otherwise they are \emph{independent}.
  We call a set  $\SSS \subseteq \TTT$  \emph{independent} if the tangles in $\SSS$ are pairwise independent.
\end{definition}

Note that if $T'$ is a descendant of $T$ then the rank of $T$ must be
higher than the rank of $T'$.

In the example above, the tangles $T_E$ and $T_C$ are
dependent as $T_E \in B(T_C)$. Similarly, $T_D$ and $T_C$ are dependent as
are $T_D$ and $T_A$. 

\begin{lemma}\label{lem:opposite-orientations-uncross-independent}
  Let $X = (A, B)$ and $X' = (A', B')$ be directed separations of
  order $l$ and let $T, T'$ be tangles such that $(A, B) \in T \setminus T'$ and
  $(A', B') \in T' \setminus T$. If $B = X^+$ and $B' = X'^-$ then $X$ and $X'$
  are uncrossed or there are separations $X_1 = (A_1, B_1) \in T\setminus T'$
  and $X_2 = (A_2, B_2) \in
  T' \setminus T$ of
  order $l$ such that $B_1 \subsetneq B$, $B_2 \subsetneq B'$ (and thus $A \subsetneq A_1$ and
  $A' \subsetneq A_2$).
\end{lemma}
\begin{proof}
  By assumption, $(A, B), (B', A') \in T$ and $(B, A), (A', B') \in T'$.
  By~\cref{lem:tangle-tech}, if $U = (U_1, U_2) = (A \cup B', B \cap A')$ is of order $\leq l$, then $(U_1,
  U_2) \in T$ whereas $(U_2, U_1) \in T'$. Similarly, if $L = (L_1, L_2) = (A \cap B', B \cup A')$ is of order $\leq l$, then
  $(L_1, L_2) \in T'$ and $(L_2, L_1) \in T$. But
  by~\cref{lem:submodularity}, $|L| \leq l$ or $|U| \leq l$ and as neither
  can be strictly smaller than $l$, $|L| = |U| = l$. But then $U$ and
  $L$ distinguish $T$ and $T'$ and we can set $X_1 = U$ and $X_2 = L$.
\end{proof}

\begin{corollary}\label{cor:opposite-orientations-uncross-independent}
  Let  $T \neq T' \in
  \TTT $ be independent tangles such that $\sigma(T)$ is outgoing and
  $\sigma(T')$ is incoming. Then $\sigma(T)$ and $\sigma(T')$ are uncrossed.
\end{corollary}
\begin{proof}
  Towards a contradiction suppose that $T, T' \in \TTT$, 
  $\sigma(T) = (A, B) \in T$ and $\sigma(T') = (A', B') \in T'$ cross, and $\sigma(T)$ is outgoing whereas
  $\sigma(T')$ is incoming. See \cref{fig:diff-orient-uncross} a). Thus, $B = \sigma(T)^+$ and $B' = \sigma(T')^-$.
  Furthermore, as $T$ and $T'$ are independent, $(B, A) \in T'$ and $(B', A') \in T$. 

  Let $U = (B \cap A', A \cup B')$ be the upper corner of the pair $(\sigma(T), \sigma(T'))$ and $L =
  (B \cup A', A \cap B')$ be the lower corner.  See
  \cref{fig:diff-orient-uncross} a). $U$ and $L$ both distinguish
  between $T$ and $T'$ and therefore must both be of order $l$.
  But this is a contradiction to the minimality of $\sigma(T)$ and $\sigma(T')$.
  \begin{figure}[t]
    \centering
    \beginpgfgraphicnamed{fig-diff-orient-uncross-a}\begin{tikzpicture}[x=.5cm,y=.5cm]
      \draw[fill=orange!10,draw=none] (-3,3) rectangle (0,0) ;
      \draw[fill=blue!10,draw=none] (3,-3) rectangle (0,0) ; \node at
      (-2,2) { $\black{}U$ } ; \node at (-1,1) { $T$ } ; \node at
      (2,-2) { $\black{}L$ } ; \node at (1,-1) { $T'$ } ;
      \draw[sepout] (-3,0) -- (3,0) ;
      \draw[sepin] (0,3) -- (0,-3) ;
      \node[anchor=west] at (3.5,0) { $\sigma(T)$ };
      \node at (0,-3.5) {  $\sigma(T')$ };
      \node at (0,-5) { a) } ;
      \node[anchor=west] at (0,3.3) { $B'$ } ;
      \node[anchor=east] at (-0,3.3) { $A'$ } ;
      \node[anchor=south] at (3.1,0) { $B$ } ;
      \node[anchor=north] at (3.1,0) { $A$ } ;
    \end{tikzpicture}\endpgfgraphicnamed{}\hspace*{1cm}
    \beginpgfgraphicnamed{fig-diff-orient-uncross-b}\begin{tikzpicture}[x=.5cm,y=.5cm]
      \draw[fill=orange!10,draw=none] (-3,3) rectangle (0,0) ;
      \draw[fill=blue!10,draw=none] (3,-3) rectangle (0,0) ;
      \node at  (-2,2) { $\black{}U$ } ;
      \node at (1,1) { $T$ } ;
      \node at  (2,-2) { $\black{}L$ } ;
      \node at (1,-1) { $T'$ } ;
      \node at (-1.5,-1.5) { $T_o$ } ;
      \draw[sepout] (-3,0) -- (3,0) ; \draw[sepin] (0,3) -- (0,-3) ;
      \node[anchor=west] at (3.5,0) { $\sigma(T)$ }; \node at (0,-3.5) {
        $\sigma(T')$ };
      \node at (0,-5) { b) } ;
    \end{tikzpicture}\endpgfgraphicnamed{}
    \caption{Outgoing and incoming tangle-distinguishers uncross. a)
      $T$ and $T'$ are independent. b) $T$ depends on $T'$.}
    \label{fig:diff-orient-uncross}
  \end{figure}
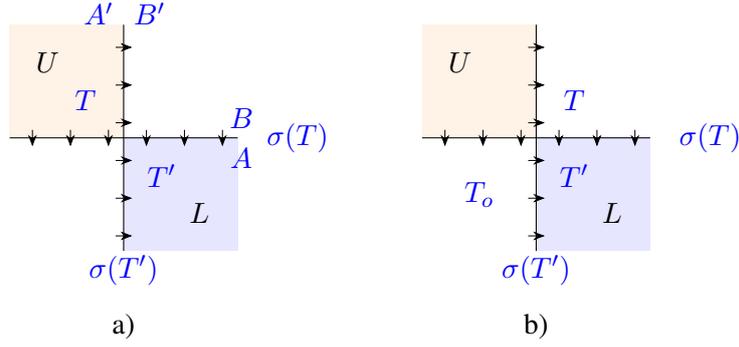
\end{proof}

Note that if $T$ and $T'$ as in the previous lemma are dependent then
$\sigma(T)$ and $\sigma(T')$ may not be uncrossed, see
\cref{fig:diff-orient-uncross} b).

For the following exposition we choose a tangle $T_o$ of maximal rank $r$
and set $\sigma(T_o) := (\emptyset, V(G))$. Furthermore, we declare the rank of
$T_o$ to be $r+1$, i.e.~$T_o$ now has the highest rank among all
tangles. Clearly every other tangle is a descendant of 
$T_o$. 

Let $D$ be the digraph with vertex set $\TTT$ and edges
$E(D) := \{(T, T') \sth T' \text{\black{} is a descendant of } T \}$.

It is easily seen that $D$ is acyclic as edges only point from tangles
of higher rank to tangles of lower rank. Furthermore, by our choice of
$T_o$ above, $T_o$ is the sole tangle with an edge to every other
tangle. By construction, if $T \neq T_o$, then $\sigma(T)$ distinguishes $T$
and $T_o$. 

We show next that $E(D)$ is transitive.

\begin{lemma}[Transitivity]\label{lem:transitivity}
  Let $T_1, T_2, T_3 \in \TTT$ be tangles. 
  If $(T_1, T_2) \in E(D)$ and $(T_2, T_3) \in E(D)$ then $(T_1, T_3) \in
  E(D)$.
\end{lemma}
\begin{proof}
  For $i = 1, 2, 3$ let $X_i \coloneqq \sigma(T_i) = (A_i, B_i) \in T_i$.
  By assumption of the lemma, $(A_1, B_1) \in T_2$ and $(A_2, B_2)
  \in T_3 \setminus T_1$. Furthermore, $(B_i, A_i) \in T_o$ for all $1 \leq i
  \leq 3$.
  We need to show that $(B_3, A_3) \in T_1$ and $(A_1, B_1) \in T_3$.

  $(B_3, A_3) \in T_1$ follows from the fact that $T_1$ is of higher
  rank than $T_3$. Thus it remains to show that $(A_1, B_1) \in T_3$.

  Towards a contradiction, suppose that $(B_1, A_1) \in
  T_3$. If $B_1 = X_1^+$ and $B_2 = X_2^+$ then,
  by~\cref{lem:tangle-tech}, $U = (A_1 \cup A_2, B_1 \cap B_2) \in T_2$ unless
  $|U| > l$ whereas
  $L = (B_1 \cup B_2, A_1 \cap A_2) \in T_o$ unless $|L| > l$. But one of
  $|U|,|L|$ must be of order at most $l$ and therefore, as  $U$ and $L$ both
  distinguish tangles from $\TTT$, both are of order $l$. Thus,
  by choice of $X_2$ (minimality of $B_2$), we may assume that $(A_2, B_2) = 
  (A_1 \cup A_2, B_1 \cap B_2)$ and therefore $(A_1 \cup A_2, B_1 \cap B_2) \in
  T_3$. But then, $(B_1, A_1)$ cannot be in $T_3$ as otherwise $B_1$
  and $A_2$ would be two small sides covering all of $V(G)$.

  Now suppose $B_2 = X_2^-$. Again suppose that $(B_1, A_1) \in T_3$.
  Consider the pair $(X_1, X_2)$ and the upper and lower separations
  $U = (A_1 \cup B_2, B_1 \cap A_2)$ and $L = (A_2 \cup B_1, B_2 \cap A_1)$. 
  By~\cref{lem:tangle-tech}, if $U$ is of order $\leq l$, then as
  $(A_1, B_1) \in T_1$ and $(B_2, A_2) \in 
  T_1$, also $(A_1 \cup B_2, B_1 \cap A_2) \in T_1$. Similarly, if $|L| \leq l$
  then $(A_1 \cup B_1, B_2 \cap A_1) \in T_3$. Thus $U$ and $L$ distinguish
  $T_1$ and $T_3$ and therefore must both be of order $l$. But then,
  $(A_1 \cup B_2, B_1 \cap A_2) \in T_1$ contradicts the minimality of $B_1$.

  The cases where $B_1 = X_1^-$ and $B_2 = X_2^+$ or  $B_1 = X_1^-$ and
  $B_2 = X_2^-$ are analogous.
  Thus, in either case, $(A_1, B_1) \in T_3$ and $T_3$ therefore depends
  on $T_1$.
\end{proof}

The digraph $D$ essentially is \emph{dag-labelling of $\TTT$}, i.e.~a dag distinguishing all tangles in $\TTT$.
What remains to be done is to turn $D$ into a (rooted) tree.
As a first step, let $D'$ be the digraph obtained from $D$ by
eliminating all transitive edges, i.e.~$E(D') := \{ e = (T, T') \in E(D) \sth $
there is no path from $T$ to $T'$ in $D - e \}$.
Furthermore, we define a function $\gamma \sth E(D') \rightarrow \SSS_l(G)$ where for all
$e = (T, T') \in E(D')$ we define $\gamma(e) = \sigma(T')$. 
Let \[
    \TTT^+  :=  \{ T \in \TTT \sth \beta_T(\sigma(T)) = \sigma(T)^+ \} \setminus T_o \quad\text{and}\quad
    \TTT^-  :=  \{ T \in \TTT \sth \beta_T(\sigma(T)) = \sigma(T)^- \} \setminus T_o. 
\] 

If no vertex in $D'$ has indegree $> 1$ then clearly $D'$ is the
required tree-labelling for $\TTT$. 
Otherwise, we proceed as follows. Recall
from~\cref{lem:opposite-orientations-uncross-independent} that if $T$
and $T'$ are independent and orient $\sigma(T)$ and $\sigma(T')$ in
different ways, then $\sigma(T)$ and $\sigma(T')$ are uncrossed. Furthermore, in
$D'$ no vertex has incoming edges from dependent tangles. Thus, in
$D'$, for all $u \in V(D')$ such that $T_o$ is not the only in-neighbour
of $u$, either all in-neighbours are from $\TTT^+$ or
all in-neighbours are from $\TTT^-$.

We will now define a sequence $((D_i, \gamma_i))_{i \geq 0}$ of digraphs $D_i$
and functions $\gamma_i \sth E(D_i) \rightarrow \SSS(G)$ with
vertex sets $V(D_i) = V(D')$ as follows. Set $D_0 = D'$ and $\gamma_0 := \gamma$. By
construction, $(D_0, \ \gamma_0)$ satisfy the following properties which we will
maintain for all $D_i$ and $\gamma_i$. To simplify notation, let  $(D_i)_t := \{ t' \in V(D_i) \sth t'$ is reachable from
  $t$ by a directed path in $D_i \}$, for all  $t \in V(D_i)$.
\begin{enumerate}
\item If $e = (T, T') \in E(D')$ then $\gamma(e) = (A, B)$ is a
  separation of order $l$ such that all $T'' \in (D_0)_{T'}$ contain $(A, B)$ and all 
  $T'' \in V(D_0) \setminus (D_0)_{T'}$ contain $(B, A)$. In particular $T_o$
  contains $(B, A)$.
Furthermore, there is no separation $(A', B')$ of order $l$ distinguishing all
  tangles in $(D_0)_{T'}$ from all tangles in $V(D_0) \setminus (D_0)_{T'}$ with $B'
  \subsetneq B$ and $A \subsetneq A'$.
\item If $e, e' \in E(D_i)$ have the same head then $\gamma(e) = \gamma(e')$.
\end{enumerate}

Now suppose $D_i$ has already been defined. 

\begin{lemma}\label{lem:new-tangle-intersection}
  Let $T \in V(D_i)$ be a node and $T_1, T_2 \in V(D_i)$ be in-neighbours of
  $T$. Then there is a
  separation $X = (A, B)$ of order $l$ such that if $T' \in (D_i)_{T_1}
  \cup (D_i)_{T_2}$ then $(A, B) \in T'$,
  otherwise $(B, A) \in T'$. 
\end{lemma}
\begin{proof}
  For $i = 1,2$ let
  $e_i = (T'_i, T_i)$ be an incoming edge and $X_i = (A_i, B_i) = \gamma(e_i)$.
  By~\cref{lem:opposite-orientations-uncross-independent},  as $T_1,
  T_2$ are independent, either $X_1$
  and $X_2$ are both outgoing (i.e.~$B_i = X_i^+$) or both are
  incoming. W.l.o.g.~suppose both are outgoing.
  By the first condition above,
  $(B_i, A_i) \in T_o$ for $i = 1, 2$ and, as $T$ is a descendant of
  $T_1, T_2$, $(A_i, B_i) \in T$. Thus $(A_1 \cup A_2, B_1 \cap B_2)$ and
  $(A_1 \cap A_2, B_1 \cup B_2)$ both distinguish $T_o$ and $T$ and
  therefore must both be of order $l$.

  We claim that $X = (A_1 \cap A_2, B_1 \cup B_2)$ satisfies the
  requirements of the lemma. By definition, if $T' = T_1$ or $T'$ is a
  descendant of $T_1$, then $(A_1, B_1) \in T'$. But then, $(A_1 \cap A_2,
  B_1 \cup B_2) \in T'$ as otherwise $B_1\cup B_2$ would be a small side which,
  together with the small side $A_1$ would cover all of $V(G)$.
  Similarly, if $T'$ is
  $T_2$ or a descendant of $T_2$, then  $(A_1 \cap A_2, B_1 \cup B_2) \in T'$. 

  Thus it remains to show that no other tangle $T''$ contains $(A_1 \cap A_2,
  B_1 \cup B_2)$. But if $T'' \not\in  \{ T_1, T_2\}$ and $T''$ is not a
  descendant of $T_1$ or $T_2$, then it contains $(B_1, A_1)$ and
  $(B_2, A_2)$ and therefore must also contain $(B_1 \cup B_2, A_1 \cap
  A_2)$. 
\end{proof}

We are now ready to complete the construction of $D_{i+1}$. If $D_i$
has no vertex of indegree $>1$, the construction stops here and we
define $D_\infty := D_i$. Otherwise choose such a node $T$ with
$\delta^-_{D_i}(T) > 1$ of minimal rank
among all such vertices and choose in-neighbours $T_1 \neq T_2$ of $T$.
By~\cref{lem:new-tangle-intersection}, there is a separation
$X = (A, B)$ of order $l$ such that $(A, B) \in T'$ for all
$T' \in (D_i)_{T_1} \cup (D_i)_{T_2}$ and $(B, A) \in T'$ for all
$T' \in V(D_i) \setminus  ((D_i)_{T_1} \cup (D_i)_{T_2})$. We choose such a
separation $X = (A, B)$ such that $B$ is inclusion-wise minimal among
all such separations.

We define $D_{i+1}$ as the digraph obtained from $D_i$ as follows. Replace
each incoming edge 
$e_1 = (T_1', T_1)$ of $T_1$ by $e_1' = (T_1', T_2)$ with $
\gamma_{i}(e_1') = X$. Let $Y = \gamma_{i}(e_1)$ for an edge $(T_1', T_1)$. We add the
edge $e = (T_2, T_1)$ with $\gamma_{i+1}(e) = Y$. If $e$ is an edge with
head $T_2$ we set $\gamma_{i+1}(e) = X$ . Finally, we set $\gamma_{i+1}(e) =
\gamma_i(e)$ for all $e \in E(D_i) \cap E(D_{i+1})$.

Let us see that  $D_{i+1}$ still satisfies the conditions
above. The last condition is satisfied by explicit
construction. Condition $1$ is satisfied for all edges in $E(D_{i+1})
\cap E(D_i)$ by induction hypothesis. For the edge $e = (T_2, T_1)$
Condition 1 
also follows from the induction hypothesis as $\gamma_{i+1}(e)$ is the same as
the label $\gamma_{i}(e')$ of the incoming edges of $T_1$ in $D_i$. Finally,
for the edges with head $T_2$ the condition follows from our choice of $X$. 

What remains to show is that the process terminates. For this we
define the \emph{conflict number} $c(D_i)$ of $D_i$ as the sum $\sum \{
n(m-r(T)) + \delta^-_{D_i}(T) \sth T \in V(D'), \delta^-_{D_i}(T) > 1 \}$, where
$r(T)$ is the rank of $T$, $n$ is the number of vertices in $D_i$
(which is the same for all $i$) and $m$ is the maximal rank of any
tangle in $\TTT$.
In the step from $D_i$ to $D_{i+1}$, we choose a node $T$
of minimal rank with $\delta^-_{D_i}(T) > 1$ and reduce the number of
incoming edges, i.e.~$\delta^-_{D_i}(T) > \delta^-_{D_{i+1}}(T)$. We may
increase the number of incoming edges of $T_2$ but $T_2$ has higher
rank than $T$. Thus, in the conflict number we remove the term for $T$
and add or increase the term for $T_2$, i.e.~$c(D_{i+1}) = c(D_i) -
n(m-r(T)) - \delta^-_{D_i}(T)) + n(m-r(T_2)) + \delta^-_{D_{i+1}}(T_2)$. But as
$\delta^-_{D_{i+1}}(T_2) < n$ and $r(T) < r(T_2)$ we get $c(D_i) >
c(D_{i+1})$. Thus the construction must stop after a finite number of
iterations with a digraph $D_\infty = D_s$ in which there are no vertices of
indegree $> 1$ left. Thus, $D_\infty$ is a tree and $\gamma_\infty$
satisfies the requirements of a tree-labelling. 

To turn $D_\infty$ formally into a tree labelling we define $\beta_\infty(T) = T$ for
all $T \in V(D_S)$.

This concludes the construction of the tree-labelling and the proof of \cref{thm:order-k-tree-labellings}.

\subsection{Non-Minimal Tangle Distinguishers}

The goal of this subsection is to lift the previous result from sets
$\TTT$ of tangles which are pairwise $l$-distinguishable but
$l{-}1$-indistinguishable to arbitrary sets of pairwise distinguishable tangles.

\begin{definition}[cone]\label{def:cone}
  Let $\TTT$ be a set of pairwise distinguishable tangles and let  $T \in
  \TTT$ be a tangle of order $k$.
  \begin{enumerate}
  \item For $l \geq 1$ we define the \emph{restriction $T_{|l}$ of $T$ to order
      $l$}, or simply the \emph{$l$-restriction} of $T$, as the set
    \[
       T_{|l} := \{ (A, B) \in T \sth |A \cap B| \leq l \}.
     \]
     If $\TTT' \subseteq \TTT$ we define $\TTT'_{|l} := \{ T_{|l} \sth T \in \TTT' \}$.
  \item The \emph{$\TTT$-cone} of $T$ is defined as the set
    \[
       \cone_{\TTT}(T) := \{ T' \in \TTT \sth T \subseteq T'\}.
     \]
   \item We define the \emph{strict $\TTT$-cone} of a tangle $T$ of order
     $l$ as the set $\scone_{\TTT}(T) = \{ T'_{|l+1} \sth T' \in \cone_{\TTT}(T)
     \}$.
     \item If $T$
     is any tangle in $\TTT$ of order $> l$ we define $\scone_\TTT(T, l) :=
     \scone_{\TTT}(T_{|l})$ and      $\cone_\TTT(T, l) :=
     \cone_{\TTT}(T_{|l})$.
  \end{enumerate}
    We usually omit the index $\TTT$ when it is clear from the context.
\end{definition}

Note that if $T$ is a tangle of order $l$ then any pair of tangles in
$\scone_{\TTT}(T)$ are $l{+}1$ but not $l$-distinguishable. 

\begin{lemma}\label{lem:non-min-uncross}
  Let $T, T'$ be tangles and $l$ be an integer such that $T_{|l-1} =
  T'_{|l-1}$ but $T_{|l} \neq T'_{|l}$. Let $\SSS_1 \subseteq \cone_{\TTT}(T)$ and $\SSS_2
  \subseteq \cone_{\TTT}(T')$. There is no separation $X = (A, B)$ such that
  $X$ separates $T_1 \in \SSS_1$ and $T_2 \in \SSS_1$ and $X$ also separates
  $T_3 \in \SSS_2$ and $T_4 \in \SSS_2$ and $X$ is a minimum order
  $T_1{-}T_2$-separation or a minimum order $T_3{-}T_4$-separation.
\end{lemma}
\begin{proof}
  By assumption  $T_1$ and $T_2$ are distinguishable from $T_3$ and
  $T_4$ by a separation $Y$ of order $l$. W.l.o.g.~we assume that $(Y^-, Y^+) \in T_1 \cap 
  T_2$ and $(Y^+, Y^-) \in T_3 \cap T_4$. 

  Towards a contradiction, suppose that a separation $X$ as in the
  statement of the lemma exists. That is,  $X$ separates $T_1$ from $T_2$ and
  also $T_3$ from $T_4$ and it is a minimum order distinguisher for
  $T_1$ and $T_2$ or a minimum order distinguisher for
  $T_3$ and $T_4$. W.l.o.g.~suppose $X$ is a minimum order
  distinguisher for $T_1$ and $T_2$.
  By renaming $T_1$ and $T_2$ or $T_3$ and $T_4$, if necessary, we may
  assume that $(X^-,
  X^+) \in T_1 \cap T_3$ and $(X^+, X^-) \in T_2 \cap T_4$.
  Let $k := |X|$.
  
  Then $k > l$, as $T_{1|l} = T_{2|l}$ and $T_{3|l} =
  T_{4|l}$. 
  See~\cref{fig:non-min-tree-lab} for an illustration.

  Now consider the pair $(Y, X)$ and its upper and lower quadrant $U =
  (X^+ \cap Y^+, X^-\cup Y^-)$ and $L = (X^+ \cup Y^+, X^- \cap Y^-)$.
  Observe that $|U|, |L| \geq l$ as both
  distinguish a tangle from $\{ T_1, T_2 \}$ from a tangle in $\{ T_3,
  T_4 \}$. By submodularity $|U|+|L| \leq k+l$ and thus $|U|, |L| \leq k$ and
  therefore $|U|, |L|$ are less than the order of the tangles $T_1,
  \dots, T_4$. 

  
  By assumption,
  $(X^+, X^-) \in T_2$ and $(Y^-, Y^+) \in T_2$. If in addition $(X^-\cup
  Y^-, X^+ \cap Y^+) \in T_2$, then $X^+, Y^-$, and $X^- \cup  Y^-$ would be
  three small sides covering $V(G)$, contradicting the tangle
  properties.
  Thus,  $(X^+ \cap Y^+, X^-\cup
  Y^-) \in T_2$,  $U$ is a $T_1{-}T_2$-distinguisher and therefore $|U| \geq k$. Together with $|L| \geq l$ this implies that $|U|
  = k$ and $|L| = l$. But $(X^+, X^-) \in T_4$ and $(Y^+, Y^-) \in T_4$
  and therefore $(X^+ \cup Y^+,
  X^- \cap Y^-) \in T_4$ as otherwise $X^+, Y^+$, and $X^- \cap Y^-$ would all
  be small sides in $T$ but their union is $V(G)$. But as $(Y^+, Y^-)
  \in T_3$ and $(X^-, X^+) \in T_3$, $(X^+ \cup Y^+,
  X^- \cap Y^-)$ cannot be in $T_3$. Thus $L$ distinguishes $T_3$ and
  $T_4$, a contradiction to $T_{3|l} = T_{4|l}$.
\end{proof}
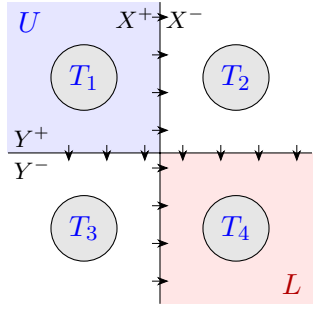
\begin{figure}
  \begin{center}
    \beginpgfgraphicnamed{fig-non-min-tree-lab}\begin{tikzpicture}
      \path [draw=none, fill=blue!10] (-2,0) to (0,0) to (0,2) to
      (-2,2) to cycle; \path [draw=none, fill=red!10] (0,-2) to (0, 0)
      to (2,0) to (2,-2) to cycle; \path[draw=black, circle] (-1,1)
       node[circle, draw=black, fill=black!10] {$T_1$} (1,1)
       node[circle, draw=black, fill=black!10] {$T_2$} (-1,-1)
       node[circle, draw=black, fill=black!10] {$T_3$} (1,-1)
       node[circle, draw=black, fill=black!10] {$ T_4$}; \draw[sepout]
       (-2,0) to (2,0) (0,-2) to (0,2) ;
       \path[font={\footnotesize},inner sep=2pt] (0,2) node
       [anchor=north east] {$\black{}X^+$} (0,2) node [anchor=north
       west] {$\black{}X^-$} (-2,0) node [anchor=south west]
       {$\black{}Y^+$} (-2,0) node [anchor=north west] {$\black{}Y^-$}
       ; \path (-2,2) node [anchor=north west] {$\blue{}U$} (2,-2)
       node [anchor=south east] {$\red{}L$} ;
     \end{tikzpicture}\endpgfgraphicnamed{}
   \end{center}
\caption{Uncrossing separations in \cref{lem:non-min-uncross}.}
\label{fig:non-min-tree-lab}
\end{figure}

We are now ready to prove the main result of this section.

\smallskip

\begin{proof}[Proof of \cref{thm:tangle-tree-labellings}]
  Let $\TTT$ be a set of tangles in $G$ as in the statement of the
  theorem.   Let $m$ be the maximal order of any tangle in $\TTT$.
  For $T \in \TTT$ and $l \leq m$ let $\CCC(T, l) := \cone_\TTT(T_{|l})$ be the cone
  of the $l$-restriction of $T$.
  
  By induction on $l$ from $m$ to $0$ we will construct for each
  $T \in \TTT$ and $l \leq m$ such that $\CCC(T, l) \neq \emptyset$ a tree-labelling for
  $\CCC(T, l)$.
  
  Clearly, a tree-labelling for $G$ follows immediately from
  the set of tree-labellings for the $0$-cones $\CCC(T, 0)$, for $T \in \TTT$.

  If $l = m$ then $\CCC(T, l) = \{ T \}$ for all $T \in \TTT$ of order
  $m$ and otherwise $\CCC(T, l) = \emptyset$. In the first case, let $L_{T,l}$ be
  a tree with a single node $r = r_{T,l}$ and define $\beta_{T,l}(r) = T$
  and $\gamma_{T,l} = \emptyset$. Then $\LLL(T, l) := (L_{T, l}, \beta_{T,l}, \gamma_{T,l})$ is a   
  tree-labelling for $\CCC(T, l)$.

  Now suppose tree-labellings $\LLL(T, l+1)$ for all $\CCC \in \{ \CCC(T, l+1) : T \in \TTT\}$ have
  been constructed. Let $\CCC = \CCC(T, l)$ for some $T \in \TTT$ and let $\SSS :=
  \scone_{\TTT}(T, l)$. If $|\SSS| = 1$, i.e.~any pair $T_1, T_2 \in \CCC(T, l)$
  is $l$-indistinguishable, then we
  set $\LLL(T, l) := \LLL(T, l+1)$. Otherwise, choose for each
  $T' \in \SSS$ a tree-labelling $\LLL(T', l+1)$, which exists by induction
  hypothesis. By definition of strict cones, if $T_1 \neq T_2 \in \SSS$ then
  $T_1$ and $T_2$ are $l+1$-distinguishable but not
  $l$-distinguishable. Thus, by~\cref{thm:order-k-tree-labellings}, there is a tree-labelling
  $\LLL = \LLL(\SSS) = (L, \beta, \gamma)$ for $\SSS$. Observe that any node in $L$ is
  labelled by a tangle $T'$ in $\SSS$ which is the common
  $l+1$-restriction of all tangles in $\cone_{\TTT}(T', l+1)$. Thus, any
  node in $L$ possibly represents an entire set of tangles.

  To obtain a tree-labelling for $\CCC(T, l)$, we need to combine the
  ``outer'' tree-labelling $\LLL$
  with the ``inner'' tree-labellings for $\{ \cone_{\TTT}(\beta(t), l+1) \sth t \in V(L)\}$.
  See \cref{fig:tree-lab-comb} for an illustration.
  \begin{figure}[t]
    \centering
    \includegraphics[height=10cm]{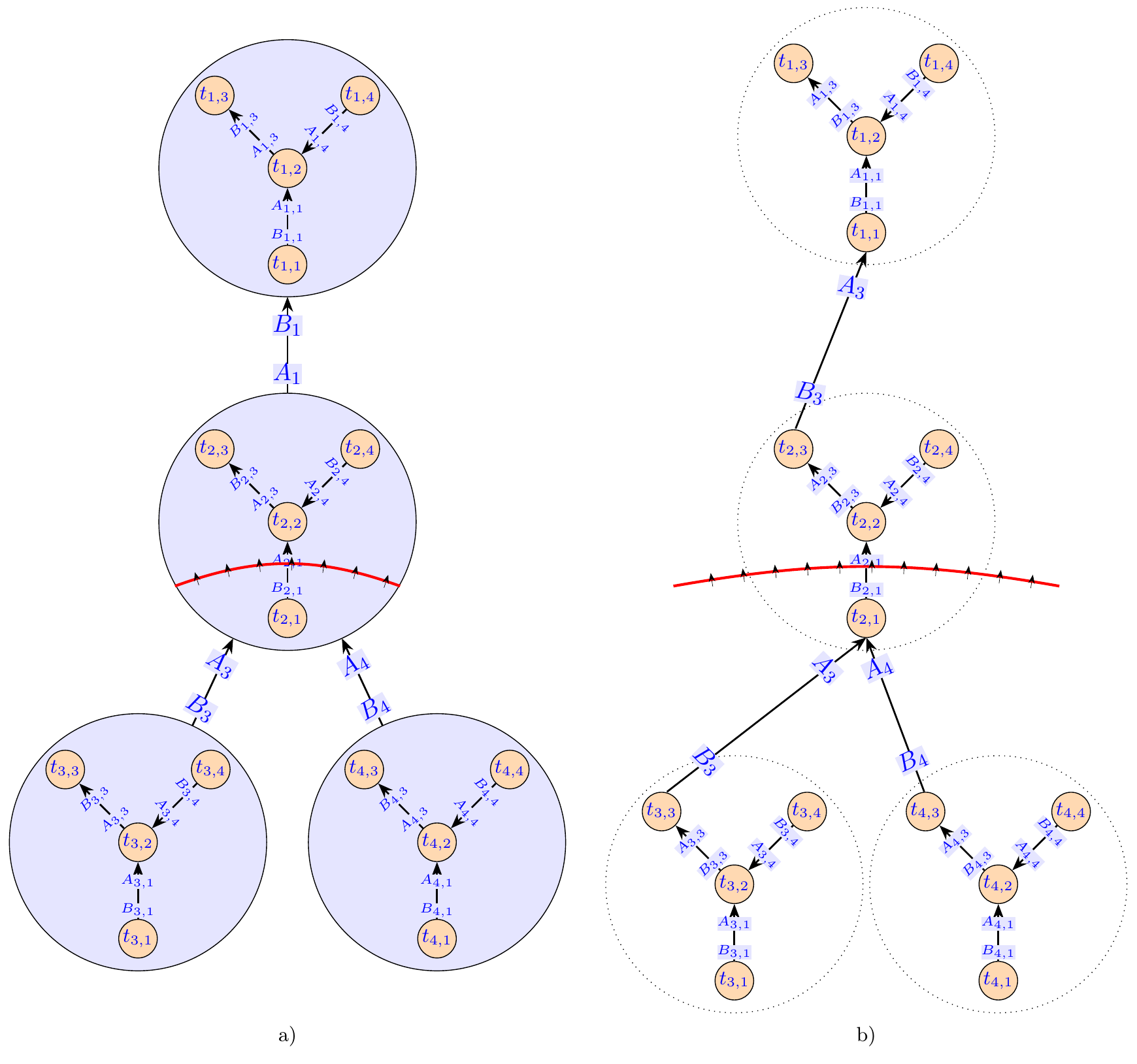}
    \caption{Combining tree-labellings. The large (blue) disks
      represent the elements of $\SSS$, and the smaller trees inside a disk $t$ represent
      the tree-labelling for the cone of $t$. }
    \label{fig:tree-lab-comb}
  \end{figure}

  For  all $t \in V(L)$  let $\LLL(t) =  (L_t, \beta_t, \gamma_t) := \LLL(\beta(t), l+1)$.
  Let $e = (s, t) \in E(L)$ be an edge of the outer tree and let $X_e =
  (A, B)$ be the separation at $e$. By construction, all tangles from
  $\SSS$ contained in  $\{ \beta(u) \sth u \in V(L_{e,s})\}$ contain $(B, A)$ whereas
  all tangles in $\{ \beta(u) \sth u \in V(L_{e,t})\}$ contain $(A, B)$. It is
  immediate from the definition of cones that if $u \in V(L_{e,t})$ and
  $T$ is a tangle in $\cone_\TTT(u, l+1)$, then $T$ contains $(A, B)$
  whereas if $T \in \cone_\TTT(u,l+1)$ for some $u \in V(L_{e,s})$, then $(B,
  A) \in T$. Thus the separations at the outer edges $e \in E(L)$ are
  consistent with the inner tree-labellings $\LLL(t)$ for  all $t
  \in V(L)$.

  Now let  $t \neq t' \in V(L)$, $e \in E(\LLL(t))$ be an edge and 
   $(A', B')$ := $\gamma_t(e)$ be the separation at $e$. Furthermore, let $T_1, T_2$ be tangles in
  $\cone_\TTT(\beta(t'), l+1)$. Then either $(A', B') \in T_1 \cap T_2$ or $(B',
  A') \in T_1 \cap T_2$.
  To see this, recall that by the definition of tree-labellings there
  are $t_3, t_4 \in V(L_t)$ and tangles $T_3 = \beta_t(t_3), T_4 = \beta_t(t_4)$ in $\cone_\TTT(\beta(t), l+1)$ such
  that $\gamma_t(e)$ is a minimum order $T_3{-}T_4$-distinguisher.
  As $|\gamma_t(e)| > l$, by~\cref{lem:non-min-uncross}, $\gamma_t(e)$ cannot
  also separate $T_1, T_2$. 
  In other words, if $X$ is a separation at an edge $e$ of some inner
  tree-labelling $\LLL(t)$, then $X$ does not distinguish any pair of
  tangles in another inner labelling $\LLL(t')$ for $t \neq t'$.
  This allows us to combine the tree-labellings as follows.

  Let $(L', \beta', \gamma')$ be the union of $\bigcup \{ \LLL(t) \sth t \in V(L) \}$,
  i.e.~$L'$ is the disjoint union of the individual tree-labellings
  and  $\beta'(u) = \beta_t(u)$ and $\gamma'(e) = \gamma_t(e)$ for all $t \in V(L)$,  $e \in E(L_t)$, and $u \in V(L_t)$.  What is left to do is to add edges
  corresponding to the edges in $E(L)$.
  
  Let $e = (s, t) \in E(L)$ and let $X_e = (A_e, B_e)$ be the separation
  at $e$. We orient the edges of $\LLL(s) = (L_s, \beta_s, \gamma_s)$ and $\LLL(t)
  = (L_t, \beta_t,
  \gamma_t)$ as follows:  $e' = (u,v) \in E(L_s)$ with separation $(A', 
  B')$ is oriented towards $u$ if $(B', A') \in \beta(w)$ for all $w \in
  V(L_t)$. Otherwise it is oriented towards $v$. Analogously,  $e'
  = (u,v) \in E(L_t)$ with separation $(A', 
  B')$ is oriented towards $u$ if $(B', A') \in \beta(w)$ for all $w \in
  V(L_s)$. Otherwise it is oriented towards $v$. 
  As $L_s, L_t$ are trees, there are $u \in V(L_s)$ and $v \in V(L_t)$
  such that all incident edges of $u$ point towards $u$ and all
  incident edges of $v$ are oriented towards $v$.
  We add the edge $e' = (u, v)$ to
  $E(L')$ and set $\gamma'(e') = \gamma(e)$. 

  This completes the construction. We claim that $(L', \beta', \gamma')$ is a
  tree-labelling for $\CCC(T, l)$. It is easily seen that $\beta'$ is a
  bijection between $\CCC(T, l)$ and $V(L')$. So Condition $1$ holds. 
  Towards proving Condition
  $2$, let $t \neq t' \in V(L)$ and let $P$ be the unique path between $t$
  and $t'$ in $L'$. If there is an $s \in V(L)$ such that $t, t' \in
  V(L_s)$, then $P \subseteq L_s$ and Condition $2$ follows from the induction
  hypothesis. Otherwise there are $s \neq s' \in V(L)$ with $t \in V(L_s)$
  and $t' \in V(L_{s'})$. But then the path $P$ must contain an edge $e
  \in E(L)$ of order $l$ -- which is the minimal order of any edge on $P$
  -- and $\gamma'(e) = \gamma(e)$ is a minimum order distinguisher for $\beta'(t)$
  and $\beta'(T')$ by the construction of $L$.

  Condition $3$ follows analogously. This completes the proof of \cref{thm:tangle-tree-labellings}.%
%
\end{proof}

\begin{remark}
  The definition of  tree-labellings does not guarantee that if $t, t'$
  are distinct nodes of a tree-labelling $(L, \beta, \gamma)$ and $e$ is an edge on the
  path between $t$ and $t'$ then $\gamma(e)$ distinguishes $\beta(t)$ and
  $\beta(t')$. This is only guaranteed for edges of minimal order on the
  path. \cref{fig:no-uncross} illustrates an example for a set
  of tangles where there is no optimal tree-labelling with the
  stronger property that every edge on the path between two nodes
  distinguishes the associated tangles. For, in the example in
  \cref{fig:no-uncross}, the tangle $T_4$ can be distinguished
  from $T_2$ by an order $5$ separation and also from $T_3$ by an order
  $5$ separation but there is no separation distinguishing $T_4$ from
  $T_2$ and $T_3$ of order $5$. But $T_1$ can be distinguished from
  $T_2, T_3, T_4$ by an order $1$ separation $X_1$. Thus if we require that
  in a tree-labelling the path between two notes $t$ and $t'$ must
  contain an edge labelled with a minimal-order distinguisher for
  $\beta(t)$ and $\beta(t')$, then the separation $X_1$ must occur on each
  path between $T_1$ and any of the three other tangles. But $X_1$
  does not distinguish between any pair of tangles from $T_2, T_3,
  T_4$, and thus if we also require that on any path between two nodes
  each edge must be labelled by a distinguisher between the associated
  tangles, then $X$ must not occur on any path between $T_2, T_3$, and
  $T_4$. Thus there must be an edge labelled by $X_1$ and $T_1$ occurs
  on one side of the edge and $T_2, T_3, T_4$ occur together on the
  other side. But then, the subtree containing $T_2, T_3, T_4$ must
  contain edges $e_2, e_3$ such that $\gamma(e_2) = X_2$ distinguishes
  $T_2$ from $T_4$ and  $\gamma(e_3) = X_3$ distinguishes
  $T_3$ from $T_4$. But $X_3$ distinguishes $T_1, T_3$ from $T_2, T_4$
  and $X_2$ distinguishes $T_1, T_2$ from $T_3, T_4$. Therefore we
  cannot arrange $T_2, T_3, T_4$ into a tree and connect $T_1$ to this
  tree in a way that every edge on the paths from $T_2$ and $T_3$ to
  $T_1$ is a distinguisher for $T_2$ and $T_1$ or $T_3$ and $T_1$,
  resp.

  However, if we relax the structure of a tree-labelling $(L, \beta, \gamma)$
  to allow $L$ to be a dag, then the proof above can easily be
  modified to yield a dag $L$ with $\beta$ and $\gamma$ as before such that
  $(L, \beta, \gamma)$ satisfies the conditions $1, 2$, and $3$ of a
  tree-labelling and the extra condition that on any path between two
  nodes $t, t'$ every edge on the path is labelled by a separation
  distinguishing $\beta(t)$ and $\beta(t')$. \hfill$\dashv$
\end{remark}
\section{From Tree-Labellings to Tree-Decompositions}

In this section we extend the tree-labelling theorem to directed
tree-decompositions. 

\begin{definition}\label{def:dtd-for-T}
    Let $k$ be an integer and let $\TTT$ be a set of tangles in a digraph
  $G$ of order $> k$. A \emph{directed tree-decomposition for $\TTT$}, or
  \emph{distinguishing $\TTT$}, is a tuple $\DDD := (L, \beta, \gamma, \tau, \omega)$, where $L$
  is a tree, $\tau$ is an injective map from $\TTT$ to $V(L)$, $\beta \sth V(L) \rightarrow
  2^{V(G)}$, $\gamma \sth E(L) \rightarrow \SSS(G)$, and $\omega \sth E(L)
  \rightarrow 2^{V(G)}$, such that
  \begin{enumerate}
  \item $\LLL := (L[\{\tau(T) \sth T \in \TTT\}], \tau, \gamma)$ is a tree-labelling for $\TTT$,
  \item $(L, \beta, \omega)$ is a directed tree-decomposition of $G$,
  \item for all $e = (s, t) \in E(L)$, if $\gamma(e) = (A, B)$ then $A \cap B
    \subseteq \omega(e)$.
  \end{enumerate}
  We say that $\DDD$ \emph{extends}, or is \emph{consistent with}, the
  tree-labelling $\LLL$.
  The \emph{edge-width} of $\DDD$ is $\max\{ |\omega(e)| \sth e \in E(L) \}$.
\end{definition}

\begin{theorem}\label{thm:tangle-decomp}
  Let $k$ be an integer and let $\TTT$ be a set of tangles in a digraph
  $G$ of order $> k$. Then there is a directed tree-decomposition for
  $\TTT_{|k}$ in $G$ of edge-width at most $k^2 + 2k$.

  More precisely, every tree-labelling $\LLL$ of $\TTT_{|k}$ can be extended to
  a directed tree-decomposition of edge-width $\leq k^2 + 2k$ extending $\LLL$.
\end{theorem}

By the results of the previous section we may assume that we are given
a tree-labelling $\LLL := (L, \tau_L, \gamma_L)$ for $\TTT$. Let $\leq^s_L$ be a sibling-ordering on
$L$, i.e.~$\leq^s_L$ contains for each node $t \in V(L)$ an ordering of its children.
Let $\leq_L$ be the pre-order {\smaller\textsc{DFS}} ordering generated by $\leq^s_L$, i.e.~the
ordering such that the root of $L$ is the smallest element and if $s
\leq_L^s t$ are children of the root then all nodes in $L_s$ are smaller
than all nodes in $L_t$.

Let $V(L) = \{ t_1, \dots, t_n \}$ be numbered such that $t_i \leq_L t_j$ whenever $i <
j$. Thus, $t_1$ is the root of $L$. For $1
\leq i \leq n$ let $T_i := \tau_L(t_i)$ be the  associated tangle and, if 
$j > 1$ let $e(t)$ be the last
edge on the path $P$ in $L$ from $t_1$ to $t_j$. Let $\sigma(t_j) := \gamma(e(t_j))$ and
 $B_j := \beta_{T_j}(\sigma(t_j))$ be the big side of $\sigma(t_j)$ in the tangle
$T_j$. 
We define $\partial(B(t_j)) = \partial^+(B(t_j))$, and call $B(t_j)$ \emph{outgoing},  if $B(t_j)$ is oriented outwards in
$\sigma(t_j)$. Otherwise $B(t_j)$ is \emph{incoming} and we define $\partial(B(t_j)) =
\partial^-(B(t_j))$. We denote the set
$B_j \setminus \partial(B_j)$ as $\interior{B_j}$.

We say that $B_i$ and $B_j$ are
\emph{independent} if  $t_i$ and
$t_j$ are independent. Otherwise $B_i$ and $B_j$ are \emph{dependent}. In
this case we say that $B_i$ is \emph{above} $B_j$, and $B_j$ is
\emph{below} $B_i$, if $t_j \in V(L_{t_i})$. To ease notation, we say
that $i$ is independent of $j$ and $\CCC_i$ is independent of $\CCC_j$ if
$B_i$ is independent of $B_j$.

Let $\BBB := \{ B_1, \dots, B_{n}\}$. For $1 \leq i \leq n$, let $\CCC^0_i := \{ C \sth C
\subseteq B_i$ and $C$ is a strong component of $G -
\partial(B_i) \}$. For each $1 \leq i \leq n$ and $C \in \CCC^0_i$ we remove $C$ from $\CCC^0_i$ if there
is an independent $\CCC^0_j$ and some $C_j \in \CCC^0_j$ and either $C \subsetneq C_j$ or $C = C_j$ and
$j < i$. Let, for $1 \leq i \leq n$,  $\CCC_i$ be the set of 
components of $\CCC^0_i$ that were not deleted in this step.
For two strongly connected subgraphs $C, C'$ we say that $C$ and $C'$
\emph{overlap} if $C \setminus C'$ and $C' \setminus C$ and $C \cap C'$ are all non-empty.

For $1 \leq j \leq n$ let $I_j :=  \bigcup \CCC_j$ and let $\kappa(B_j, a_j)$ be the set of all indices $i$ such
that
 $B_i$ and  $B_j$ are independent and there is a
  component $C \in \CCC_i$ such that $a_j \in V(C_i)$ and
  $V(C_i) \cap I_j \neq \emptyset$ and $V(C_j) \setminus I_j \neq \emptyset$. 

We call  $a_j \in \partial(B_j)$ \emph{conflicting} if $\kappa(B_j, a_j) \neq \emptyset$.
The \emph{resolvant} $\rho(B_j, a_j)$ is defined as $\min \kappa(B_j, a_j)$ and
the resolvants of $B_j$ are the elements of  $\{ \rho(B_j, a_j) \sth a_j \in \partial(B_j) \}$.
Finally, we define $\kappa(B_j) := \bigcup_{a_j \in \partial(B_j)}\kappa(B_j, a_j)$.

We say that $B_j$ and $B_l$ are in \emph{conflict} if
$j \in \kappa(B_l)$.

\begin{lemma}\label{lem:dis:1}
  Let $j$ and $i$ be independent. If there is $a_l \in \partial(B_l)$ such that
  $j \in \kappa(B_l, a_l)$ then there are  components $C_l
  \in \CCC_l$ and $C_j \in \CCC_j$   and $a_j \in \partial(B_j)$ such that $C_l$ and
  $C_j$ overlap,  $a_l \in C_j$ and $a_l
  \in C_j$. In particular, if  $j \in \kappa(B_l)$ then $l \in \kappa(B_j)$.
\end{lemma}
\begin{proof}
  If $j \in \kappa(B_l)$ then there must be an $a_l \in \partial(B_l)$ and 
   $C_j \in \CCC_j$ and $C_l \in \CCC_l$ such that $a_l \in
  V(C_j)$ and $C_j$ and $C_j$ overlap. Thus $C_j \cap C_l \neq \emptyset$. But as
  $C_j$ is a component of $G - \partial(B_j)$ and $C_l$ is strongly
  connected, either $C_l \subseteq C_j$ or $C_j \cap \partial(B_l) \neq \emptyset$. But in the
  first case one of $C_l, C_j$ would have been removed in the first step.
\end{proof}

\begin{lemma}\label{lem:dis:independent}
  Suppose that $j$ and $l$ are independent and $j \in \kappa(B_l, a_l)$ for
  some $a_l \in \partial(B_l)$. Then there is  $a_j \in \partial(B_j)$ such that $l
  \in \partial(B_j)$ and 
  $j \in \rho(B_l, a_l)$ or $l \in \rho(B_j, a_j)$ or $\rho(B_l, a_l) = \rho(B_j, a_j)$.
\end{lemma}
\begin{proof}
  By assumption,  $j \in \rho(B_l, a_l)$. By~\cref{lem:dis:independent},  there are
  components $C_j \in \CCC_j$ and $C_l \in \CCC_l$ which overlap and witness the
  conflict of $B_j$ and $B_l$ and an element $a_j \in V(C_l) \cap \partial(B_j)$.
  Let $i_j = \rho(B_j, a_j)$ and $i_l = \rho(B_l, a_l)$. If $i_j =
  l$ or $i_l = j$ we are done. So we may assume that $i_j \neq l$ and $i_l \neq j$ and
  therefore $i_j < l$ and $i_l < j$.  By definition, there is $C_i \in
  \CCC_{i_l}$ containing $a_l$ which conflicts with some component $C_l
  \in \CCC_l$. But then $C_j \cap C_i \neq \emptyset$ and as neither can be contained in
  the other, the two components overlap
  and witness a conflict. This implies that $i_j \leq i_l$.

  Analogously, there must be a component $C'_i \in \CCC_{i_j}$ which
  witnesses a conflict with $B_j$ and contains $a_j$. By the same
  argument as before, as $a_j \in C'_i \cap C_l \neq \emptyset$ the component $C'_i$
  must witness a conflict between $B_{i_j}$ and $B_l$ and thus $i_l \leq
  i_j$. This implies that $i_j = i_l$. 
\end{proof}

For all $1 \leq i \leq n$ let $\omega_i := \bigcup \{ \partial(B_j) \sth j \in \rho(B_i) \cup \{ i \}\}$ and let
$\DDD'_i$ be the set of components $C$ of $G - \omega_i$ such that there is a
component $C' \in \CCC_i$ with $C \subseteq C'$.
For all $C \in \DDD'_i$, if there is $C' \in \DDD'_j$, $j \neq i$, such that $C \subseteq
C'$ then we delete $C$ from $\DDD'_i$ if $C \subsetneq C'$ or  $j < i$.
Let $\DDD_i$ be the components of $\DDD'_i$ not deleted in the previous step.
We define $D_i := \bigcup \{ V(C) \sth C \in \DDD_i \}$.

We show next that the system $\DDD = (\DDD_i)_{1 \leq i \leq n}$ does not contain any
conflicts between independent tangles.

\begin{lemma}
  If $B_j$ and $B_l$ are independent then they do not have a conflict
  in $\DDD$.
\end{lemma}
\begin{proof}
  Towards a contradiction, suppose  $C_j \in \CCC_j$ and $C_l \in \CCC_l$ are in conflict.
  Then there must be $a_j \in V(C_l)$ and $a_l \in V(C_j)$ witnessing the
  conflict. Let $A_j = V(C_l) \cap \partial(B_j)$ and $A_l = V(C_j) \cap \partial(B_l)$.

  If $j \in \rho(B_l)$ or $l \in \rho(B_j)$, say $j \in \rho(B_l)$,  then $\partial(B_j) \subseteq
  \omega_l$ and $\DDD_l$ does not contain any component that contains a vertex
  of $B_j$. Thus, in these cases no conflict is possible between $\DDD_j$
  and $\DDD_l$.

  Thus we may assume that $j \not\in \rho(B_l)$ and $l \not\in \rho(B_j)$.
  By~\cref{lem:dis:independent}, in this case for each $a_j \in A_j$
  there is an $a_l \in A_l$ such that $\rho(A_j, a_j) = \rho(A_l, a_l)$. Let
  $I = \bigcup \{ \rho(B_j, a_j) \sth a_j \in A_j \}$ and $X = \bigcup \{ \partial(B_i) \sth  i \in I \}$. Thus $X \subseteq \omega_i \cap
  \omega_j$ and  $\DDD_j$ and $\DDD_l$ do not contain components in any of the
  sets $B_i$ for $i \in I$. Note that by construction $A_j, A_l \subseteq \bigcup \{
  B_i \sth i \in I \}$ and the components of $G - \omega_j$ into which $C_j$ is
  split are removed from $\DDD_j$ if they are contained in $\bigcup \{ B_i \sth i \in
  I \}$. Now suppose that there are components $C'_j \in \DDD_j$ and $C'_l \in
  \DDD_l$ such that $C'_j \subseteq C_j$ and $C'_l \subseteq C_l$ and $C'_j$ and $C'_l$
  are in conflict. Let $v \in V(C'_j) \cap V(C'_l)$. We claim that $C'_j \subseteq
  B_l$ and $C'_l \subseteq B_j$ and thus one of the components would have been
  removed from $\DDD_l, \DDD_j$, resp. For, if $C'_l$ contains a vertex $u$ of
  $B_l \setminus B_j$, then there is a closed walk $W$ in $C'_l \subseteq C_l$ containing
  $v$ and $u$ and this walk must contain an element of $\partial(B_j)$. But
  then $W$ must also contain an element of $X$, which is impossible. 
\end{proof}

Now consider the tree $L$ with the labelling functions $\beta'$ mapping
each $t_i \in V(L)$ to $D_i \cup \partial(B_i)$ and $\gamma'$ mapping each edge $e = (t_i, t_j)$
to $\omega_j$. By construction, every vertex appears in at most one
component $C \in \bigcup_i \DDD_i$ and moreover, if a vertex $v$ is not contained
in $\bigcup_i D_i$ then $v \in \partial(B_i)$ for some $1 \leq i \leq n$. The triple $\LLL' :=
(L,
\beta', \gamma')$ is almost what we need to construct our directed
tree-decomposition. 
What remains to be done is to remove remaining conflicts between
dependent tangles. In terms of $\LLL'$ this means that if $t_i$ is an
ancestor of $t_j$ then $D_j \setminus D_i \neq \emptyset$. If, furthermore, there is a
component $C \in \DDD_j$ with $C \cap D_i \neq \emptyset$ and $C \setminus D_i \neq \emptyset$, then $t_i$
and $t_j$ are in conflict and we need to remove this conflict before
we can turn $\LLL'$ into a directed tree-decomposition.

Towards this aim, we first prove the following lemma.

\begin{lemma}\label{lem:dependent}
  Let $t_i \in V(L)$ and $t_j \in V(L_{t_i})$ such that $B_j \setminus B_i \neq \emptyset$.
  Then $B_i$ and $B_j$ are oriented differently and $|\partial(B_j \setminus B_i)| <
  |\partial(B_j)|$, where  $\partial(B_j \setminus B_i) = \partial^+(B_j \setminus B_i)$
  if $B_i$ is incoming and $\partial(B_j \setminus B_i) = \partial^-(B_j \setminus B_i)$ otherwise.
\end{lemma}
\begin{proof}
  We consider the case where  $B_i$ is outgoing. The other case is
  symmetric. We show first that $B_j$ must be incoming. For, suppose
  that  $B_j$ is also outgoing. Let $k = |\partial(B_i)|$ and $l = \partial(B_j)$.
  Let $T_o$ be a tangle such that $\sigma(t_i)$ is a minimum order
  distinguisher for $T_i$ and $T_o$. W.l.o.g.~we assume that $\sigma(t_j)$
  is a mimimum order distinguisher for $T_j$ and $T_i$. 
  By~\cref{cor:submodularity}, $U = B_i \cap B_j$ and $D = B_i \cup B_j$ are
  outgoing sides of separations of order $l_u = \partial^+(B_i \cap B_j)$ and
  $l_d = \partial^+(B_i \cup B_j)$ and $k + l \geq l_u + l_d$. Furthermore, the
  separation
  $X_d$ with outgoing side $D$ separates $T_i$ and
  $T_o$ and the separation $X_u$ with outgoing side $U$ separates
  $T_j$ and $T_i$. Thus $l_d \geq l$ and $l_o \geq k$. By the minimality of
  $B_j$ this implies that $U = B_j$ and therefore $\sigma(T_j)$ and
  $\sigma(T_i)$ are uncrossed.

  Thus $B_i$ and $B_j$ must be oriented differently and this proves
  the first part of the lemma.

  If $B_j$ is incoming, then $B_i \setminus B_j$ and $B_j \setminus B_i$ are the
  upper and the lower corner of the pair $(\sigma(T_i), \sigma(T_j))$. Let $l_u
  = |\partial^+(B_i \setminus B_j)|$ and $l_d = |\partial^-(B_j \setminus B_i)|$. As $\sigma(T_j)$ is a
  minimum order distinguisher between $T_i$ and $T_j$, we have $l_u \geq
  l$. In fact, $l_u > l$ by the minimality of $B_i$. But this implies
  $l_d < k$, as required.
\end{proof}

Clearly, this lemma also implies that if $D_j \setminus D_i \neq \emptyset$ then $B_i$
and $B_j$ must be oriently differently. Thus if
$t_j$ and $t_i$ are in conflict and there is $t_l$ between $t_i$ and
$t_j$ such that $D_j$ and $D_l$ are also in conflict, then $D_l \subseteq D_i$.
Therefore,  if we remove $D_j \setminus D_l$ from $D_j$ then we also resolve the
conflict with $D_i$. Thus, to remove the conflict, it suffices to
choose the node $t_l$ on the path from $t_i$ to $t_j$ in $L$ which is
closest to $t_j$ and has opposite orientation to $t_j$.

For each $t_j$ that has a conflict with an ancestor let $p(t_j)$ be
this node. For each of these $t_j$ that has a conflict we set $\omega'_j :=
\omega_j \cup \partial(B_{p(t_j)})$. 
As a last step, if $t_j$ has a conflict with an ancestor $t_i$, where
we may assume that $t_i = p(t_j)$, then we set $\DDD''_j$ as the set of
components $C$ of $G - \omega'_j$ such that there is a component $C' \in \DDD_j$
with $C \subseteq C'$ and $C \subseteq D_i$. Furthermore, if $t_p$ is the parent of
$t_j$ in $L$, then we create a new child $s_j$ of $t_p$ with $\gamma'((t_p,
s_j)) := \omega'_j$ and $\beta''(s_j)$ as the union of all components of $G -
\omega'_j$ such that there is a component $C' \in \DDD_j$ 
with $C \subseteq C'$ but  $C \cap D_i = \emptyset$. We set $\omega''((t_p, t_j)) := \omega'_j$ and
$\beta''(t_j) := \bigcup \DDD''_j$.

For all other nodes $t_s$ with parent $t_p$ we set $\omega''((t_p, t_s)) :=
\omega_s$ and $\beta''(t_s) := \bigcup \DDD_s$. Finally, we set $\beta''(t_1) := V(G)$.
Recall that $t_1$ is the root of $L$. Let $L'$ be the tree obtained
from $L$ by adding the new vertices as above. 

This completes the construction. By
construction, if $s, s' \in V(L')$ are independent in $L'$, then $\beta''(s)
\cap \beta''(s') = \emptyset$ and if $s$ is an ancestor of $s'$ then $\beta''(s') \subseteq \beta''(s)$.

Finally, by construction, if $e = (s, t) \in E(L')$, then $\beta''(t)$ is
the union of strong components of $G - \omega''(e)$.

Now, for each $t \in V(L')$ we set  $\beta'''(t) := \beta''(t) \setminus \bigcup \{ \beta''(s) \sth
(t, s) \in E(L') \}$. Then for each vertex $u \in V(G)$ there is exactly
one node $t \in V(L')$ such that $u \in \beta'''(t)$. Therefore, $\WWW := (L',
\beta''', \omega'')$ is a directed tree-decomposition of $G$ of edge
width $\leq k^2 + 2k$.

Now let $\tau$ be the injective function that associates with any tangle $T_j \in
\TTT$ the node $\tau(T_j) := t_j \in V(L')$ and let $\gamma''$ be the function that
maps every edge $e \in E(L') \cap E(L)$ to the
separation $\gamma''(e)$. Then, if $T_j$ and $T_i$ are tangles then
for each edge $e$ on the path between $t_j$ and $t_i$ which minimises $|\gamma''(e)|$
among all edges of this path the separation $\gamma''(e)$ is a minimum order
distinguisher between $T_j$ and $T_i$.

Thus, $(L', \beta''', \gamma'', \tau, \omega'')$ is a directed  tree-decomposition
extending $\LLL$. This completes the proof of~\cref{thm:tangle-decomp}.





\section{Algorithmic Aspect}\label{sec:alaspect}

Finally, we give a polynomial time algorithm for the canonical decomposition in Theorem \ref{thm:tangle-decomp} for fixed $k$. All the proofs so far only need
to find separations of order at most $k-1$. This can be easily implemented by the standard min-cut and max-flow algorithms.
But all the proofs so far assume that we are given all distinguishable brambles of order at least $k$. Thus it remains to find all such distinguishable brambles in polynomial time.

To this end, we use the result by Reed \cite{Reed99}. We need the following definition.
\begin{definition}[well-linked set]
  \label{def:well_linked}
  A \emph{well-linked set} of order $m$ in a digraph $D$ is a set of vertices $W \subseteq V(D)$ with $|W|=m$ such that for all subsets $A,B \subseteq W$ with $|A|=|B|$ there are $|A|$ disjoint paths from $A$ to $B$ in $D-(W\setminus (A \cup B))$.
\end{definition}

Reed shows that any minimum cover for a bramble of order $k$ is well-linked. We now detect all covers of brambles. For this we just guess all $k$ vertices that yield a well-linked set. This can be clearly done in $O((kn)^k)$ time.

Given two brambles $\BBB_1, \BBB_2$ of order at least $k$, if they are distinguishable, there are elements $C_1 \in \BBB_1$ and $C_2 \in \BBB_2$ such that
there is a separation $(A, B)$ of order at most $m-1$ with $C_1 \subseteq A\setminus B$ and $C_2 \subseteq B\setminus A$.
For any minimum cover $Z_1$ ($Z_2$, esp.) of $\BBB_1$ ($\BBB_2$, resp.), because $Z_1$ ($Z_2$, resp.) is well-linked,
$|B \cap Z_1| \leq k-1$ ($|A \cap Z_2| \leq k-1$, resp.), thus $(A \cap B) \cup (B \cap Z_1) \cup (A \cap Z_2)$ is also a separation with $Z_1$ in one side and $Z_2$ in the other side, and its order is at most $3k-3$. It follows that if two brambles $\BBB_1, \BBB_2$ of order at least $k$ are distinguishable, there is a separation  of order at most $3k-3$ such that their minimum covers $Z_1, Z_2$ are separated by this separation.

The above argument costs a factor 3 error. But it is still true that if there are two brambles $\BBB_1, \BBB_2$ of order at least $3k$, and they
are separated by at most $k-1$ vertices, then they are distinguishable. For our algorithmic purpose this is enough.
So in order to detect all such distinguishable brambles of order at least $3k$, we just need to test whether or not two well-linked sets $W_1, W_2$ of order at least $3k$
are separated by at most $k-1$ vertices. Given all well-linked sets of order at least $3k$, this can be clearly done by the standard min-cut and max-flow algorithms. Hence we can detect all distinguishable brambles of order at least $3k$ via their minimum covers (but the separations for brambles are of order at most $k$). If we need to detect all distinguishable tangles, then for each minimum cover of distinguishable brambles of order at least $3k$, we just detect all (weak) separations $(A, B)$ of order at most $k-1$, and 
this yields a corresponding tangle of order $k$, where the ``big'' side contains the minimum cover. 

All the proofs given above to find separations to distinguish two brambles can go through by using their minimum covers (to find separations). 
Thus we can obtain a polynomial time algorithm ($O((kn^{k})$ time) for the canonical decomposition in Theorem \ref{thm:tangle-decomp} for fixed $k$ (with tangles of order at least $3k$).

\section{Routing Through a Wall}
\label{route}

In this section, we show that if a digraph contains a large
cylindrical wall which cannot be separated from the
terminals by low order separations, then we can construct a half-integral solution in polynomial time. 
Using this, in Section~\ref{leafbag}, 
we prove that if a given graph has no two distinguishable tangles of certain high order (with respect to $k$), 
then \textsc{$k$-Half-Or-No-Integral Disjoint Paths} can be solved in polynomial time. This will correspond to the algorithm on leaf bags of the canonical tree decomposition, 
when we design a dynamic programming algorithm for the problem.

More precisely, we prove the
following theorem.

\begin{theorem}
  \label{main1}
  Let $G$ be a digraph, $k\ge 3$ be an integer, and $S := \{s_1, \dots, s_k\}$, $T := \{t_1, \dots, t_k\} \subseteq V(G)$ be sets of size $k$ 
  such that $G$ contains a cylindrical wall $W$ of order $m=k(6k^2+2k+3)$.
  Then in time $\OOO(n^c)$, for some constant $c$ independent of $G, S$, and $T$, one can output either
    \begin{enumerate}[nosep]
  \item a separation $(B_1\rightarrow A_1)$ of order less than 
    $k$ such that $A_1$ contains $S$ and $B_1$ contains a subwall of $W$ of order at least $m-2k$, 
  \item a separation $(B_2\rightarrow A_2)$ of order less than $k$ 
  such that $B_2$ contains $T$ and $A_2$ contains a subwall of $W$ of order at least $m-2k$, or    
   \item a set of paths $P_1, \dots, P_k$ in $G$ such that
  $P_i$ links $s_i$ to $t_i$ for $i\in [k]$, and
  each vertex in $G$ is used by at most two of
  these paths.
  \end{enumerate}
\end{theorem}




Assume that $(C_1, \dots, C_{m}, P^{1}_1, P^{2}_{1},\dots,
    P^{1}_{m},P^{2}_{m})$ be the given set of nested cycles and horizontal paths in a wall.
    For each $i\in [m]$, we call $P^1_i\cup P^2_i$ a \emph{bidirected} horizontal path.
	For two vertex sets $A$ and $B$ in a digraph $G$, 
	a \emph{$(A,B)$-linkage} is a set of pairwise vertex-disjoint paths from $A$ to $B$.

We need the following lemma.
\begin{lemma}\label{subwallreroute}
	Let $k\ge 3$ be an integer.
    Given a cylindrical wall $W$ of order $3k$ with the tuple 
    of nested cycles and horizontal paths $(C_1, \dots, C_{3k}, P^{1}_1, P^{2}_{1},\dots,
    P^{1}_{3k},P^{2}_{3k})$ and
    a set of $2k$ distinct vertices $s_1,\dots, s_k, t_1, \dots, t_k$ 
    such that 
    \begin{itemize}
    \item each $s_i$ is the starting vertex of some path $P^a_b$, and each $t_i$ is the last vertex of some path $P^{a'}_{b'}$,
    \end{itemize}
	one can find in polynomial time a set of $k$ paths $Q_1, \dots, Q_k$ in $W$ such that $Q_i$
    links $s_i$ to $t_i$, for all $i \in [k]$, and
    each vertex of $W$
    is used in at most two of these paths.
\end{lemma}

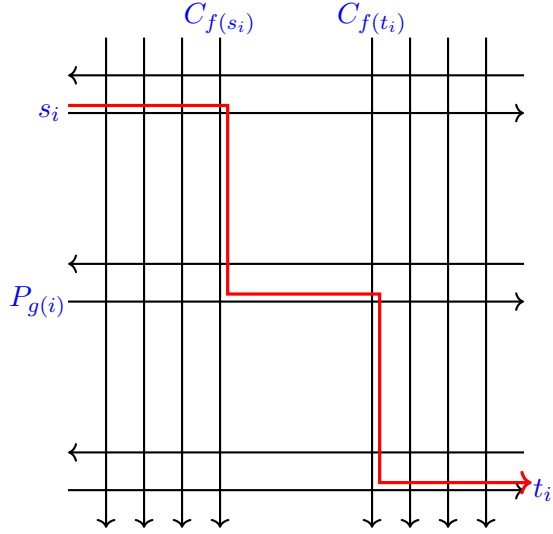
\begin{figure}
  \begin{center}
    \beginpgfgraphicnamed{fig-constructionqi}\begin{tikzpicture}[scale=0.5]
      \tikzstyle{vertex}=[shape=circle, fill=black, draw, inner
      sep=.4mm]
      \tikzstyle{emptyvertex}=[shape=circle, fill=white,
      draw, inner sep=.7mm]

      \foreach \x in {1,3.5, 6} 
    {
    \draw[->, thick] (0,2*\x)--(12,2*\x);
    \draw[->, thick] (12,2*\x+1)--(0,2*\x+1);
	}  

      \foreach \x in {1,2,3,4, 8, 9, 10, 11} 
    {
    \draw[->, thick] (\x,14)--(\x,1);
		}  

    \draw[->, very thick, red] (0,12.2)--(4.2,12.2)--(4.2, 7.2)--(8.2, 7.2)--(8.2, 2.2)--(12.2, 2.2);

     \node at (-.5, 12) {$s_i$};
     \node at (4, 14.5) {$C_{f(s_i)}$};
     \node at (-.8, 7) {$P_{g(i)}$};
     \node at (8, 14.5) {$C_{f(t_i)}$};
     \node at (12.5, 2) {$t_i$};

    \end{tikzpicture}\endpgfgraphicnamed \addtolength{\textfloatsep}{-200pt}
    \caption{Construction of $Q_i$ in Lemma~\ref{subwallreroute}.}\label{fig:constructionqi}
  \end{center}
\end{figure}

 \begin{proof}
        For a path $P\in \{P^{1}_1, P^{2}_{1},\dots,
    P^{1}_{3k},P^{2}_{3k}\}$ and cycles $C_{i},C_{j}\in \{C_1, \dots, C_{3k}\}$ 
    we denote by $P[C_{i},C_{j}]$ the subpath of $P$ between the cycles $C_{i}$ and $C_{j}$. 
    Similarly, for a cycle $C\in \{C_1, \dots, C_{3k}\}$ and paths $P_{i},P_{j}\in \{P^{1}_1, P^{2}_{1},\dots,
    P^{1}_{3k},P^{2}_{3k}\}$ we denote by $C[P_{i},P_{j}]$ the subpath of the cycle starting at $P_{i}$ and ending at $P_{j}$.
    Remark that $C[P_i, P_j]$ and $C[P_j, P_i]$ are distinct.
    
    We will concentrate on constructing, for each $i\in [k]$, a path $Q_i$ linking $s_i$ and $t_i$.
    We depict the construction of $Q_i$ in Figure~\ref{fig:constructionqi}.
   
       Let $S:=\{s_1, \ldots, s_k\}$, $T:=\{t_1, \ldots, t_k\}$, $m_1=| (S\cup T)\cap V(C_1)|$ and $m_2= | (S\cup T)\cap V(C_{3k})|$. 
       Let 
    $f_{1}$ and $f_{2}$ be two bijections $f_1: (S\cup T)\cap V(C_1) \rightarrow [m_1]$ and
    $f_2: (S\cup T)\cap V(C_{3k}) \rightarrow [m_2]$.
    We want to use these bijections to assign nested cycles in $\{C_1, C_2, \ldots, C_{3k}\}$ so that the construction $Q_i$ starting at $s_i$ or ending at $t_i$ uses this cycle to go to another row.
     Let $f:(S\cup T)\rightarrow [3k]$ with
    $$f(x)=
    \begin{cases}
    1+f_{1}(x) & \text{if } x\in (S\cup T)\cap V(C_{1})\\
    3k-f_{2}(x) & \text{if } x\in (S\cup T)\cap V(C_{2})
    \end{cases}.
    $$
    Since  $1+m_{1}<3k-m_{2}$ ($m_{1}+m_{2}=2k$), $f$ is an injective function.
    For every $x\in S\cup T$, let $P_{x}$ be the path in $\{P^{1}_1, P^{2}_{1},\dots,
    P^{1}_{3k},P^{2}_{3k}\}$ that contains $x$, and
    for each $i\in [k]$, we define that 
    $$P^{out}_{i} =
    \begin{cases}
    P_{s_{i}}[C_{1},C_{f(s_{i})}] & \text{when } s_{i}\in V(C_{1})\\
    P_{s_{i}}[C_{f(s_{i})},C_{3k}] & \text{when } s_{i}\in V(C_{3k}) 
    \end{cases}
    $$
    and
    $$P^{in}_{i} =
    \begin{cases}
    P_{t_{i}}[C_{1},C_{f(t_{i})}] & \text{when } t_{i}\in V(C_{1})\\
    P_{t_{i}}[C_{3k},C_{f(t_{i})}] & \text{when } t_{i}\in V(C_{3k}) 
    \end{cases}.
    $$

    Note that since there are $3k$ rows in $(P^{1}_1, P^1_2, \dots, P^{1}_{3k})$
    and $3k$ rows in $(P^{2}_1, P^2_2, \dots, P^{2}_{3k})$,
    there are at least $k$ rows in each of those sets that do not contain any vertex from $S\cup T$.
    Let $Z_{1}$ be a set of $k$ rows from $(P^{1}_1, P^1_2, \dots, P^{1}_{3k})$
    and $Z_{2}$ be a set of $k$ rows from $(P^{2}_1, P^2_2, \dots, P^{2}_{3k})$ such that none of them contain a vertex from $S\cup T$.
    We assign an injection $g:[k]\rightarrow Z_{1}\cup Z_{2}$ so that $g(i)\in Z_{1}$ if $f(s_{i})<f(t_{i})$ and $g(i)\in Z_{2}$ otherwise. 
    Briefly speaking, this function assigns a row where the path $P^{out}_i$ starting at $s_i$ and the path $P^{in}_i$ ending at $t_i$ join.
    From now on, for reasons of uniformity, we will
    abuse notation and denote the path $g(i)$ as $P_{g_i}$.

   We are now ready to construct the paths $Q_{i}$, for $i\in [k]$.
    We define that $$Q_{i}=P^{out}_{i}\cup C_{f(s_{i})}[P_{s_{i}},P_{g_{i}}] \cup P_{g_i}[C_{f(s_{i})},C_{f(t_{i})}]\cup C_{f(t_{i})}[P_{g_i},P_{t_i}]\cup P^{in}_{i}.$$

%
    Then, for all $ i \in [k]$,  $Q_i$ is a path from $s_i$ to $t_i$.
    Let $\mathcal{Q}:=\{Q_1, Q_2, \ldots, Q_k\}$.
    We claim that every vertex in $W$ is contained in 
    at most two paths in $\mathcal{Q}$.
   Note that, by construction, all paths $P^{in}_{i}$, $P^{out}_{i}$, and $P_{g_{i}}[C_{f(s_{i})},C_{f(t_{i})}]$, $i\in [k]$ are disjoint. Moreover, all paths $C_{f(s_{i})}[P_{s_{i}},P_{g_{i}}]$ and $C_{f(t_{i})}[P_{g_{i}},P_{t_i}]$,
   $i\in [k]$ are disjoint.
   Hence, any vertex that belongs to more that one path has to belong to some intersection of a path from $P^{in}_{i}$, $P^{out}_{i}$, and $P_{g_{i}}[f(s_{i}),f(t_{i})]$, $i\in [k]$, and a path from 
   $C_{f(s_{j})}[P_{s_{j}},P_{g_{j}}]$ and $C_{f(t_{j})}[P_{g_j},P_{t_j}]$, $j\in [k]$.
   This implies that every vertex in $W$ is contained in at most two paths in $\mathcal{Q}$ and furthermore, vertices in $s_1, \ldots, s_k, t_1, \ldots, t_k$ are used once.
   Clearly, $\mathcal{Q}$ can be constructed in polynomial time.   
 \end{proof}

	\begin{lemma}\label{lem:linkageinwall1}
	Let $k,w\ge 3$ be integers with $w\ge 2k(k+2)$.
	Let $W$ be a cylindrical wall of order $w$ with the tuple of nested cycles and horizontal paths $(C_1, \dots, C_{m}, P^{1}_1, P^{2}_{1},\dots,
    P^{1}_{w},P^{2}_{w})$.
    Let $A$ be a set of $k$ vertices that are contained in distinct nested cycles, 
    and $B\subseteq V(P^1_1)$ be a set of $k$ vertices that are also contained in distinct nested cycles.
    Then there is a set of $k$ vertex-disjoint paths from $A$ to $B$ in $W$.
	\end{lemma}
	\begin{proof}
	Let $A:=\{s_1, s_2, \ldots, s_k\}$ and $B:=\{t_1, t_2, \ldots, t_k\}$.
	Since $B$ is a subset of $V(P^1_1)$, there are at most $k+1$ horizontal paths containing a vertex of $A\cup B$.
	Since $w\ge 2k(k+2)\ge (k+2)(2k-1)+(k+1)+1$, there are $2k$ consecutive horizontal paths that do not contain vertices of $A\cup B$.
	Let $Q_1, Q_2, \ldots, Q_{2k}$ be those horizontal paths.
	We define $P^{in}_i$ as the path from $s_i$ to $Q_1$ along the nested cycle containing $s_i$, unless it meets a vertex of $B$.
	Let $x$ be the number of paths $P^{in}_i$.
	Similarly we define $P^{out}_i$ as the path from $Q_{2k}$ to $t_i$ along the nested cycle containing $t_i$, unless it meets a vertex of $B$.
	Let $y$ be the number of paths $P^{out}_i$.
	It is easy to observe that $x=y$, and there is a set of $(k-x)$ disjoint paths
	which consists of paths from $s_i$ to a vertex of $B$ before hitting $Q_i$.
	Between those $x$ points in $Q_1$ and $y$ points in $Q_{2k}$, 
	we can find a set of $k$ vertex-disjoint paths between them on the subwall induced by $Q_1, Q_2, \ldots, Q_{2k}$ and the corresponding parts of nested cycles.
	This implies that there is a set of $k$ vertex-disjoint paths from $A$ to $B$ in $W$.
	\end{proof}
	
	\begin{lemma}\label{lem:linkageinwall2}
	Let $k,w\ge 3$ be integers with $w\ge 2k(k+2)$.
	Let $W$ be a cylindrical wall of order $w$ with the tuple of nested cycles and horizontal paths $(C_1, \dots, C_{w}, P^{1}_1, P^{2}_{1},\dots,
    P^{1}_{w},P^{2}_{w})$.
    Let $A$ be a set of $2k+1$ vertices that are contained in distinct bidirected horizontal paths, 
    and $B\subseteq V(P^1_1)$ be another set of $k$ vertices that are also contained in distinct nested cycles.
    Then there is a set of $k$ vertex-disjoint paths from $A$ to $B$ in $W$.
	\end{lemma}
	\begin{proof}
	Let $A'\subseteq A$ be the set of vertices that are not in $P^1_1\cup P^2_1$.
	Let $Q_1, Q_2, \ldots, Q_k$ be the set of nested cycles containing vertices of $B$.
	We choose a set $T$ of $k$ distinct bidirected horizontal paths containing a vertex of $A$, 
	so that no two bidirected horizontal paths in $T$ are not consecutive.
	As $|A'|=2k$, we may choose such $k$ bidirected horizontal paths. Let $A''\subseteq A'$ be the subset that is contained in $T$.
	
	We may choose $k$ disjoint paths $R_1, R_2, \ldots, R_k$ from $A''$ to disjoints nested cycles of $Q_1, Q_2, \ldots, Q_k$ using $2k$ bidirected horizontal paths (with nested cycles between them) so that 
	there are $k$ disjoint paths from the end points of $R_1, R_2, \ldots, R_k$ to $B$ using $Q_1, Q_2, \ldots, Q_k$ without hitting the internal vertices of other paths in $R_1, R_2, \ldots, R_k$. It follows the lemma.
	\end{proof}

\begin{proof}[Proof of Theorem~\ref{main1}]	
     We will prove the theorem by making use of Lemma~\ref{subwallreroute}. Thus, we begin by showing how to reduce our instance to one that satisfies the conditions of Lemma~\ref{subwallreroute}.
    Let $(C_1, \dots, C_{m}, P^{1}_1, P^{2}_{1},\dots,
    P^{1}_{m},P^{2}_{m})$ be the nested cycles and horizontal paths
    constituting the cylindrical wall $W$.
    For each $i\in [m-1]$, let $EC_i$ be the subgraph of the wall induced by the union of $V(C_i)$ and the set of internal vertices of every path $P^{1}_{j}[C_i,C_{i+1}]$ and $P^{2}_{j}[C_{i+1},C_i]$, $j\in [m]$, and let $EC_m=C_m$.
    For convenience, we say that $EC_i$ is an \emph{extended column} of $W$. 
    Observe that $\{V(EC_i)| i\in [m]\}$ is a partition of $V(W)$.    
    Let $\mathcal{R}$ denote the set of all extended columns in $W$.

	\begin{claim}\label{claim:pathS}
	For an integer $t\le \frac{m}{k}$, 
we can find in polynomial time either 
a separation in $G$ described in (1), or
a sequence $(\mathcal{P}_1, \mathcal{R}_1), \ldots, (\mathcal{P}_{t}, \mathcal{R}_{t})$ where
\begin{enumerate}
\item $\mathcal{R}_i$ is a set of $k$ extended columns of $W$, and for $i\neq j$, $\mathcal{R}_i$ and $\mathcal{R}_j$ are disjoint,
\item $\mathcal{P}_i$ is an $(S,X_{i})$-linkage for some set $X_{i}$ of $k$ vertices in pairwise distinct extended columns of~$\mathcal{R}_i$,
\item for every extended column $H \notin \bigcup_{i\in [t]} \mathcal{R}_i$, none of the paths in $\bigcup_{i\in [t]}\mathcal{P}_i$ meet $H$.  
\end{enumerate} 	
\end{claim}
	\begin{ClaimProof}
	Suppose we have such a sequence $(\mathcal{P}_1, \mathcal{R}_1), \ldots, (\mathcal{P}_{t-1}, \mathcal{R}_{t-1})$.
	As $m\ge tk$, there are at least $k$ extended columns of $W$ that are not contained in $\bigcup_{i\in [t-1]} \mathcal{R}_i$.
	We choose such a set $\mathcal{S}$ of $k$ extended columns in $\mathcal{R} \setminus \bigcup_{i\in [t-1]} \mathcal{R}_i$, 
	and choose the subset $X\subseteq V(P^1_1)$ of $k$ vertices that are contained in distinct extended columns of $\mathcal{S}$ and furthermore contained in the nested cycles.
	By Menger's theorem, one can find in polynomial time either a separation $(B\rightarrow A)$ of order less than $k$ with $S\subseteq A$ and $X\subseteq B$, 
	or an $(S,X)$-linkage in $G$.

	Suppose first that we find a former separation $(B\rightarrow A)$.
	By Lemma~\ref{lem:linkageinwall1}, $A$ does not contain $k$ vertices in distinct nested cycles, 
	and by Lemma~\ref{lem:linkageinwall2}, $A$ does not contain $2k+1$ vertices in distinct bidirected horizontal paths. 
	It implies that $B$ contains a subwall of order at least $m-2k$, fulfilling the requirements of (1).
	
	
	Suppose now that there is an $(S,X)$-linkage $\mathcal{P}=\{P_{1}',\dots,P_{k}'\}$ in $G$. We are going to construct appropriate sets $\mathcal{P}_{t}$, $\mathcal{R}_{t}$, and $X_{t}$
	that satisfy the three conditions above. 
	Let $\mathcal{P}_{t}=\{P_{1},\dots,P_{k}\}$ be a $(S, X_t)$-linkage for some set $X_t$ of $k$ vertices in pairwise distinct columns of $\mathcal{R} \setminus \bigcup_{i\in [t-1]} \mathcal{R}_i$
	such that $P_{i}$ is a subpath of $P_{i}'$, $i\in [k]$, and
	 $\sum_{i\in [k]}|V(P_{i})|$ is minimal.
	We claim that $\mathcal{P}_{t}$, and the endpoints $X_{t}$ of the paths in $\mathcal{P}_{t}$ together with the set of the extended columns they belong to, say $\mathcal{R}_{t}$,
	satisfy all three conditions of the claim. Notice that the first two conditions hold immediately by construction.
	Thus, in what remains, we show that for every extended column $H \notin \bigcup_{i\in [t]} \mathcal{R}_i$, none of the paths in $\bigcup_{i\in [t]}\mathcal{P}_i$ meet $H$.

	Towards a contradiction assume that one of the paths in $\bigcup_{i\in [t]}\mathcal{P}_i$ meets some $H \in \mathcal{R} \setminus \bigcup_{i\in [t]} \mathcal{R}_i$.
	By induction, all paths in $\bigcup_{i\in [t-1]}\mathcal{P}_i$ do not meet any extended column $H \in \mathcal{R} \setminus \bigcup_{i\in [t-1]} \mathcal{R}_i$ and thus they also do not meet any extended column in 
	$H \in \mathcal{R} \setminus \bigcup_{i\in [t]}  \mathcal{R}_i$. Therefore, some path $P\in \mathcal{P}_{t}$ meets some $H\in \mathcal{R} \setminus \bigcup_{i\in [t]}  \mathcal{R}_i$ and let $x\in V(P)\cap V(H)$.
	Then we may replace $P$ by the subpath of $P$ from $s$ to $x$, where $s$ is the start vertex of $P$ and obtain a new $(S,Y)$-linkage $\mathcal{P}'$ such that all $k$ end vertices of the paths are in pairwise distinct columns of 
	$\mathcal{R} \setminus \bigcup_{i\in [t-1]} \mathcal{R}_i$, $P_{i}'$ is a subpath of $P_{i}$, $i\in [k]$, and thus $\sum_{P\in \mathcal{P}'}|V(P)| < \sum_{i\in [k]}|V(P_{i})|$, a contradiction.
%
%
%
%
%
%
%
	\end{ClaimProof}
	
	In the same way, we can show the following.
	\begin{claim}\label{claim:pathT}
	For an integer $t\le \frac{m}{k}$,	
	we can find in polynomial time either 
	a separation in $G$ described in (2), or
a sequence $(\mathcal{Q}_1, \mathcal{U}_1), \ldots, (\mathcal{Q}_{t}, \mathcal{U}_{t})$ where
\begin{enumerate}
\item $\mathcal{Q}_i$ is a set of $k$ extended columns of $W$, and for $i\neq j$, $\mathcal{Q}_i$ and $\mathcal{Q}_j$ is disjoint,
\item $\mathcal{U}_i$ is an $(Y,T)$-linkage for some set $Y$ of $k$ vertices in pairwise distinct columns of $\mathcal{U}_i$,
\item for every extended column $H \notin \bigcup_{1\le i\le t} \mathcal{U}_i$, none of the paths in $\bigcup_{1\le i\le t}\mathcal{Q}_i$ meet $H$.  
\end{enumerate} 	
	\end{claim}

	By applying Claims~\ref{claim:pathS} and \ref{claim:pathT} with $t=k$, 
	we obtain sequences 
	\[(\mathcal{P}_1, \mathcal{R}_1), \ldots, (\mathcal{P}_{k}, \mathcal{R}_{k})\]
	and \[(\mathcal{Q}_1, \mathcal{U}_1), \ldots, (\mathcal{Q}_{k}, \mathcal{U}_{k}).\]
	Note that there are at most $2k^{2}$ columns of $W$ contained in $\bigcup_{i\in [k]}\mathcal{R}_{i}\cup\bigcup_{j\in [k]}\mathcal{U}_{j}$.
	Since $m= k(6k^2+2k+3)=2k^2+(2k^2+1)3k$, there is a set of $3k$ consecutive extended columns $\mathcal{M}=\{EC_{i+1}, EC_{i+2}, \ldots, EC_{i+3k}\}$ for some $i\in [0,k(6k^2+2k+3)-3k]$
	such that none of them is an extended column contained   
	in $\mathcal{R}_1\cup \cdots \cup \mathcal{R}_k \cup \mathcal{U}_1 \cup \cdots \cup \mathcal{U}_k$.
	
	Let $H=EC_{i+1}\cup EC_{i+2}\cup \cdots \cup EC_{i+3k-1} \cup C_{i+3k}$ (for the last one, we only add the corresponding nested cycle).
	We choose disjoint sets $X$, $Y\subseteq V(C_{i+1})\cup V(C_{i+3k})$ such that 
	\begin{itemize}
	\item $X$ consists of $k$ vertices in $C_{i+1}$ and $k$ vertices in $C_{i+3k}$
	where each $x\in X\cap V(C_{i+1})$ is a nail belonging to some path $P^{1}_{j}$, for some $j\in [m]$,
	and each $x\in X\cap V(C_{i+3k})$ is a nail belonging to some path $P^{2}_{j}$, for some $j\in [m]$,
	\item $Y$ consists of $k$ vertices in $C_{i+1}$ and $k$ vertices in $C_{i+3k}$, 
	where each $y\in Y\cap V(C_{i+1})$ is a nail belonging to some path $P^{2}_{j}$, for some $j\in [m]$,
	and each $y\in Y\cap V(C_{i+3k})$ is a nail belonging to some path $P^{1}_{j}$, for some $j\in [m]$.
	\end{itemize}
	Note that in the graph $W-(V(H)\setminus X)$, for each extended column $C$ not in $\mathcal{M}$, 
	there are $k$ vertex-disjoint paths from $C$ to $X$, because $X$ contains $k$ vertices in $C_{i+1}$ and also $k$ vertices in $C_{i+3k}$. Similarly, 
	in $W-(V(H)\setminus Y)$,  
	there are $k$ vertex-disjoint paths from $Y$ to $C$.
		
	Now, we consider the graph $G-(V(H)\setminus X)$, and 
	claim that there is a $(S,X)$-linkage in $G-(V(H)\setminus X)$.
	Suppose for contradiction that there is no such a linkage. Then by Menger's Theorem, 
	there is a separation $(B\rightarrow A)$ of order at most $k-1$ in $G-(V(H)\setminus X)$ such that $S\subseteq A$ and $X\subseteq B$.
	
	Observe that all of $(\mathcal{P}_1, \mathcal{R}_1), \ldots, (\mathcal{P}_{k}, \mathcal{R}_{k})$ are contained in $G-(V(H)\setminus X)$.
	Since $\mathcal{R}_1, \ldots, \mathcal{R}_k$ are pairwise disjoint, there exists $\mathcal{R}_j\in \{\mathcal{R}_1, \ldots, \mathcal{R}_k\}$ that does not contain a vertex in $A\cap B$. 
	Suppose all of extended columns in $\mathcal{R}_j$ are contained in $B\setminus A$.
	Then $A\cap B$ intersects every $(S,X_{j})$-linkage, a contradiction to the fact that $\mathcal{P}_{j}$ is an $(S,X_{j})$-linkage of order $k$.
	Therefore, there is an extended column in $\mathcal{R}_j$ that is contained in $A\setminus B$.
	In this case, there are $k$ vertex-disjoint paths from this column to $X$ in  $G-(V(H)\setminus X)$, again a contradiction.
	We conclude that there is an $(S,X)$-linkage in $G-(V(H)\setminus X)$.
	Let $X'=\{x_1, \ldots, x_k\}$ be the set of end vertices of paths in the linkage that are linked from $s_1, \ldots, s_k$, respectively.
		
	Similarly, there is a $(Y, T)$-linkage in $G-(V(H)\setminus Y)$.
	Let $Y'=\{y_1, \ldots, y_k\}$ be the start vertices of the paths in the linkage that are linked to $t_1, \ldots, t_k$, respectively.
	
	Then using Lemma~\ref{subwallreroute}, 
	we can construct $k$ paths $Q_1, \ldots, Q_k$ in $H$ such that each $Q_i$
   	links $x_i$ to $y_i$, for all $i \in [k]$, and
    	each vertex of $H$ 
	is used in at most
    	two of these paths.
        Together with the $(S,X')$-linkage and the $(Y',T)$-linkage, we obtain
	paths $P_1, \dots, P_k$ in $G$ such that
	$P_i$ links $s_i$ to $t_i$ for $i\in [k]$ and each vertex in $G$ is used in at most two of these paths.    
\end{proof}

\section{Graphs with no two distinguishable tangles}
\label{leafbag}


	In this section, we prove the following.
	\begin{theorem}\label{notwotangles}
	Let $k, w\ge 2$ be integers with $w\ge k(k+1)$.
	 There is a function $f_{nt}:(\mathbb{N}\setminus \{1\}) \times (\mathbb{N}\setminus \{1\})\rightarrow \mathbb{N}$ such that 
	given a graph $G$ with no two directed walls of order $w$ separated by a separation of order less than $k(k+1)/2$, 
	and vertices $s_1, \dots, s_k, t_1, \dots, t_k$
	one can in time $n^{f_{nt}(k, w)}$ either 
	\begin{itemize}
	\item determine that there is no set of pairwise vertex-disjoint paths $P_1, \ldots, P_k$ in $G$ such that $P_i$ connects $s_i$ to $t_i$, or	
	\item find paths $P_1, \dots, P_k$ in $G$ such that
  $P_i$ links $s_i$ to $t_i$ for $i=1,\dots,k$ and moreover
  each vertex in $G$ is used in at most two of
  these paths (that is, outputs a half-integral solution). 
  \end{itemize}
  \end{theorem}
  This will correspond to base cases on leaf bags of a canonical tree decomposition.
  Let $h:\mathbb{N}\rightarrow \mathbb{N}$ be the function coming from the result of Kawarabayashi and Kreutzer~\cite{stephan} that
  every digraph of directed tree-width more than $h(w)$ contains a wall of order $h(w)$.

	As a base algorithm, 
	we use the algorithm for \textsc{Disjoint Paths} problem on graphs of bounded directed tree-width due to Johnson et al.~\cite{JRST}.

	\begin{theorem}[Johnson et al.~\cite{JRST}]\label{thm:bddtw}
	Let $k, w$ be positive integers.
	 There is a function $f_{tw}$ such that 
	given a graph $G$ and its directed tree-decomposition of width at most $w$,  
	and vertices $s_1, \dots, s_k, t_1, \dots, t_k$
	one can solve \textsc{Disjoint Paths} problem
	in time $n^{f_{tw}(k+w)}$.
	\end{theorem}


	We will prove by induction on $k$.
	Let $S:=\{s_1, \ldots, s_k\}$ and $T:=\{t_1, \ldots, t_k\}$.
	If $G$ contains no directed wall of order $w+k(k+1)$, then $G$ has directed tree-width at most $h(2w)$, and the result follows from Theorem~\ref{thm:bddtw}.
	Thus, we may assume that $G$ contains a directed wall of order $w+k(k+1)$.
	Intuitively, if there is no small separation from $S$ to $W$, and also no small separation from $W$ to $T$ (in the sense of 1 and 2 in Theorem~\ref{main1}), 
	then using Theorem~\ref{main1}, we can produce a half-integral solution.
	In case when there is a small separation for one direction, we can reduce to two subproblems,
	where one part is a digraph of bounded directed tree-width, 
	and the other part has smaller sets of terminals (for integral solution), so that we can apply the induction hypothesis.

	We introduce a notion of a \emph{pattern graph} that makes easy to present a pattern that an integral solution crosses a small separation.
	For $X\in \{(L\rightarrow R), (R\rightarrow L)\}$, a pattern graph of type $(X, k,t)$ is a graph $G=G(L,M,R)$ on three disjoint vertex sets $L\cup M\cup R$ satisfying that 
	\begin{itemize}
	\item $|M| \le t$, 
	\item $G$ is the disjoint union of $k$ paths, 
	\item $G[L]$ and $G[R]$ have no edges except isolated edges,
	\item there is no path starting at $R$ if $X=(R\rightarrow L)$, and there is no path starting at $L$ otherwise,
	\item there is no edge from $L$ to $R$ if $X=(R\rightarrow L)$, and there is no edge from $R$ to $L$ otherwise.
	\end{itemize}
	We first observe that a pattern graph has bounded size.
	\begin{lemma}\label{lem:patternterminal}
	Let $G(L,M,R)$ be a pattern graph of type $(X, k,t)$.
	If $X=(R\rightarrow L)$, then we have $|L|\le 4t+2k$ and $|R|\le t$.
	If $X=(L\rightarrow R)$, then we have $|L|\le t$ and $|R|\le 4t+2k$.
	\end{lemma}
	\begin{proof}
	We prove for $X=(R\rightarrow L)$.
	There are at most $t$ paths in $G[L\cup M]$ containing a vertex of $M$. 
	For such a path $P$, $P$ contains at most $|V(P)\cap M|+1$ vertices, 
	since it contains no edge fully contained in $L$ by the given assumption.
	Thus, in total, such paths may contain at most $2t$ vertices of $L$.
	On the other hand, note that there are at most $t+k$ possible isolated vertices or isolated edges in $L$ (there may be at most $t$ starting points given by edges from $R$ to $L$).  
	Thus, $L$ contains at most $2t+(2t+2k)=4t+2k$ vertices.
	For $R$, since there is no path starting at $R$, $R$ contains at most $|M|$ vertices.
	
	A symmetric argument holds when $X=(L\rightarrow R)$.
	\end{proof}	
	
	Suppose we have a separation $(B\rightarrow A)$ of order $t<k$ such that 
	$S\subseteq A$ and $B$ contains most of the wall. 
	We guess a pattern graph $H$ of type $((R\rightarrow L), k,t)$ where $k$ is the number of given source and terminal pairs.
 	We furthermore guess new terminal pairs in $A$ and $B$ corresponding to $H$.
	We explain this procedure in the proof more formally.
	Note that for a fixed pattern graph $H$, it is sufficient to guess at most $4t+2k$ vertices in $A\setminus B$, 
	at most $t$ vertices in $B\setminus A$ by Lemma~\ref{lem:patternterminal} and 
	a bijection from $A\cap B$ to $M$.
	To decide whether $G$ contains a linkage from $S$ to $T$, 
	we can ask whether each of $A$ and $B$ has an integral solution corresponding to the pattern graph.
	
	If for every guessed pattern graph $H$ and guessed new terminals, 
	one of $A$ and $B$ does not have the corresponding integral solution, 
	then we will show that $G$ contains no integral solution.
	Otherwise, the algorithm of Theorem~\ref{thm:bddtw} produces an integral solution ($G[A]$ has bounded tree-width), 
	and by induction on the number of terminals, 
	we may obtain a half-integral solution in $G[B]$.
	Combining them, we can output a half-integral solution, as required.
	We formally prove below.

\medskip

\begin{proof}[Proof of Theorem~\ref{notwotangles}]
	Let $f(2, w):=3$ for all $w\in \mathbb{N}$, 
	and for $k\ge 3$, let 
	\[f(k, w)=f(k-1, w)+f_{tw}(h(2w))+ (f_{tw}(h(w))+6k) +4w+8k+5.\]

	We prove by induction on $k$. 
   When $k=2$, two paths $P_1$ and $P_2$ linking pairs $(s_1, t_1)$ and $(s_2, t_2)$ provide a half-integral solution
   unless one of the pairs is not connected.
  	Thus, we may assume that $k\ge 3$.

	Let $S:=\{s_1, \ldots, s_k\}$ and $T:=\{t_1, \ldots, t_k\}$.
	If $G$ contains no directed wall of order $w+k(k+1)$, then $G$ has directed tree-width at most $h(2w)$, and the result follows from Theorem~\ref{thm:bddtw}.
	Thus, we may assume that $G$ contains a directed wall of order $w+k(k+1)$, denoted by $W$.

	By making use of Thoerem~\ref{main1}, 
	we first test whether there are separations described in (1) and (2) of Theorem~\ref{main1}. 
	We guess a set $\mathcal{C}$ of at most $k-1$ nested cycles and a set of $\mathcal{R}$ of at most $2k$ bidirected horizontal paths.
	Let $F$ be the subwall of $W$ that contains all nested cycles not in $\mathcal{C}$ and all bidirected horizontal paths not in $\mathcal{R}$.

	We test whether 
	there is a separation $(B\rightarrow A)$ of order less than $k$ such that 
	$S\subseteq A$ and $V(F)\subseteq B$, and
	also test whether
	there is a separation $(A\rightarrow B)$ of order less than $k$ such that 
	$V(F)\subseteq A$ and $T\subseteq B$.
	These can be tested in polynomial time by Menger's theorem.
	If there are no such separations, 
	then we can construct in polynomial time a half-integral solution by using Theorem~\ref{main1}.
	
	Thus, we may assume that one of separations exists for some guessed sets $\mathcal{C}$ and $\mathcal{R}$.
	By symmetry, we assume that there exists a separation $(B\rightarrow A)$ with $S\subseteq A$ and $V(F)\subseteq B$; 
	in the other case,  we can prove by a symmetric argument.
	Note that $|A\cap B|<k$ and $B\setminus A$ contains a directed wall of order at least $w+k(k+1)-2k\ge w+(k-1)k$.

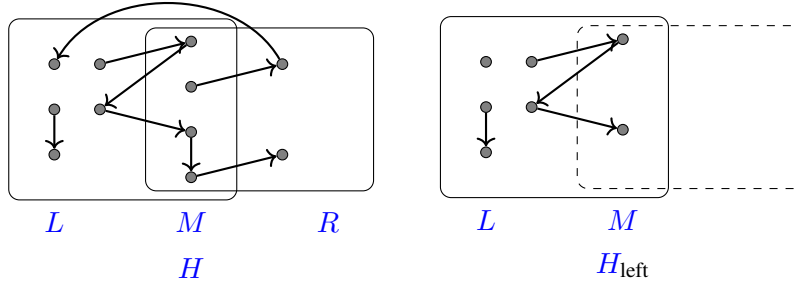
\begin{figure}
  \centering
 \beginpgfgraphicnamed{fig-pattern-a} \begin{tikzpicture}[scale=0.6]
  \tikzstyle{w}=[circle,draw,fill=black!50,inner sep=0pt,minimum width=4pt]

   \node at (-3, -2) {$L$};
   \node at (0, -2) {$M$};
   \node at (0, -3) {$H$};
   \node at (3, -2) {$R$};

      \draw (-3,-.5) node [w] (a3) {};
      \draw (-3,1-.5) node [w] (a2) {};
      \draw (-3,2-.5) node [w] (a1) {};

      \draw (-2,1-.5) node [w] (b2) {};
      \draw (-2,2-.5) node [w] (b1) {};
      
      \draw (0,-1) node [w] (c4) {};
      \draw (0,0) node [w] (c3) {};
      \draw (0,1) node [w] (c2) {};
      \draw (0,2) node [w] (c1) {};
      
      \draw (2,0-.5) node [w] (d3) {};
      \draw (2,2-.5) node [w] (d1) {};

	 \draw[->, thick] (a2)--(a3);
	 \draw[->, thick] (b1)--(c1);
	 \draw[->, thick] (c1)--(b2);
	 \draw[->, thick] (c2)--(d1);
	 \draw[->, thick] (b2)--(c3);
	 \draw[->, thick] (c3)--(c4);
	 \draw[->, thick] (c4)--(d3);
    \draw(d1) [->, thick, in=60,out=120] to (a1);	 
	 
\draw[rounded corners] (-4, -1)--(-4,2.5)--(1,2.5)--(1,-1.5)--(-4,-1.5)--(-4, -1);
\draw[rounded corners] (4, -1)--(4,2.3)--(-1,2.3)--(-1,-1.3)--(4,-1.3)--(4, -1);

   \end{tikzpicture}\endpgfgraphicnamed  \qquad
    \beginpgfgraphicnamed{fig-pattern-b} \begin{tikzpicture}[scale=0.6]
  \tikzstyle{w}=[circle,draw,fill=black!50,inner sep=0pt,minimum width=4pt]

   \node at (-3, -2) {$L$};
   \node at (0, -2) {$M$};
   \node at (0, -3) {$H_{\text{left}}$};

      \draw (-3,-.5) node [w] (a3) {};
      \draw (-3,1-.5) node [w] (a2) {};
      \draw (-3,2-.5) node [w] (a1) {};

      \draw (-2,1-.5) node [w] (b2) {};
      \draw (-2,2-.5) node [w] (b1) {};
      
      \draw (0,0) node [w] (c3) {};
      \draw (0,2) node [w] (c1) {};

	 \draw[->, thick] (a2)--(a3);
	 \draw[->, thick] (b1)--(c1);
	 \draw[->, thick] (c1)--(b2);
	 \draw[->, thick] (b2)--(c3);
	 
\draw[rounded corners] (-4, -1)--(-4,2.5)--(1,2.5)--(1,-1.5)--(-4,-1.5)--(-4, -1);
\draw[dashed, rounded corners] (4, -1)--(4,2.3)--(-1,2.3)--(-1,-1.3)--(4,-1.3)--(4, -1);

   \end{tikzpicture}\endpgfgraphicnamed   \caption{The pattern graph $H$ and its restriction $H_{\text{left}}$ in Theorem~\ref{notwotangles}. 
   We consider the image of the end points of remaining paths in $H_{\text{left}}$ as new terminal pairs for \textsc{Disjoint Paths} problem on $G[A]$.
   }\label{fig:pattern}
\end{figure}

	Let $t:=|A\cap B|$. To guess a possible linkage in $G$ from $S$ to $T$, 
	we guess a pattern graph $H=H(L,M,R)$ of type $((R\rightarrow L), k, t)$.
	We also guess an injective function $g:V(H)\rightarrow V(G)$ satisfying the following conditions:
	\begin{enumerate}
	\item $g(L)\subseteq A\setminus B$, $g(M)=A\cap B$, and $g(R)\subseteq B\setminus A$, 
	\item for every edge $uv$ in $H[M]$, $g(u)g(v)$ is an edge in $G[A\cap B]$, 
	\item for every edge $uv$ in $H$ with $u\in R$ and $v\in L$, $g(u)g(v)$ is an edge in $G$.
	\end{enumerate}
	By Lemma~\ref{lem:patternterminal}, there are at most $n^{(4t+2k)+t+t}=n^{6t+2k}$ possible functions $g$, and 
	we can check whether a function $g$ satisfies the three conditions in polynomial time.

	Since the edges of $G[A\cap B]$ mapped from $H[M]$ are fixed, 
	when we create instances of \textsc{Disjoint Paths} problem on $G[A]$ or $G[B]$, 
	we remove these edges and also remove some vertices where we do not want to use them as an internal vertex.
	We will do this below.

	Let $H_{\text{left}}$ be the subgraph of $H[L\cup M]$ obtained by 
	first removing all edges in $H[M]$, and then removing all isolated vertices in $H[M]$. 
	Let $M_{\text{left}}$ be the set of removed vertices from $M$, when defining $H_{\text{left}}$.
	See Figure~\ref{fig:pattern} for an illustration.
	Similarly, let $H_{\text{right}}$ be the subgraph of $H[M\cup R]$ obtained by first removing all edges in $H[M]$, 
	and then removing all isolated vertices in $H[M]$.
	 Let $M_{\text{right}}$ be the set of removed vertices from $M$, when defining $H_{\text{right}}$.
	Clearly, $H_{\text{left}}$ is also the disjoint union of paths, and no one contains some edges in $H[M]$ as an internal edge.
	The same property holds for $H_{\text{right}}$.
	Let $(a_1, b_1), \ldots, (a_{k'}, b_{k'})$ be the set of endpoints of paths in $H_{\text{left}}$, 
	and let $(c_1, d_1), \ldots, (c_{k''}, d_{k''})$ be the set of endpoints of paths in $H_{\text{right}}$.
	
	We check the following:
	\begin{itemize}
	\item[(1)] we check whether there are disjoint paths each linking $g(a_i)$ to $g(b_i)$, for $i\in [k']$, in $G[A]-g(M_{\text{left}})$, and 
	\item[(2)] we check whether there are disjoint paths each linking $g(c_i)$ to $g(d_i)$, for $i\in [k'']$, in $G[B]-g(M_{\text{right}})$.
	\end{itemize}
	Note that $k'\le 4t+2k$.

	Since $G$ has no two directed walls of order $w$ separated by a separation of order less than $k$, 
	$G[A]$ contains no directed wall of order $w$, 
	and thus $G[A]$ has directed tree-width at most $h(w)$.
	Therefore, one can use the algorithm in Theorem~\ref{thm:bddtw} 
	to check in time $n^{f_{tw}(h(w)+k')}$
	whether there are disjoint paths each linking $g(a_i)$ to $g(b_i)$ in $G[A]-g(M_{\text{left}})$, 
	and if so, produce such a set of disjoint paths.
	
	For (2), we aim to apply induction hypothesis.
	First note that since all vertices of $S$ are contained in $A$ and there are no edges from $A\setminus B$ to $B\setminus A$, 
	the new starting points of paths are all contained in $A\cap B$, 
	and thus, we have $k''\le |A\cap B|\le k-1$.
	We claim that $G[B]$ has no two directed walls of order $w$ that is separated by a separation of order less than $\frac{k''(k''+1)}{2}$.
	Suppose there are two such walls  
	separated by a separation $(D\rightarrow C)$ of $G[B]$ of order less than $\frac{k''(k''+1)}{2}$.
	Then it is easy to see that $(D\rightarrow A\cup C)$ is a separation in $G$ of order less than 
	\[\frac{k''(k''+1)}{2}+(k-1)\le \frac{(k-1)k}{2}+(k-1)\le \frac{k(k+1)}{2}\] that distinguishes two directed walls of order $w$, 
	which contradicts our given assumption.
	Thus, $G[B]$ has no two distinguishable walls of order $w$.
	
	Thus, by induction hypothesis, 
	one can in time $n^{f(k-1, w)}$ either 
	\begin{itemize}
	\item determine that there is no set of pairwise vertex-disjoint paths $Q_1, \ldots, Q_{k''}$ in $G$ such that $Q_i$ connects $c_i$ to $d_i$, or	
	\item find paths $Q_1, \dots, Q_{k''}$ in $G$ such that
  $Q_i$ links $c_i$ to $d_i$ for $i=1,\dots,k''$ and moreover
  each vertex in $G[B]$ is used in at most two of
  these paths. 
  \end{itemize}
  In the former case, we can say that there is no integral solution having the guessed pattern graph as a crossing pattern.
  Suppose the latter case holds.
	In this case, we take the union of obtained paths for $G[A]-M_{\text{left}}$, 
	obtained paths for $G[B]-M_{\text{right}}$, 
	guessed edges in $G[A\cap B]$ and edges from $B\setminus A$ to $A\setminus B$.
	We claim that every vertex is used in at most twice.  	
	If a vertex in $M$ is used as an internal vertex in the pattern graph $H$, 
	then it is used in only one of $G[A]$ or $G[B]$, 
	and thus, it is used in at most twice in the final paths.
	Assume that a vertex $v$ in $M$ is used in both $G[A]$ and $G[B]$.
	In this case, 	$v$ is a new terminal in both $G[A]$ and $G[B]$.
	Since the obtained paths in $G[A]$ form an integral solution, 
	a path in $G[A]$ and a path in $G[B]$ that consider $v$ as an end point
	becomes one path using $v$ as an internal vertex, 
	and at most one another path in $G[B]$ may go through $v$.
	Therefore, every vertex of $M$ is used twice, and we conclude that every vertex of $G$ is used at most twice, as required. 
	
	The total running time is (we take bound ${2w \choose 2k}\le n^{2w}$ for convenience)
	\[\begin{array}{rl}
	& n^{f_{tw}(h(2w))}\times {w+k(k+1) \choose 2k}{w+k(k+1)\choose k-1}n^3 \times n^{6t+2k}n^2 \\
	& \times n^{(f_{tw}(h(w))+4t+2k)}\times n^{f(k-1, w)}\le n^{f(k,w)}. 
	\end{array}\]
  \end{proof}

  \section{Dynamic programming algorithm}\label{sec:DP}

Now, we constitute a dynamic programming algorithm on the tangle tree of a digraph, developed in Section~\ref{sec:alaspect}.
In the application, we do not need a directed tree-decomposition that distinguishes all maximal tangles (or brambles); we need a directed tree-decomposition that distinguishes all tangles of order at least $f(k)$ for some function $f$, where $k$ is the number of a given set of sources and terminals. 

Let $M$ be the maximum order of a tangle in a given graph $G$. We will use such a directed tree-decomposition with $m=f(k)$ if $f(k)\le M$.
If $f(k)>M$, then the directed tree-width of the given graph has bounded directed tree-width with respect to $k$, so we can just use the algorithm by Johnson et al.(Theorem~\ref{thm:bddtw}) for disjoint paths.
Thus, we may assume that $f(k)\le M$.

	Let $f_{tw}$ and $f_{nt}$ be the functions described in Theorems~\ref{thm:bddtw} and \ref{notwotangles}, respectively.

As written in  Section~\ref{sec:alaspect}, each leaf bag contains a tangle of order at least $3m$, but there are no two tangles separated by a separation of order less than $m$.
Furthermore, observe that in any solution to \textsc{$k$-Disjoint Paths}, its restriction on the leaf bag produces a set of $a \le \frac{3m(3m+1)}{2}$ paths (as $|\Gamma(t)| \leq \frac{3m(3m+1)}{2}$ for every internal node $t$); so we guess all possible $a$ source and terminal pairs and solve
\textsc{$a$-Half-Or-No-Integral Disjoint Paths} in time $n^{f_{nt}(a,m)}$ by Theorem~\ref{notwotangles}.

We now consider an internal node $t$ with children $t_1$ and $t_2$ in the decomposition tree. 
For each $i\in [2]$, let $A_i$ be the set of all bags $\beta(x)$ where $x$ is a descendant of $t_i$ in the tangle tree.
By inductive argument, for any guessed set of $b\le \frac{3m(3m+1)}{2}$ pairs of terminals, 
we can solve 
\textsc{$b$-Half-Or-No-Integral Disjoint Paths} in time $n^{g(k)}$ for some function $g$.
Then we do the following.
\begin{itemize}
\item For each $i\in [2]$, we choose a set of terminals $((s^i_{1}, t^i_{1}), (s^i_{2}, t^i_{2}), \ldots, (s^i_{m_i}, t^i_{m_i}))$
	where $m_i\le \frac{3m(3m+1)}{2}$.
\item By induction hypothesis, we can solve \textsc{$m_i$-Half-Or-No-Integral Disjoint Paths} on $A_i$ with  the guessed set of terminals $((s^i_{1}, t^i_{1}), (s^i_{2}, t^i_{2}), \ldots, (s^i_{m_i}, t^i_{m_i}))$. If it returns that there is no integral solution, then it means that there is no integral solution to the whole graph respecting this guessed set of terminals in $A_i$.
	Thus, we go through another sets of terminals. 
	We may assume that the algorithm outputs a half-integral solution in each $A_i$. 
	\item Now, we replace $A_i$ with $\{s^i_{1}, t^i_{1}, s^i_{2}, t^i_{2}, \ldots, s^i_{m_i}, t^i_{m_i}\}$
	with matching edges $(s^i_{i_j}, t^i_{i_j})$.
	Furthermore, 
	among all edges between $A_i$ and $W\setminus A_i$, 
	we remain for each $j$, all edges having $s^i_j$ as a head, and 
	all edges having $t^i_j$ as a tail.
	As $|\beta(t)|\le \frac{3m(3m+1)}{2}$, the resulting graph has at most $\frac{f(k)(f(k)+1)}{2}+4k^2$ vertices. 
	Thus, by brute-force, we may check in polynomial time whether there is an integral solution respecting the guessed sets of terminals or not.
	If yes, then by expanding to the half-integral solution in each $A_i$, we get a half-integral solution in the whole of $A$.
	
\end{itemize}

\begin{theorem}\label{thm:main-alg1-dup}
	For every fixed $k\geq 1$ there is a polynomial-time algorithm for
	\textsc{$k$-Half-Or-No-Integral Disjoint Paths.}%
\end{theorem}


\bibliographystyle{plain}










\end{document}